\documentclass[11pt]{amsart}
\usepackage{mystyle}
   \numberwithin{equation}{section}
\begin{document}
\title{Volume-dependent Field Theories}
\author[Richard Wedeen]{Richard Wedeen}
\address[Richard Wedeen]{Ph.D. from Department of Mathematics,
   University of Texas at Austin}
\email{rwedeen.math@gmail.com}

\begin{abstract}
We develop the axiom system in \cite{KS} to define \emph{volume-dependent field theories (VFTs)}, a class of non-topological quantum field theories whose dependence on the background metric factors through the associated density. We construct a well-defined Lorentzian limit of a Wick-rotated VFT defined on smooth, possibly degenerate Lorentzian bordisms with incoming and outgoing boundary both nonempty. When the VFT is reflection positive, this extends the main theorem in \cite{KS}.
\end{abstract}
\maketitle

\tableofcontents
\section{Introduction}
The axiomatization of topological field theories (\cite{Atiyah:1988aa}, following \cite{Segal1988}) has allowed the development of a fruitful mathematical enterprise that has in turn shed light back onto physics.  It is natural to wonder what a general mathematical framework for non-topological field theories would add to our understanding of the world.

A non-topological quantum field theory depends on a background Lorentzian metric of space-time. To study such theories, \cite{KS} proposed \emph{Wick-rotating} to the larger space of \emph{allowable complex metrics}. 
These metrics form a convex open subset of the space of non-degenerate complex symmetric 2-tensors on space-time. The space of allowable metrics contains the space of Riemannian metrics and its \emph{Shilov boundary}, the holomorphic generalization of the set of extremal points for a complex domain, contains the space of Lorentzian metrics. Wick-rotated quantum field theory can then be defined as a functor out of a bordism category of \emph{manifold germs} equipped with a background allowable complex metric to a suitable category of topological vector spaces such that the vector spaces and linear maps assigned by the functor depend \emph{holomorphically} on the infinite dimensional complex manifold of allowable metrics. One can then attempt to recover the Lorentzian theory by evaluating on manifolds whose allowable metric approaches the Shilov boundary.

In two dimensions, an allowable metric is equivalent to a pair of complex structures together with a complex density. A topological theory has no dependence on the metric and a conformal theory is one that depends only on the pair of complex structures. In this paper, we study field theories which depend solely on the density.

Given any manifold $M$, there is a holomorphic map
\begin{equation}\label{eq:met to dens}
\sqrt{\det}: \Met_\C(M) \raw \Dens_\C(M)
\end{equation}
from the complex manifold of allowable complex metrics on $M$ to the complex manifold of \emph{allowable densities} on $M$ which extends to a map on Shilov boundaries
\begin{equation}
\Met_{\Lor}(M) \raw \Dens_{i\R}(M)
\end{equation}
from possibly degenerate Lorentzian metrics to purely imaginary densities. We define a \emph{volume-dependent field theory (VFT)} to be a field theory whose dependence on the allowable metric factors through (\ref{eq:met to dens}). Equivalently, it is a field theory whose background fields are allowable densities and whose dependence on the allowable density is holomorphic.

To construct the Lorentzian limit of a Wick-rotated field theory, one would like to evaluate the functor on a sequence of bordisms whose allowable metric approaches the boundary and take a limit. In general, there is no clear reason why this limit should be independent of the limiting sequence of metrics. 

However, in one dimension the situation is nicer. In this case, (\ref{eq:met to dens}) is a biholomorphic equivalence -- allowable metrics and densities determine the same data on 1-manifolds. 
In particular, a Riemannian metric is equivalent to a real and positive density and Moser's theorem \cite{moser1965volume} becomes relevant:
\begin{theorem}[Moser]
Let $X$ be a compact, connected $n$-manifold and let $\omega_0, \omega_1$ be real and positive densities of the same total volume. Then there exists a diffeomorphism $\varphi \in \Diff(X)$ such that
\begin{equation}
\varphi^*(\omega_1) = \omega_0
\end{equation}
\end{theorem}
Together with diffeomorphism invariance of the functor, this implies that 1-dimensional Euclidean field theories depend only on the total volume. Put differently: given any density $\omega_t$ in a path of positive densities with constant volume, Moser's argument shows that there exists an infinitesimal diffeomorphism which can be identified with a real vector field $\xi_t$ whose action by Lie derivative on $\omega_t$ satisfies 
\begin{equation}\label{eq:lie derivative}
\mathcal L_{\xi_t}\omega_t = \dot\omega_t.
\end{equation}
Diffeomorphism invariance then shows that the value of the functor does not change as the density is varied along the path $\omega_t$.

When the metric is allowable, Moser's argument can be complexified to produce a \emph{complex} vector field $\xi_t$ satisfying (\ref{eq:lie derivative}). The same argument together with holomorphicity of the functor then shows that the field theory on a connected bordism depends only on the \emph{total complex volume} valued in the right half plane $\C_{>0}$ of complex numbers with positive real part. This dependence implies that the functor obeys the semigroup law on $\C_{>0}$ which allows one to show that the Lorentzian limit is independent of the sequence of metrics used to construct it.

If the field theory is \emph{reflection positive}, then one can see this explicitly. An allowable germ of a 0-manifold fixed under co-orientation reversal and complex conjugation is assigned a \emph{Hermitian} object in the target category which defines a Hilbert space $\mathcal H$. 
Composition of bordisms respects the semigroup law on $\C_{>0}$ and, together with holomorphicity, reflection positivity implies that an interval of length $s \in \C_{>0}$ is assigned the trace class operator
\begin{equation}
e^{-sH}: \mathcal H \raw \mathcal H
\end{equation}
where $H$ is a self-adjoint, unbounded operator on $\mathcal H$ with discrete spectrum bounded below.
The Lorentzian limit then assigns to a bordism of imaginary length $it$ the unitary operator
\begin{equation}
e^{-itH}: \mathcal H \raw \mathcal H.
\end{equation}

To apply this argument in higher dimensions, \cite{KS} restricts to \emph{real analytic and globally hyperbolic} bordisms. A bordism $X: \Sigma_0 \rightsquigarrow \Sigma_1$ is globally hyperbolic if every maximally extended time-like geodesic travels from $\Sigma_0$ to $\Sigma_1$. This implies that $X \cong \Sigma_0 \times I$ is diffeomorphic to a cylinder and there is a global time function $\tau: X \raw i\R$ whose level surfaces foliate $X$ by Riemannian manifolds -- the Lorentzian metric on $X$ can be expressed as $g = d\tau^2 + h_\tau$, where $h_\tau$ is a Riemannian metric on the fiber above $\tau \in i\R$. Real analyticity of $X, g$ and $\tau$ implies that $g, \tau$ extend holomorphically to a holomorphic thickening $X_\C$ of $X$. Specifically, $\tau$ extends to a holomorphic bundle 
\begin{equation}
\tau_\C: X_\C \raw U
\end{equation}
where $U \subset \C$ is a neighborhood containing the purely imaginary image of $\tau$. There is a holomorphic trivialization of $\tau_\C$ that extends the diffeomorphism $X \cong \Sigma_0 \times I$. This implies that above every curve $\gamma: I \raw U$ in the base satisfying $\gamma'(t) > 0$, there is a unique totally real allowable bordism $X_\gamma \subset X_\C$ in the preimage $\tau_\C^{-1}(\gamma)$.

Adapting the one-dimensional argument above, \cite{KS} show that the functor applied to $X_\gamma$ depends only on the total length of $\gamma$ in $\C_{>0}$ measured using the allowable metric obtained by restricting the complex symmetric 2-tensor $d\tau^2$ to $\gamma \subset U$. This means that the functor obeys a semigroup law on $\C_{>0}$ which allows one to construct a well-defined Lorentzian limit:

\begin{theorem}[\cite{KS} Theorem 5.2]\label{KS 5.2}
A reflection positive quantum field theory induces a functor 
\begin{equation}
\mathscr L: \Bord_{n,n-1}^\omega(\Met_{g.h.}) \raw \Hilb
\end{equation}
out of the category of real analytic bordisms equipped with a real-analytic globally hyperbolic metric to the category of Hilbert spaces and bounded maps which takes bordisms to unitary operators.
\end{theorem}

The main goal of this paper is to extend this theorem when the theory is volume-dependent. In this case, Moser's argument can be applied in all dimensions to show:

\begin{reptheorem}{thm:vol dependence thm}
A VFT of general dimension depends only on the total complex volume of a given bordism. 
\end{reptheorem}

This implies the key fact that \emph{on cylinders, a VFT obeys the semigroup law on $\C_{>0}$}. This makes possible the construction of a well-defined Lorentzian limit. Furthermore, Moser's argument applies in the smooth setting so one can discard the real analyticity assumption and enlarge the domain to smooth Lorentzian manifolds. In fact, if one allows limits to degenerate Lorentzian metrics in the Shilov boundary, one can go beyond cylinders and construct Lorentzian limits of general bordisms so long as their incoming and outgoing boundaries are both nonempty. However, the operators assigned to such bordisms may no longer be bounded. Instead, they will be morphisms in a certain category of rigged Hilbert spaces which we denote $\mathcal{NP}^h$. We summarize this in the following.
\begin{reptheorem}{thm:LHilb2}
A reflection positive volume-dependent field theory induces a functor
\begin{equation}
\mathscr L : \Bord_{n,n-1}^{in\wedge out}(\Met_{\Lor} ) \raw \mathcal{NP}^h
\end{equation}
out of the category of smooth, possibly degenerate Lorentzian bordisms with nonempty incoming and nonempty outgoing boundary. It sends cylindrical bordisms to unitary operators and general bordisms to unbounded operators of rigged Hilbert spaces.
\end{reptheorem}

We now give an overview of the paper.
Section \ref{section QFT} develops the definition of field theory presented in \cite{KS}. Given a sheaf $\mathscr F$ on manifolds, we define a bordism category $\Bord_{n,n-1}(\mathscr F)$ whose objects are \emph{cylindrical germs} of closed $n-1$-manifolds equipped with a co-orientation and a germ of a section of $\mathscr F$, and whose morphisms are thickened germs of bordisms also equipped with a section of $\mathscr F$. Disjoint union defines a symmetric monoidal product. For the purposes of defining holomorphicity, we restrict ourselves to sheaves of smooth sections of fiber bundles whose fibers are finite dimensional complex manifolds. Allowable metrics and densities are examples of such sheaves.

In Section \ref{section NP}, we define the codomain category $\mathcal{NP}$ of \emph{nuclear pairs} whose objects are continuous, injective maps with dense image
\begin{equation}
\check E \mono \hat E
\end{equation}
from a \emph{nuclear dual Fr\'echet space} $\check E$ to a \emph{nuclear Fr\'echet space $\hat E$}. The symmetric monoidal structure is given by the topological tensor products on the nuclear spaces $\check E$ and $\hat E$. In \cite{KS}, nuclear pairs arise as the direct and inverse limits of a system of Fr\'echet spaces and nuclear maps assigned to a cylindrical bordism by the field theory. In this paper, we take the category of nuclear pairs as fundamental and impose a condition we term \emph{coherence} on the functor to ensure that the nuclear pair that arises from the direct and inverse limits associated to a cylindrical germ is isomorphic to the nuclear pair the functor assigns to it.

In Section \ref{section field theory} we define an \emph{$\mathscr F$-field theory} to be a coherent symmetric monoidal functor
\begin{equation}
Z: \Bord_{n,n-1}(\mathscr F) \raw \mathcal{NP}
\end{equation}
that depends \emph{holomorphically} on the sections of $\mathscr F$. We also require that the functor sends bordisms to morphisms of nuclear pairs consisting of nuclear maps, which we term \emph{nuclearity}. Coherence and nuclearity give way to an action of germs of diffeomorphisms which we elucidate in Theorem \ref{thm:diffeomorphism extension}.

When the sheaf $\mathscr F$ is equipped with an \emph{antiholomorphic involution}, there is a twisted involution $\tau_{\mathscr F}$ on $\Bord_{n,n-1}(\mathscr F)$ obtained by composing the antiholomorphic involution with the twisted involution of reversing co-orientations. Likewise, there is a twisted involution $\tau_{\mathcal{NP}}$ on $\mathcal{NP}$ of complex conjugation and taking strong duals of nuclear spaces. We say $Z$ is \emph{reflection positive} if it is $(\tau_{\mathscr F}, \tau_{\mathcal{NP}})$-equivariant and if the images of fixed points land in the subcategory $\mathcal{NP}^h$ of \emph{Hermitian nuclear pairs} consisting of $\tau_{\mathcal{NP}}$-fixed points which satisfy a positivity condition. This positivity condition gives Hermitian nuclear pairs the structure of a rigged Hilbert space.

In Section \ref{section VFTs}, we define a VFT to be a $\Met_\C$-field theory
\begin{equation}
Z: \Bord_{n,n-1}(\Met_\C)\raw \mathcal{NP}
\end{equation}
whose dependence on the allowable complex metric factors through (\ref{eq:met to dens}). We use Moser's argument to show total volume dependence and use this to prove that every VFT is naturally isomorphic to one that factors through a nuclear functor 
\begin{equation}\label{eq:V}
V: \Bord_{n,n-1}(\C_{>0}) \raw \mathcal{NP}
\end{equation}
out of the category of bordisms whose connected components are labeled by an element of the semigroup $\C_{>0}$ representing the total volume. 

In Section \ref{section lorentzian limit} we construct a Lorentzian limit for general VFTs
\begin{equation}\label{eq:L}
L: \Bord_{n,n-1}^{in \wedge out}(i\R) \raw \mathcal{NP}
\end{equation}
and define the short-distance topological limit as the restriction of this functor to bordisms with 0 volume which in particular is a topological field theory partially defined on a subcategory of the full bordism category.

In Section \ref{reflection positive vfts}, we study reflection positive VFTs. In this case the functor (\ref{eq:V}) can be refined to a functor
\begin{equation}
V^{Hilb}: \Bord_{n,n-1}(\C_{>0}) \raw \Hilb
\end{equation}
where the codomain is now the category of separable Hilbert spaces and bounded maps.
The value of $V^{Hilb}$ on a closed $n-1$ manifold $\Sigma$ is a Hilbert space $\mathcal H_\Sigma$ and evaluating on a cylindrical bordism gives a family of trace-class operators $V^{Hilb}(\Sigma\times I, s) \in \End(\mathcal H_\Sigma)$ depending holomorphically on the total volume $s \in \C_{>0}$.
Reflection positivity and the semigroup law
imply there exists an unbounded operator $H_\Sigma \in \End(\mathcal H_\Sigma)$ with discrete spectrum bounded below such that
\begin{equation}
V^{Hilb}(\Sigma \times s) = \exp(-sH_\Sigma)
\end{equation}
for all $s \in \C_{>0}$.
We then prove Theorem \ref{thm:LHilb2}, where we recover that the Lorentzian limit (\ref{eq:L}) of a reflection positive VFT assigns unitary operators to cylindrical bordisms. Furthermore, to general bordisms it assigns an unbounded operator of rigged Hilbert spaces satisfying a sub-exponential growth condition specified by the eigenvalues of the Hamiltonians $H_\Sigma$.

In Section \ref{section long distance limit} we define the \emph{long-distance topological limit} of a reflection positive VFT. This is a topological field theory obtained from taking the limit as the real part of the total volume goes to positive infinity. Finally, in Section \ref{section 2d classification} we classify all reflection positive VFTs in dimension 2 up to natural isomorphism.

I thank Esteban Cardenas, Dan Freed, Charlie Reid, Alberto San Miguel Malaney, Will Stewart, and Jackson Van Dyke for helpful conversations.

\section{Geometric Axioms for QFT}\label{section QFT}
In this section we define quantum field theory with background fields valued in a sheaf of complex manifolds following \cite{KS}.

\subsection{The Bordism Category of Germs}\label{section bordism category}
Let $\Man_n$ be the site of $n$-manifolds without boundary and embeddings between them and let $\mathscr F: \Man_n^{op} \raw \Set$ be a sheaf.
In this section we define a symmetric monoidal category $\Bord_{n,n-1}(\mathscr F)$ of bordisms between germs of manifolds with background fields valued in $\mathscr F$.

\subsubsection{Preliminaries on Germs}\label{subsection preliminaries on germs}
Let $k\leq n$ and let $Y$ be a compact $k$-manifold possibly with boundary. We will use $[Y]$ to denote the diffeomorphism type of $Y$ as a manifold with boundary.
Let $\mathcal G_{[Y]}$ be the category whose objects $Y \subset N$ consist of $n$-manifolds $N \in \Man_n$ containing a smoothly embedded manifold of diffeomorphism type $[Y]$ and whose morphisms are embeddings
\begin{equation}\label{eq:embedding}
\begin{tikzcd}[row sep=small, column sep = small]
N \arrow[r, hook] & N' \\
Y \arrow[u, phantom, sloped, "\subset"] \arrow[r,"\sim"] & Y' \arrow[u, phantom, sloped, "\subset"]
\end{tikzcd}
\end{equation}
restricting to a diffeomorphism of manifolds with boundary.
Each object $Y \subset N$ determines a subcategory $\mathfrak U_Y^N \subset \mathcal G_{[Y]}$ of open neighborhoods $Y \subset U \subset N$ directed by inclusion. Set 
\begin{equation}
\mathring U^N_Y := \invlim_{U \in \mathfrak U^N_Y} U
\end{equation}
which is an object in the pro-completion $\Pro(\mathcal G_{[Y]})$. We call $\mathring U^N_Y$ a \emph{manifold germ} of $Y$.

The set of morphisms between germs in the pro-completion is by definition
\begin{equation}\label{eq:diff germ}
\Pro(\mathcal G_{[Y]})(\mathring U^N_Y, \mathring U^{N'}_{Y'})
= \invlim_{V \in \mathfrak U^{N'}_{Y'}} \colim_{U \in \mathfrak U^N_Y} \mathcal G_{[Y]}(U,V)
\end{equation}
We set $\mathring{\mathcal G}_{[Y]}$ to be the full subcategory of $\Pro(\mathcal G_{[Y]})$ generated by all manifold germs $\mathring U^N_Y$.

\begin{lemma}
For all $V \in \mathfrak U^{N'}_{Y'}$, the canonical map
\begin{equation}\label{eq:colim bijection}
\mathring{\mathcal G}_{[Y]}(\mathring U^N_Y,\mathring U^{N'}_{Y'}) \raw \colim_{U \in \mathfrak U^N_Y} \mathcal G_{[Y]}(U,V)
\end{equation}
is a bijection of sets.
\end{lemma}
\begin{proof}
Let $V \xhookrightarrow{j} V'$ be a morphism in $\mathfrak U^{N'}_{Y'}$. There is a map of sets
\begin{equation}
j_*: \colim_{U \in \mathfrak U^N_Y} \mathcal G_{[Y]}(U,V) \raw \colim_{U \in \mathfrak U^N_Y} \mathcal G_{[Y]}(U,V')
\end{equation}
induced from post-composing an embedding $f: U \mono V$ with $j$. 
If $\mathring f, \mathring f$ are elements of
\begin{equation}
\colim_{U \in \mathfrak U^N_Y}\mathcal G_{[Y]}(U,V)
\end{equation}
satisfying $j_*\mathring f = j_*\mathring f'$, then there exists $U \in \mathfrak U^N_Y$ and embeddings $f,f': U \mono V$ such that $j\circ f = j \circ f'$. Injectivity of $j$ implies $\mathring f = \mathring f'$ which implies injectivity of $j_*$. 
If $\mathring g$ is a germ of an embedding $g: U \mono V'$, let $\tilde U := g^{-1}(V) \subset U$. Then $g|_{\tilde U}: \tilde U \mono V$ has a germ whose image under $j_*$ is $\mathring g$ which implies surjectivity, and therefore bijectivity, of $j_*$. As each map in the inverse system defining $\mathring{\mathcal G}_{[Y]}(\mathring U^N_Y,\mathring U^{N'}_{Y'})$ is a bijection, this proves the claim.
\end{proof}

Given $U \in \mathfrak U^N_Y, V \in \mathfrak U^{N'}_{Y'}$ there is a composition
\begin{equation}\label{eq:germ rep}
\mathcal G_{[Y]}(U,V) \raw \colim_{\tilde U \in \mathfrak U^N_Y} \mathcal G_{[Y]}(\tilde U,V) \xrightarrow{\sim} \mathring{\mathcal G}_{[Y]}(\mathring U^N_Y, \mathring U^{N'}_{Y'})
\end{equation}
of the canonical map to the colimit and the inverse of (\ref{eq:colim bijection}). We will refer to the image of $f: U \mono V$ under this composition as the \emph{germ} of $f$.

\begin{proposition}\label{prop:diffeo rep}
Let $\mathring f \in \mathring{\mathcal G}_{[Y]}(\mathring U^N_Y, \mathring U^{N'}_{Y'})$. There exists $U \in \mathfrak U^N_Y, V \in \mathfrak U^{N'}_{Y'}$ and a diffeomorphism $f: U \xrightarrow{\sim} V$ in $\mathcal G_{[Y]}(U,V)$ whose germ is $\mathring f$.
\end{proposition}
\begin{proof}
By definition of the colimit and the previous lemma, there exists $U \in \mathfrak U^N_Y, \tilde V \in \mathfrak U^{N'}_{Y'}$, and $f: U \mono \tilde V \in \mathcal G_{[Y]}(U, \tilde V)$ whose germ if $\mathring f$. Then $f$ is a diffeomorphism onto its image $V:= \Ima(f)$ whose germ is $\mathring f$.
\end{proof}

Composition of morphisms in $\mathring{\mathcal G}_{[Y]}$ can be expressed as follows. Let $\mathring f \in \mathring{\mathcal G}_{[Y]}(\mathring U^N_Y, \mathring U^{N'}_{Y'})$ and $\mathring g \in \mathring{\mathcal G}_{[Y]}(\mathring U^{N'}_{Y'}, \mathring U^{N''}_{Y''})$. By Proposition \ref{prop:diffeo rep}, we can choose $U \in \mathfrak U^N_Y, V \in \mathfrak U^{N'}_{Y'}, W \in \mathfrak U^{N''}_{Y''}$ and diffeomorphisms $f: U \xrightarrow{\sim} V$ and $g: V \xrightarrow{\sim} W$ whose germs are $\mathring f, \mathring g$. The composition $\mathring g \circ \mathring f$ is the germ of $g \circ f$.

\begin{corollary}
$\mathring{\mathcal G}_{[Y]}$ is a groupoid.
\end{corollary}
\begin{proof}
Let $\mathring f \in \mathring{\mathcal G}_{[Y]}(\mathring U^N_Y, \mathring U^{N'}_{Y'})$. Proposition \ref{prop:diffeo rep} implies there exists a diffeomorphism $f: U \xrightarrow{\sim} V$ whose germ is $\mathring f$. The inverse $\mathring f^{-1}$ is the germ of $f^{-1}: V \xrightarrow{\sim} U$.
\end{proof}
\begin{remark}
We call an element of (\ref{eq:diff germ}) an \emph{isomorphism of manifold germs}.
\end{remark}
\begin{notation}
Denote $\Diff(\mathring U^N_Y) := \mathring{\mathcal G}_{[Y]}(\mathring U^N_Y,\mathring U^N_Y)$.
\end{notation}

\begin{remark}The sheaf $\mathscr F: \Man_n^{op} \raw \Set$ induces a pre-sheaf $\mathscr F: \mathring{\mathcal G}_{[Y]}^{op} \raw \Set$ which assigns to a manifold germ $\mathring U^N_Y$ the set
\begin{equation}
\mathscr F(\mathring U^N_Y) = \colim_{U \in \mathfrak{U}^N_Y}\mathscr F(U)
\end{equation}
\end{remark}

When $Y:= \Sigma^{n-1}$ is a closed $n-1$-manifold, we set $\mathcal G_{[\Sigma]}' \subset \mathcal G_{[\Sigma]}$ to be the full sub-category of co-orientable inclusions $\Sigma \subset N$ and $\mathring{\mathcal G}_{[\Sigma]}' \subset \mathring{\mathcal G}_{[\Sigma]}$ the corresponding sub-groupoid of co-orientable manifold germs. There is a pre-sheaf $\mathscr P^{c}_{[\Sigma]}: (\mathcal G_{[\Sigma]}')^{op} \raw \Set$ that assigns to a co-orientable inclusion $\Sigma \subset N$ the set of co-orientations which induces a pre-sheaf $\mathring{\mathscr P}^c_{[\Sigma]}: (\mathring{\mathcal G}_{[\Sigma]}')^{op} \raw \Set$ of co-orientations on manifold germs. 

\begin{definition}
Let $\mathring{\mathcal G}^c_{[\Sigma]}$ be the groupoid of pairs $(\mathring U_\Sigma^N, \mathring c)$ of a manifold germ $\mathring U_\Sigma^N \in \mathring{\mathcal G}'_{[\Sigma]}$ and a co-orientation $\mathring c \in \mathring{\mathscr P}^c_{[\Sigma]}(\mathring U_\Sigma^N)$ with morphisms $(\mathring U_\Sigma^N, \mathring c) \rightsquigarrow (\mathring U_{\Sigma'}^{N'}, \mathring c')$ consisting of co-orientation preserving isomorphisms $\mathring f \in \mathring{\mathcal G}_{[\Sigma]}(\mathring U_\Sigma^N, \mathring U_{\Sigma'}^{N'})$ satisfying $\mathring f^*\mathring c' = \mathring c$.
\end{definition}

\begin{remark}
When $Y$ is a $k$-manifold with corners for all $k \leq n$, there will be a corresponding presheaf $\mathscr P^c_{[Y]}: \mathring{\mathcal G}_{[Y]}^{op} \raw \Set$ of co-orientation data needed to glue bordisms between manifolds of higher codimension.
\end{remark}

\begin{definition}\label{def:cylindrical germ}
Let $C_\Sigma := \Sigma \times \R$ be the cylinder and let $\Sigma \subset C_\Sigma$ denote the inclusion of the 0-slice $\Sigma \times \{0\} \subset \Sigma \times \R$. We will call $\mathring \Sigma := \mathring U^{C_\Sigma}_\Sigma \in \mathring{\mathcal G}'_{[\Sigma]}$ the \emph{cylindrical germ of $\Sigma$}. Let $\mathring c_+ \in \mathring{\mathscr P}^c_{[\Sigma]}(\mathring \Sigma)$ be the positive co-orientation.  We will call $\mathring \Sigma^+:= (\mathring \Sigma ,\mathring c_+) \in \mathring{\mathcal G}_{[\Sigma]}^c$ the \emph{positively co-oriented cylindrical germ of $\Sigma$}.
\end{definition}

\begin{proposition}
$\mathring \Sigma^+$ is an initial object in $\mathring{\mathcal G}_{[\Sigma]}^c$.
\end{proposition}

\begin{proof}
Let $(\mathring U^N_\Sigma, \mathring c) \in \mathring{\mathcal G}_{[\Sigma]}^c$.
Choose a Riemannian metric on $N$. The co-orientation $\mathring c$ determines a co-orientation and a trivialization of the normal bundle on $\Sigma \subset N$. For $\varepsilon >0$ small, the exponential map is a co-orientation preserving diffeomorphism
\begin{equation}
\exp_\varepsilon : U_\varepsilon \xrightarrow{\sim} V
\end{equation}
from $U_\varepsilon := \Sigma \times (-\varepsilon, \varepsilon)$ in $\mathfrak U_\Sigma^{C_\Sigma}$ onto a tubular neighborhood $V \in \mathfrak U_\Sigma^N$ which restricts to the identity on $\Sigma \subset C_\Sigma$. The germ of $\exp_\varepsilon \in \mathcal G_{[\Sigma]}(U_\varepsilon,V)$ is an isomorphism from $(\mathring U_\Sigma^{C_\Sigma}, \mathring c_+)$ to $(\mathring U_\Sigma^N, \mathring c)$.
\end{proof}

We now suppose $Y: = X$ is a compact $n$-manifold with boundary. Let $X \subset N$ be an object in $\mathcal G_{[X]}$ and let $p: \partial X \raw \{0,1\}$ be a partition of the boundary. Set $\partial X_i := p^{-1}(i)$. The inclusion $X \subset N$ restricts to inclusions of the boundary components $\partial X_i \subset N$ which determine manifold germs $\mathring U_{\partial X_i}^N$ and we let $\mathring c_i \in \mathring{\mathscr P}^c_{[\partial X_i]}(\mathring U_{\partial X_i}^N)$ be the incoming and outgoing co-orientations for $i=0$ and $i=1$.  

Let $X' \subset N'$ be another object in $\mathcal G_{[X]}$ and $p': \partial X' \raw \{0,1\}$ a partition. An isomorphism $\mathring f \in \mathring{\mathcal G}_{[X]}(\mathring U^N_X, \mathring U^{N'}_{X'})$ restricts to a diffeomorphism $f|_{\partial X}: \partial X \xrightarrow{\sim} \partial X'$ and we say $\mathring f$ \emph{respects partitions} if $(f|_{\partial X})^*p' = p$. Such an isomorphism $\mathring f$ induces isomorphisms $\mathring f_i \in \mathring{\mathcal G}_{[\partial X_i]}(\mathring U^N_{\partial X_i}, \mathring U^{N'}_{\partial X_i'})$.

\begin{definition}
The \emph{groupoid of bordisms of type $[X]$} is the groupoid $\mathring{\mathcal G}^*_{[X]}$ with objects consisting of tuples 
\begin{equation}
\mathcal X := (\mathring U^N_X, p, \mathring\theta_0, \mathring \theta_1)
\end{equation}
where 
\begin{itemize}
\item $\mathring U^N_X \in \mathring{\mathcal G}_{[X]}$ is a manifold germ
\item $p: \partial X \raw \{0,1\}$ is a partition of the boundary
\item $\mathring \theta_i: \mathring \Sigma_i^+ \xrightarrow{\sim} (\mathring U^N_{\partial X_i}, \mathring c_i)$ are co-orientation preserving isomorphisms in $\mathring{\mathcal G}^c_{[\Sigma_i]}$
\end{itemize}
and isomorphisms 
\begin{equation}
(\mathring U^N_X, p, \mathring \theta_0, \mathring \theta_1) \xrightarrow{\sim} (\mathring U^{N'}_{X'}, p', \mathring\theta_0', \mathring\theta_1')
\end{equation}
for each partition-respecting isomorphism $\mathring f \in \mathring{\mathcal G}_{[X]}(\mathring U^N_X, \mathring U^{N'}_{X'})$ satisfying $\mathring f_i \circ \mathring \theta_i = \mathring\theta_i'$. We say $\mathcal X$ is a \emph{bordism from $\mathring \Sigma_0^+$ to $\mathring \Sigma_1^+$} which we denote $\mathcal X: \mathring \Sigma_0^+ \rightsquigarrow \mathring \Sigma_1^+$ and we say $\mathring f$ is a \emph{diffeomorphism of bordisms}.
\end{definition}

\subsubsection{\texorpdfstring{$\mathscr F$}{F}-Bordisms}

Let $\mathcal X := (\mathring U^N_X, p, \mathring\theta_0, \mathring\theta_1)$ be a bordism from $\mathring \Sigma_0^+$ to $\mathring \Sigma_1^+$. For $i=0,1$, there are maps
\begin{equation}\label{eq:r^X_i}
r_i^{\mathcal X}: \mathscr F(\mathring U^N_X) \raw \mathscr F(\mathring U^N_{\partial X_i}) \xrightarrow{\mathring\theta_i^*} \mathscr F(\mathring \Sigma_i)
\end{equation}
which restrict a germ of a background field on $X$ to the boundary components $\partial X_i$ and pull it back to $\mathring \Sigma_i$ along $\mathring \theta_i$.

\begin{definition}
Let $X$ be a compact $n$-manifold with boundary. The \emph{groupoid of $\mathscr F$-bordisms of type $[X]$} is the groupoid $\mathring{\mathcal G}^{\mathscr F}_{[X]}$ with objects consisting of pairs
\begin{equation}
(\mathcal X, \mathring \sigma_X)
\end{equation}
where 
\begin{itemize}
\item $\mathcal X := (\mathring U^N_X, p, \mathring\theta_0, \mathring \theta_1)$ is a bordism of type $[X]$
\item $\mathring \sigma_X \in \mathscr F(\mathring U^N_X)$ is a germ of a section
\end{itemize}
and isomorphisms
\begin{equation}
(\mathcal X, \mathring \sigma_X) \xrightarrow{\sim} (\mathcal X', \mathring \sigma_{X'})
\end{equation}
for each $\mathring f \in \mathring{\mathcal G}^*_{[X]}(\mathcal X, \mathcal X')$ satisfying $\mathring f^* \mathring\sigma_{X'} = \mathring \sigma_X$. We say $(\mathcal X, \mathring \sigma_X)$ is an \emph{$\mathscr F$-bordism from $(\mathring \Sigma_0^+, \mathring \sigma_0)$ to $(\mathring \Sigma_1^+, \mathring \sigma_1)$} where $\mathring\sigma_i := r^{\mathcal X}_i\mathring\sigma_X$. We say $\mathring f$ is a \emph{diffeomorphism of $\mathscr F$-bordisms}.
\end{definition}

\begin{remark}
If $\mathscr F: \Man_n^{op} \raw \Set$ is the trivial sheaf that assigns to every $n$-manifold the singleton set $\{*\}$ then $\mathring{\mathcal G}^{\mathscr F}_{[X]} = \mathring{\mathcal G}^*_{[X]}$.
\end{remark}

\begin{definition}[Composition of $\mathscr F$-bordisms]
Let $\mathring \sigma_i \in \mathscr F(\mathring\Sigma_i)$ for $i=0,1,2$ and let 
\begin{equation}\label{eq:bordisms}
\begin{split}
(\mathcal X, \mathring\sigma_X): (\mathring \Sigma_0^+, \mathring\sigma_0) \rightsquigarrow (\mathring \Sigma_1^+, \mathring\sigma_1) \\
(\mathcal Y, \mathring\sigma_Y): (\mathring \Sigma_1^+, \mathring\sigma_1) \rightsquigarrow (\mathring\Sigma_2^+, \mathring\sigma_2)
\end{split}
\end{equation}
be $\mathscr F$-bordisms where
\begin{equation}
\begin{split}
\mathcal X:= (\mathring U^N_X, p, \mathring\theta_0,\mathring\theta_1): \mathring \Sigma_0^+ \rightsquigarrow \mathring \Sigma_1^+\\
\mathcal Y:= (\mathring U^M_Y, q, \mathring\eta_1, \mathring\eta_2): \mathring \Sigma_1^+ \rightsquigarrow \mathring\Sigma_2^+
\end{split}
\end{equation}
are bordisms. We assume there are sections $\sigma_N \in \mathscr F(N)$ and $\sigma_M \in \mathscr F(M)$ whose restrictions to $\mathring U^N_X$ and $\mathring U^M_Y$ are $\mathring \sigma_X$ and $\mathring\sigma_Y$; if not, we replace $N,M$ by open neighborhoods containing $X,Y$ that admit such sections.

Let
\begin{equation}\label{eq:emb reps}
\begin{aligned}[c]
\theta_0: U_0 \mono N\\
\theta_1: U_1 \mono N
\end{aligned}
\quad\quad
\begin{aligned}[c]
\eta_1: U_1 \mono M\\
\eta_2: U_2 \mono M
\end{aligned}
\end{equation}
be embeddings of neighborhoods $\Sigma_i \subset U_i \subset C_{\Sigma_i}$ restricting to $\mathring \theta_0, \mathring\theta_1, \mathring\eta_1,\mathring\eta_2$. Set 
\begin{equation}
\begin{split}
\tilde X := X \cup \Ima(\theta_1)\\
\tilde Y := \Ima(\eta_1) \cup Y
\end{split}
\end{equation}
and 
\begin{equation}
\begin{split}
V:= \Ima(\theta_0)\cup X \cup \Ima(\theta_1)\\
W:= \Ima(\eta_1) \cup Y \cup \Ima(\eta_2).
\end{split}
\end{equation}
Let
\begin{equation}
L:= V \cup_{U_1} W
\end{equation}
be the smooth manifold obtained from gluing $V,W$ along $U_1$ with $\theta_1, \eta_1$. The sections $\sigma_N, \sigma_M$ restrict to sections $\sigma_V, \sigma_W$ on $V,W$ which agree on the overlap inside $L$. By the sheaf property, they glue uniquely to a section $\sigma_L \in \mathscr F(L)$. The smooth manifold 
\begin{equation}
Y \circ X := \tilde X \cup_{U_1} \tilde Y
\end{equation}
is a submanifold of $L$ with boundary isomorphic to $p^{-1}(0) \sqcup q^{-1}(1)$. The section $\sigma_L$ restricts to a germ $\mathring\sigma_{Y\circ X} \in \mathscr F(\mathring U^L_{Y\circ X})$. Set $r: \partial(Y\circ X) \raw \{0,1\}$ to be the partition that sends $p^{-1}(0)$ to 0 and $q^{-1}(1)$ to 1. Regard $\mathring\theta_0, \mathring\eta_2$ as germs of embeddings into $L$. We say the bordism
\begin{equation}
\mathcal Y \circ \mathcal X := (\mathring U^L_{Y\circ X}, r, \mathring\theta_0, \mathring\eta_2): \mathring\Sigma_0^+ \rightsquigarrow \mathring\Sigma_2^+
\end{equation}
is \emph{a composition of $\mathcal X$ and $\mathcal Y$} and the $\mathscr F$-bordism
\begin{equation}
(\mathcal Y\circ \mathcal X, \mathring\sigma_{Y\circ X}): (\mathring\Sigma_0^+,\mathring\sigma_0) \rightsquigarrow (\mathring\Sigma_2^+,\mathring\sigma_2)
\end{equation}
is \emph{a composition of $(\mathcal X,\mathring\sigma_X)$ and $(\mathcal Y, \mathring\sigma_Y)$}.
\end{definition}

\begin{proposition}
Composition is well-defined on diffeomorphism classes of $\mathscr F$-bordisms. 
\end{proposition}
\begin{proof}
Let $(\mathcal X, \mathring\sigma_X), (\mathcal Y,\mathring\sigma_Y)$ be fixed $\mathscr F$-bordisms with notation as above. We first show that the compositions obtained from different sets of representatives (\ref{eq:emb reps}) are diffeomorphic. Let $\theta_i,\eta_j$ be as above and let
\begin{equation}
\begin{aligned}[c]
\theta_0': U_0' \mono N\\
\theta_1': U_1' \mono N
\end{aligned}
\quad\quad
\begin{aligned}[c]
\eta_1': U_1' \mono M\\
\eta_2': U_2' \mono M
\end{aligned}
\end{equation}
be another set of embeddings restricting to $\mathring \theta_0, \mathring\theta_1,\mathring\eta_1,\mathring\eta_2$. Let $\tilde X, \tilde Y, V, W, L$ be as before, and set
\begin{equation}
\begin{split}
\tilde X':= X \cup \Ima(\theta_1')\\
\tilde Y':= \Ima(\eta_1')\cup Y
\end{split}
\end{equation}
\begin{equation}
\begin{split}
V' := \Ima(\theta_0') \cup X \cup \Ima(\theta_1')\\
W' := \Ima(\eta_1')\cup Y \cup \Ima(\eta_2').
\end{split}
\end{equation}
\begin{equation}
L' := V' \cup_{U_1'} W'
\end{equation}
There is a diffeomorphism
\begin{equation}\label{eq:diffeo1}
\begin{split}
(Y\circ X)'&:= \tilde X' \cup_{U_1'} \tilde Y'\\
 &\cong \tilde X \cup_{U_1} \tilde Y\\
&=: Y\circ X
\end{split}
\end{equation}
that commutes with the embeddings of $X$ and $Y$.
As $\theta_0, \theta_0'$ and $\eta_2,\eta_2'$ define the same germs, they restrict to the same embeddings
\begin{equation}\label{eq:common emb}
\theta_0'': \tilde U_0 \mono N
\quad\text{and}\quad
\eta_2'': \tilde U_2 \mono M
\end{equation}
on small enough open neighborhoods $\Sigma_0 \subset \tilde U_0 \subset U_0\cap U_0'$ and $\Sigma_2 \subset \tilde U_2 \subset U_2\cap U_2'$. Set
\begin{equation}
\begin{split}
\tilde V:= \Ima(\theta_0'')\cup X \cup \Ima(\theta_1)\\
\tilde W:= \Ima(\eta_1)\cup  Y \cup \Ima(\eta_2'')
\end{split}
\quad\text{and}\quad
\begin{split}
\tilde V':= \Ima(\theta_0'')\cup X \cup \Ima(\theta_1')\\
\tilde W':= \Ima(\eta_1')\cup  Y \cup \Ima(\eta_2'')
\end{split}
\end{equation}
The diffeomorphism (\ref{eq:diffeo1}) extends to a diffeomorphism 
\begin{equation}\label{eq:comp diffeo}
\tilde V' \cup_{U_1'} \tilde W' \cong \tilde V \cup_{U_1} \tilde W
\end{equation}
which commutes with the embeddings (\ref{eq:common emb}) and defines a diffeomorphism of bordisms 
\begin{equation}
\mathcal Y \circ \mathcal X \cong \mathcal (Y \circ \mathcal X)'
\end{equation}
where
\begin{equation}
\begin{split}
\mathcal Y \circ \mathcal X := (\mathring U^L_{Y\circ X}, r, \mathring\theta_0,\mathring \eta_2)\\
(\mathcal Y \circ \mathcal X)' := (\mathring U^{L'}_{(Y\circ X)'}, r, \mathring\theta_0, \mathring\eta_2)
\end{split}
\end{equation}
The section $\sigma_{L'}$ pulls back to $\sigma_L$ along (\ref{eq:comp diffeo}) which gives a diffeomorphism 
\begin{equation}
(\mathcal Y \circ \mathcal X, \mathring \sigma_{Y\circ X}) \cong \mathcal ((\mathcal Y \circ \mathcal X)', \mathring \sigma_{(Y\circ X)'})
\end{equation}
of $\mathscr F$-bordisms.

The proof that if
\begin{equation}
(\mathcal X,\mathring\sigma_X) \cong (\mathcal X', \mathring \sigma_{X'}) \quad \text{and} \quad
(\mathcal Y, \mathring\sigma_Y) \cong (\mathcal Y', \mathring \sigma_{Y'})
\end{equation} 
are diffeomorphic $\mathscr F$-bordisms, then
\begin{equation}
(\mathcal Y \circ \mathcal X, \mathring \sigma_{Y\circ X}) \cong (\mathcal Y' \circ \mathcal X', \mathring\sigma_{Y'\circ X'})
\end{equation} 
is similar and we omit it. 
\end{proof}

\begin{definition}
We define the symmetric monoidal bordism semicategory $\Bord_{n,n-1}^{s}(\mathscr F)$ of germs with background fields valued in $\mathscr F$ as follows. An object is a pair $(\mathring \Sigma^+, \mathring \sigma)$ consisting of a positively co-oriented cylindrical germ of a closed $n-1$-manifold $\Sigma$ and a germ of a section $\mathring \sigma \in \mathscr F(\mathring \Sigma)$. Morphisms are $\mathscr F$-bordisms up to diffeomorphism. Composition of morphisms is composition of bordisms up to diffeomorphism. The symmetric monoidal product is given by disjoint union with the empty manifold germ acting as the unit.
\end{definition}

\begin{notation}
Denote $\Diff(\mathring \Sigma^+) := \mathring{\mathcal G}^c_{[\Sigma]}(\mathring\Sigma^+, \mathring\Sigma^+)$.
\end{notation}

\begin{definition}
The bordism category $\Bord_{n,n-1}(\mathscr F)$ is the category containing the bordism semicategory with the addition of isomorphisms $\mathring f: (\mathring \Sigma^+, \mathring \sigma) \xrightarrow{\sim} (\mathring \Sigma^+, \mathring \sigma')$ for each $\mathring f \in \Diff(\mathring \Sigma^+)$ satisfying $\mathring f^*\mathring \sigma' = \mathring \sigma$. 
\end{definition}

\begin{proposition}\label{prop:sheaf induced functor}
A morphism of sheaves $\varphi: \mathscr F \raw \mathscr G$ induces a symmetric monoidal functor
\begin{equation}\label{eq:induced sheaf functor}
\varphi_*: \Bord_{n,n-1}(\mathscr F) \raw \Bord_{n,n-1}(\mathscr G)
\end{equation}
\end{proposition}
\begin{proof}
The sheaf morphism $\varphi$ extends to a morphism of presheaves on $\Pro(\Man_n)$. We use this to define $\varphi_*$ which acts on objects and morphisms by 
\begin{equation}
\begin{split}
(\mathring \Sigma^+, \mathring\sigma) &\mapsto (\mathring \Sigma^+, \varphi(\mathring\sigma))\\
(\mathcal X, \mathring \sigma_X) &\mapsto (\mathcal X, \varphi(\mathring\sigma_X))\\
\mathring f &\mapsto \mathring f
\end{split}
\end{equation}
where the action on diffeomorphisms is justified by functoriality of sheaf morphisms.
\end{proof}


\subsection{The Category of Nuclear Pairs}\label{section NP}
Let $\mathcal{NF}$ and $\mathcal{NDF}$ denote the symmetric monoidal categories of nuclear Fr\'echet and nuclear dual Fr\'echet spaces. For a review of these categories and nuclear spaces in general we refer the reader to Appendix \ref{appendix Nuclear Spaces} which we will peruse throughout this section.
\begin{definition}
   A \emph{nuclear pair} is a continuous dense injection $\check E \xhookrightarrow{\iota_E} \hat E$ with $\check E \in \mathcal{NDF}, \hat E \in \mathcal{NF}$. A morphism of nuclear pairs is a commutative diagram
   \begin{equation}\label{eq:NP Morphism}
   \begin{tikzcd}[row sep=small, column sep = small]
      \check E \arrow[d, "\check f"'] \arrow[r, hook, "\iota_E"] &\hat E \arrow[d, "\hat f"]\\
      \check F \arrow[r,hook, "\iota_F"] &\hat F
   \end{tikzcd}
   \end{equation}
   The category of such pairs will be denoted $\mathcal{NP}$.
\end{definition}

By Proposition \ref{prop:NDF2DF}, the map $\check E \xhookrightarrow{\iota_E} \hat E$ is nuclear.

\begin{definition}\label{def:nuclear}
We will say a morphism of nuclear pairs is \emph{nuclear} if there exists a factorization
\begin{equation}\label{eq:nuclear morphism}
   \begin{tikzcd}
      \check E \arrow[d, "\check f"'] \arrow[r, hook, "\iota_E"] &\hat E \arrow[d, "\hat f"] \arrow[dl]\\
      \check F \arrow[r,hook, "\iota_F"] &\hat F
   \end{tikzcd}
   \end{equation}
through a continuous map $\hat E \raw \check F$.
\end{definition}

\begin{lemma}\label{lem:nuclear implies nuclear}
The maps $\check f$ and $\hat f$ in (\ref{eq:nuclear morphism}) are nuclear.
\end{lemma}
\begin{proof}
Both maps are a composition of a nuclear map with a continuous map. Apply Proposition \ref{prop:nuclear composition}.
\end{proof}

\begin{lemma}
$\mathcal{NP}$ is a symmetric monoidal category.
\end{lemma}

\begin{proof}
The monoidal product is given by
\begin{equation}
[\check E \xhookrightarrow{\iota_E} \hat E] \otimes [\check F \xhookrightarrow{\iota_F} \hat F] := [\check E \widehat\otimes \check F \xhookrightarrow{\iota_E \widehat\otimes \iota_F} \hat E \widehat\otimes \hat F].
\end{equation}
where $\iota_E \widehat\otimes \iota_F$ is the map from Corollary \ref{cor:completed tensor maps} which by Proposition \ref{prop:injective and dense maps} is injective with dense image. We omit the verification of the axioms for a symmetric monoidal category, which are implied by Proposition \ref{prop:N(D)F sym monoidal}.
\end{proof}

\begin{lemma}
The transpose
\begin{equation}
\hat E^* \xrightarrow{\iota_E^*} \check E^*
\end{equation}
is a nuclear pair and there is a canonical isomorphism
\begin{equation}\label{eq:reflexivity}
\begin{tikzcd}
\check E \arrow[d, "\cong"rot] \arrow[r,hook,"\iota_E"] &\hat E \arrow[d, "\cong"rot] \\
(\check E^{*})^* \arrow[r,hook,"(\iota_E^{*})^*"] & (\hat E^{*})^*
\end{tikzcd}
\end{equation}
of nuclear pairs.
\end{lemma}
\begin{proof}
By Proposition \ref{prop: transpose continuous}, the transpose $\hat E^* \xrightarrow{\iota_E^*} \check E^*$ is continuous. By Corollary \ref{cor:NF reflexive} and Proposition \ref{prop:dual NF}, $\check E$ and $\hat E$ are reflexive spaces which implies the isomorphism (\ref{eq:reflexivity}). Together with Proposition \ref{prop: Imu dense u^* iff injective}, we see that $\iota_E$ is injective and dense if and only if $\iota_E^*$ is injective and dense. Finally, Proposition \ref{prop:nuclear transpose} implies that $\hat E^* \xhookrightarrow{\iota_E^*} \check E^*$ is nuclear.
\end{proof}

Let $\delta_{\mathcal{NP}}: \mathcal{NP} \raw \mathcal{NP}^{op}$ be the functor that sends $\check E \xhookrightarrow{\iota_E} \hat E$ to $\hat E^* \xhookrightarrow{\iota_E^*} \check E^*$ and a morphism (\ref{eq:NP Morphism}) to the morphism
\begin{equation}
\begin{tikzcd}[row sep=small, column sep = small]
\hat F^* \arrow[d, "\hat f^*"'] \arrow[r, hook, "\iota_F^*"] &\check F^* \arrow[d, "\check f^*"]\\
\hat E^* \arrow[r,hook, "\iota_E^*"] &\check E^*
\end{tikzcd}
\end{equation}
The isomorphism (\ref{eq:reflexivity}) is functorial -- thus $\delta_{\mathcal{NP}}^2$ is naturally isomorphic to the identity functor which makes $\delta_{\mathcal{NP}}$ a twisted involution. For an overview of involutions on categories we refer the reader to Appendix B of \cite{Freed_2021}.

There is another functor $\alpha_{\mathcal{NP}}: \mathcal{NP} \raw \mathcal{NP}$ which sends a nuclear pair $\check E \xhookrightarrow{\iota} \hat E$ to its complex conjugate
$\overline{\check E} \xhookrightarrow{\bar\iota} \overline{\hat E}$ and a morphism (\ref{eq:NP Morphism}) to the morphism
\begin{equation}
\begin{tikzcd}[row sep=small, column sep = small]
\overline{\check E} \arrow[d, "\overline{\check f}"'] \arrow[r, hook, "\overline{ \iota_E}"] & \overline{\hat E} \arrow[d, "\overline{\hat f}"]\\
\overline{\check F} \arrow[r,hook, "\overline{\iota_F}"] & \overline{\hat F}
\end{tikzcd}
\end{equation}
The canonical natural isomorphism of $\alpha_{\mathcal{NP}}^2$ with the identity functor makes $\alpha_{\mathcal{NP}}$ an involution on $\mathcal{NP}$.

The involutions $\delta_{\mathcal{NP}}$ and $\alpha_{\mathcal{NP}}$ commute which implies the composition 
\begin{equation}
\tau_{\mathcal{NP}} := \delta_{\mathcal{NP}} \circ \alpha_{\mathcal{NP}}
\end{equation}
is a twisted involution on $\mathcal{NP}$. 
A fixed point of $\tau_{\mathcal{NP}}$ is a nuclear pair $\check E \xhookrightarrow{\iota} \hat E$ equipped with an isomorphism
\begin{equation}\label{eq:fixed point}
\begin{tikzcd}
\check E \arrow[d, "\cong"rot, "\check \theta"] \arrow[r, hook, "\iota"]
& \hat E \arrow[d, "\cong"rot, "\hat \theta"]\\
\overline{\hat E}^* \arrow[r, hook, "\overline{\iota}^*"] 
&\overline{\check E}^*
\end{tikzcd}
\end{equation}
which we will denote $\theta$.

Given a fixed point (\ref{eq:fixed point}), Proposition \ref{prop:Hom Isos} implies the composition
\begin{equation}
\hat\theta \circ \iota: \check E \xhookrightarrow{} \overline{\check E}^*
\end{equation}
is equivalent to a continuous map
\begin{equation}
\check E \widehat\otimes \overline{\check E} \raw \C
\end{equation}
and injectivity of $\iota$ implies that the corresponding sesquilinear form
\begin{equation}\label{eq:sesquilinear form}
\langle, \rangle_{\iota,\theta} : \check E \times \overline{\check E} \raw \C
\end{equation}
is non-degenerate.

\begin{definition}\label{def:hermitian np}
A $\tau_{\mathcal{NP}}$-fixed point $(\check E \xhookrightarrow{\iota} \hat E, \theta)$ will be called \emph{Hermitian} if $\langle, \rangle_{\iota,\theta}$ is a positive definite inner product. We will call a nuclear pair $\check E \xhookrightarrow{\iota} \hat E$ \emph{Hermitian} if there exists an isomorphism $\theta$ that makes $(\iota,\theta)$ a Hermitian fixed point.
\end{definition}

The domain $\check E$ of a Hermitian nuclear pair is a pre-Hilbert space. Let $E^{Hilb}$ be its Hilbert space completion. 

\begin{lemma}\label{lem:Hermitian Hilbert}
A Hermitian nuclear pair $\check E \xhookrightarrow{\iota} \hat E$ admits a factorization
\begin{equation}\label{eq:E Hilb factorization}
\check E \mono E^{Hilb} \mono \hat E
\end{equation}
into a composition of nuclear inclusions with dense image.
\end{lemma}
\begin{proof}
The first map is the injection of $\check E$ into its Hilbert space completion and is thus dense. It is continuous because the sesquilinear form (\ref{eq:sesquilinear form}) is jointly continuous. It is therefore a continuous map from a nuclear space to a Banach space and Proposition \ref{prop:nuclear space char} implies that it is nuclear. By Proposition \ref{prop:nuclear transpose}, its conjugate transpose is a nuclear map.  It is equipped with an isomorphism
\begin{equation}
\begin{tikzcd}
\overline{E^{Hilb}}^* \arrow[r,hook] \arrow[d, "\cong"rot, "\theta^{Hilb}", no head] & \overline{\check E}^* \arrow[d, "\cong"rot, "\hat \theta", no head]\\
E^{Hilb} \arrow[r,hook] & \hat E
\end{tikzcd}
\end{equation}
where $\theta^{Hilb}$ is the isomorphism induced by the Hilbert space inner product. The bottom inclusion is the second map in (\ref{eq:E Hilb factorization}). That the composition is $\iota$ is a consequence of the definition of the sesquilinear form (\ref{eq:sesquilinear form}).
\end{proof}
The proof of Lemma \ref{lem:Hermitian Hilbert} implies the isomorphism (\ref{eq:fixed point}) of a Hermitian nuclear pair extends to an isomorphism of Hilbert spaces
\begin{equation}\label{eq:hermitian hilbert iso}
\begin{tikzcd}
\check E \arrow[r,hook] \arrow[d, "\cong"rot, "\check \theta"] &E^{Hilb} \arrow[r,hook, "\iota"] \arrow[d, "\cong"rot, "\theta^{Hilb}", dashed] &\hat E \arrow[d, "\cong"rot, "\hat \theta"]\\
\overline{\hat E}^* \arrow[r,hook, "\bar{\iota}^*"] &\overline{E^{Hilb}}^* \arrow[r,hook] &\overline{\check E}^*
\end{tikzcd}
\end{equation}

Let $\mathcal{NP}^{h}$ denote the full subcategory of $\mathcal{NP}$ whose objects are Hermitian nuclear pairs.

\begin{definition}\label{def:bounded nph morphism}
A morphism in $\mathcal{NP}^h$
\begin{equation}\label{eq:nph diagram}
\begin{tikzcd}
\check E \arrow[r,hook] \arrow[d, "\check f"]& E^{Hilb} \arrow[r,hook] & \hat E  \arrow[d, "\hat f"]\\
\check F \arrow[r,hook]& F^{Hilb} \arrow[r,hook] &\hat E
\end{tikzcd}
\end{equation}
will be called $\emph{bounded (resp. unitary, ...)}$ if there exists a bounded (resp. unitary, ...) map 
\begin{equation}
f^{Hilb}: E^{Hilb} \raw F^{Hilb}
\end{equation}
making the diagram (\ref{eq:nph diagram}) commute. If there is no such bounded map, we will say it is \emph{unbounded with domain $\check E$} and we will say that the unbounded operator $f^{Hilb}$ is \emph{induced} from the morphism (\ref{eq:nph diagram}).
\end{definition}

Let $\mathcal {NP}^{h,nuc} \subset \mathcal{NP}^{h}$ be the subcategory with the same objects but with morphisms restricted to those that are nuclear or the identity.
\begin{lemma}
Every morphism in $\mathcal{NP}^{h,nuc}$ 
\begin{equation}\label{eq:npf morphism}
\begin{tikzcd}
\check E \arrow[r,hook] \arrow[d] &E^{Hilb} \arrow[r,hook, "i"]  &\hat E \arrow[d] \arrow[lld, "f"']\\
\check F \arrow[r,hook, "j"'] & F^{Hilb} \arrow[r,hook] & \hat F
\end{tikzcd}
\end{equation}
is trace-class.
\end{lemma}  
\begin{proof}
By Lemma \ref{lem:Hermitian Hilbert}, the map $f^{Hilb} := j \circ f \circ i$ is a composition of continuous and nuclear maps. Thus Proposition \ref{prop:nuclear composition} implies that $f^{Hilb}$ is a trace-class operator.
\end{proof}

Let $\Hilb$ denote the symmetric monoidal category of separable Hilbert spaces equipped with the Hilbert-Schmidt tensor product. Let
\begin{equation}\label{eq:NP^h --> Hilb}
\mathcal{NP}^{h,nuc} \raw \Hilb
\end{equation}
be the functor that sends a Hermitian nuclear pair (\ref{eq:E Hilb factorization}) to $E^{Hilb}$ and a nuclear morphism (\ref{eq:npf morphism}) to $f^{Hilb}$.
Let $\tau_{\Hilb}: \Hilb^{op} \raw \Hilb$ denote the twisted involution that sends a Hilbert space to its conjugate dual and a bounded linear map to its conjugate transpose.

\begin{proposition}
The functor (\ref{eq:NP^h --> Hilb})
is symmetric monoidal and $(\tau_{\mathcal{NP}}, \tau_{\Hilb})$-equivariant.
\end{proposition}
\begin{proof}
Let $\check E_1 \mono E_1^{Hilb} \mono \hat E_1$ and $\check E_2 \mono E_2^{Hilb} \mono \hat E_2$ be Hermitian nuclear pairs and let $\langle, \rangle_i$ be the Hilbert space inner product on $E_i^{Hilb}$.
Their monoidal product is the nuclear pair
\begin{equation}
 \check E_1 \widehat\otimes \check E_2 \xhookrightarrow{\iota_1\widehat\otimes\iota_2} \hat E_1 \widehat\otimes \hat E_2
\end{equation}
which is equivalent, under the isomorphism (\ref{eq:fixed point}), to a map
\begin{equation}
(\check E_1 \widehat\otimes \check E_2)\widehat\otimes \overline{(\check E_1 \widehat\otimes \check E_2)} \raw \C.
\end{equation}
This defines an inner product
$\langle, \rangle_{\iota_1 \hat\otimes \iota_2}$ on $\check E_1 \widehat\otimes \check E_2$ which satisfies
\begin{equation}
\langle x_1 \otimes x_2, y_1\otimes y_2 \rangle_{\iota_1 \hat\otimes \iota_2} = \langle x_1, y_1 \rangle_{\iota_1} \cdot \langle x_2, y_2 \rangle_{\iota_2}
\end{equation}
for $x_i,y_i \in \check E_i$.
Its Hilbert space completion is $E_1^{Hilb} \widehat\otimes_{HS} E_2^{Hilb}$, where $\widehat\otimes_{HS}$ is the Hilbert-Schmidt tensor product, and thus the functor is symmetric monoidal. The isomorphism (\ref{eq:hermitian hilbert iso}) is functorial and implies equivariance.
\end{proof}


\subsection{Field Theory}\label{section field theory}
In this section we give a functorial definition of field theory with background fields in a sheaf $\mathscr F$ of complex manifolds following \cite{KS}. 
\subsubsection{Sheaves of Complex Manifolds}
Let $\Fib_n$ denote the category of smooth fiber bundles
\begin{equation}
\begin{tikzcd}
E \arrow[d]\\
M
\end{tikzcd}
\end{equation}
on $n$-manifolds $M \in \Man_n$ admitting local trivializations
\begin{equation}
\begin{tikzcd}[row sep = small, column sep = small]
X \times U \arrow[rr, "\sim"] \arrow[rd] & & E|_U \arrow[ld]\\
& U &
\end{tikzcd}
\end{equation}
whose restriction to a fiber $X$ is a biholomorphism of finite dimensional complex manifolds.
We let morphisms consist of maps of smooth bundles covering embeddings of $n$-manifolds
\begin{equation}\label{eq:Fib morphism}
\begin{tikzcd}
E \arrow[r]\arrow[d] & F \arrow[d] \\
M \arrow[r,hook] & N
\end{tikzcd}
\end{equation}
We will say a morphism (\ref{eq:Fib morphism}) is \emph{holomorphic} if it is fiberwise holomorphic, and \emph{antiholomorphic} if it is fiberwise antiholomorphic.
Henceforth we will restrict ourselves to sheaves $\mathscr F$ that can be written as a composition
\begin{equation}\label{eq:sheaf of complex manifolds}
\Man_n^{op} \xrightarrow{\mathfrak X^{op}} \Fib_n^{op} \xrightarrow{C^\infty} \Set
\end{equation}
where the functor $C^\infty: \Fib_n^{op} \raw \Set$ is the contravariant functor of smooth global sections and
where $\mathfrak X: \Man_n \raw \Fib_n$ is a functor that sends embeddings $M \mono N$ of $n$-manifolds to pullback squares
\begin{equation}\label{eq:pullback square condition}
\begin{tikzcd}
\mathfrak X(M) \arrow[r]\arrow[d] \arrow[dr, phantom," " {pullback=black}, very near start, color=black] & \mathfrak X(N) \arrow[d] \\
M \arrow[r,hook] & N.
\end{tikzcd}
\end{equation}
which are in particular holomorphic morphisms in $\Fib_n$. We will refer to such sheaves as \emph{sheaves of complex manifolds}

\begin{remark}
The condition (\ref{eq:pullback square condition}) means that $\mathfrak X(M)$ has a fixed complex manifold as fiber for every $M$.
\end{remark}

\begin{definition}\label{def:holo def 1}
Let $\mathscr F$ be a sheaf of complex manifolds, $S$ a finite dimensional real (resp. complex) manifold, and $U \in \Man_n$. A map of sets 
\begin{equation}
f: S \raw \mathscr F(U)
\end{equation}
defines a section $\sigma_f$ of the pullback
\begin{equation}
\begin{tikzcd}
pr_2^* \mathfrak X(U) \arrow[d] \arrow[r] & \mathfrak X(U) \arrow[d] \\
S \times U \arrow[u, bend left = 50, "\sigma_f"] \arrow[r, "pr_2"] & U
\end{tikzcd}
\end{equation}
We say $f$ is \emph{smooth} (resp. \emph{holomorphic, antiholomorphic}) if $\sigma_f$ is smooth (resp. holomorphic, antiholomorphic) in $S$.
\end{definition}

\begin{definition}\label{def:holo def 2}
Let $S$ be a finite dimensional real (resp. complex) manifold and $M \in \Pro(\Man_n)$. We say a map $S \raw \mathscr F(\mathring M)$ of sets is \emph{smooth} (resp. \emph{holomorphic, antiholomorphic}) if it admits a factorization through a smooth (resp. holomorphic, antiholomorphic) map
\begin{equation}
S \xrightarrow{f} \mathscr F(U) \raw \mathscr F(M)
\end{equation}
for some morphism $M \mono U$ with $U \in \Man_n$.
\end{definition}

\begin{definition}\label{def: holo def 3}
Let $M,N \in \Pro(\Man_n)$. A map of sets
\begin{equation}
\mathscr F(M) \raw \mathscr F(N)
\end{equation}
is \emph{holomorphic} if the composition
\begin{equation}
S \xrightarrow{f} \mathscr F(M) \raw \mathscr F(N)
\end{equation}
is holomorphic for all holomorphic maps $f$ where $S$ is a finite dimensional complex manifold.
\end{definition}

\begin{definition}\label{def:holo sheaf morphism}
Let $\mathscr F$, $\mathscr G$ 
be sheaves of complex manifolds. A morphism of sheaves $\varphi: \mathscr F \raw \mathscr G$ is \emph{holomorphic} (resp. \emph{antiholomorphic}) if for every manifold $M \in \Man_n$ and holomorphic map $S \raw \mathscr F(M)$, the composition
\begin{equation}
S \raw \mathscr F(M) \raw \mathscr G(M)
\end{equation}
is holomorphic (resp. antiholomorphic).
\end{definition}

\begin{definition}\label{def:holo right inverse}
A morphism of sheaves of complex manifolds $\varphi:\mathscr F \raw \mathscr G$ \emph{admits a holomorphic right inverse} if for every $M \in \Man_n$, the map $\varphi(M): \mathscr F(M) \raw \mathscr G(M)$ has a holomorphic right inverse.
\end{definition}

\begin{remark}\label{rem:right inverse surj}
A morphism of sheaves that admits a holomorphic right inverse is surjective.
\end{remark}

\begin{definition}\label{def:holomorphic nt def}
Let $\mathfrak X, \mathfrak Y: \Man_n \raw \Fib_n$ be functors satisfying condition (\ref{eq:pullback square condition}). We will say a natural transformation $\psi: \mathfrak X \implies \mathfrak Y$ is \emph{holomorphic} (resp. \emph{antiholomorphic}) if it defines maps covering the identity embedding
\begin{equation}\label{eq:X => Y}
\begin{tikzcd}
\mathfrak X(M) \arrow[rd] \arrow[rr, "\psi(M)"] & & \mathfrak Y(M) \arrow[ld]\\ 
& M &
\end{tikzcd}
\end{equation}
that are holomorphic (resp. antiholomorphic). We will say $\psi$ is \emph{submersive} if the maps (\ref{eq:X => Y}) are fiberwise surjective submersions.
\end{definition}
Finally, we record the following proposition which follows from applying the definitions.
\begin{proposition}\label{prop:induced sheaf morphism}
A holomorphic (resp. antiholomorphic) natural transformation $\mathfrak X \implies \mathfrak Y$ induces a holomorphic (resp. antiholomorphic) morphism of sheaves of smooth sections $\mathscr F \raw \mathscr G$.
\end{proposition}


\subsubsection{Nuclearity}
\begin{definition}
A symmetric monoidal functor out of the bordism semicategory 
\begin{equation}
Z^s: \Bord_{n,n-1}^s(\mathscr F) \raw \mathcal{NP}
\end{equation}
is \emph{nuclear} if it sends $\mathscr F$-bordisms to nuclear morphisms of nuclear pairs (cf. Definition \ref{def:nuclear}). We say a symmetric monoidal functor of categories
\begin{equation}
Z: \Bord_{n,n-1}(\mathscr F) \raw \mathcal{NP}
\end{equation} is \emph{nuclear} if its restriction to $\Bord_{n,n-1}^s(\mathscr F)$ is nuclear.
\end{definition}


\subsubsection{Holomorphicity}
Let $Z: \Bord_{n,n-1}(\mathscr F) \raw \mathcal{NP}$ be a nuclear functor.
For each cylindrical germ $\mathring \Sigma$, $Z$ determines families of nuclear dual Fr\'echet and nuclear Fr\'echet spaces 
\begin{equation}\label{eq:N(D)F families}
\begin{split}
\check\pi_{\mathring\Sigma}: \check{\mathscr E}_{\mathring\Sigma} \raw \mathscr F(\mathring\Sigma)\\
\hat\pi_{\mathring\Sigma}: \hat{\mathscr E}_{\mathring\Sigma} \raw \mathscr F(\mathring \Sigma)
\end{split}
\end{equation}
related by a map
\begin{equation}\label{eq:nuclear pair family}
\begin{tikzcd}[row sep = small, column sep = tiny]
\check{\mathscr E}_{\mathring\Sigma} \arrow[rd, "\check \pi_{\mathring\Sigma}"'] \arrow[hook, rr] && \hat{\mathscr E}_{\mathring\Sigma} \arrow[ld, "\hat \pi_{\mathring\Sigma}"]\\
& \mathscr F(\mathring \Sigma) &
\end{tikzcd}
\end{equation}
whose restriction to a fiber $\check E_{\mathring\sigma} \mono \hat E_{\mathring\sigma}$ over $\mathring\sigma \in \mathscr F(\mathring\Sigma)$ is the nuclear pair assigned to $(\mathring \Sigma^+, \mathring \sigma)$.

For each bordism $\mathcal X:= (\mathring U^N_X, p, \mathring\theta_0,\mathring\theta_1): \mathring\Sigma_0^+ \rightsquigarrow \mathring\Sigma_1^+$, there is a correspondence of sets
\begin{equation}\label{eq:correspondence}
\begin{tikzcd}[row sep = small, column sep = small]
& \mathscr F(\mathring U^N_X) \arrow[ld, "r^{\mathcal X}_0"'] \arrow[rd, "r^{\mathcal X}_1"] &\\
\mathscr F(\mathring \Sigma_0) & & \mathscr F(\mathring \Sigma_1)
\end{tikzcd}
\end{equation}
where $r^{\mathcal X}_i$ is the map (\ref{eq:r^X_i}). Applying $Z$ gives a family of maps
\begin{equation}\label{eq:hom family}
\begin{tikzcd}[row sep = small, column sep = tiny]
\underline{\hat{\mathscr E}}_{\mathring \Sigma_0} \arrow[rd] \arrow[rr] && \underline{\check{\mathscr E}}_{\mathring \Sigma_1} \arrow[ld]\\
& \mathscr F(\mathring U^N_X) &
\end{tikzcd}
\end{equation}
where $\underline{\hat{\mathscr E}}_{\mathring \Sigma_0}:= (r^{\mathcal X}_0)^*\hat{\mathscr E}_{\mathring\Sigma_0}$ and $\underline{\check{\mathscr E}}_{\mathring \Sigma_1}:= (r^{\mathcal X}_1)^*\check{\mathscr E}_{\mathring\Sigma_1}$ are obtained by pulling back $\hat \pi_{\mathring \Sigma_0}$ and $\check \pi_{\mathring \Sigma_1}$ along $r^{\mathcal X}_i$. 

\begin{definition}\label{def:holo def}
$Z$ is \emph{holomorphic} if 
\begin{enumerate}
   \item for every cylindrical germ $\mathring\Sigma$ and holomorphic map $f: S \raw \mathscr F(\mathring\Sigma)$ from a finite dimensional complex manifold $S$ (cf. Definition \ref{def:holo def 2}), the pullback of (\ref{eq:nuclear pair family}) along $f$
   \begin{equation}\label{eq:pullback np}
   \begin{tikzcd}[row sep = small, column sep = tiny]
   f^*\check{\mathscr E}_{\mathring\Sigma} \arrow[rd] \arrow[hook, rr] && f^*\hat{\mathscr E}_{\mathring\Sigma} \arrow[ld]\\
   & S &
   \end{tikzcd}
   \end{equation}
   is a map of holomorphic vector bundles (Definition \ref{def:holomorphic bundle}).

   \item for every bordism $\mathcal X:= (\mathring U^N_X,p,\mathring\theta_0,\mathring\theta_1): \mathring \Sigma_0^+ \rightsquigarrow \mathring \Sigma_1^+$ and holomorphic map $f: S \raw \mathscr F(\mathring U^N_X)$, the pullback of (\ref{eq:hom family}) along $f$ is a holomorphic map of vector bundles
   \begin{equation}\label{eq:pullback hom}
   \begin{tikzcd}[row sep = small, column sep = tiny]
   f_0^*\hat{\mathscr E}_{\mathring\Sigma_0} \arrow[rd] \arrow[rr] && f_1^*\check{\mathscr E}_{\mathring\Sigma_1} \arrow[ld]\\
   & S &
   \end{tikzcd}
   \end{equation}
   where we have set $f_i := r^{\mathcal X}_i\circ f$.
\end{enumerate}
\end{definition}
\begin{remark}
The family of maps (\ref{eq:hom family}) determined by $Z$ is equivalent to a section $Z_\mathcal X$ of the family of nuclear dual Fr\'echet spaces
\begin{equation}\label{eq:hom family2}
\begin{tikzcd}
\underline{\hat{\mathscr E}}_{\mathring \Sigma_0}^*\widehat\otimes\underline{\check{\mathscr E}}_{\mathring \Sigma_1} \arrow[d]\\
\mathscr F(\mathring U^N_X) \arrow[u, bend right, "Z_\mathcal X"']
\end{tikzcd}
\end{equation}
where we have used Proposition \ref{prop:Hom Isos} to identify the fiber $\Hom(\hat E_{\mathring \sigma_0}, \check E_{\mathring \sigma_1}) \cong \hat E_{\mathring \sigma_0}^*\widehat\otimes\check E_{\mathring \sigma_1}$ over $\mathring \sigma$ in the preimage 
\begin{equation}\label{eq:preimage}
\mathscr F(\mathring U^N_X; \mathring\sigma_0,\mathring\sigma_1) := (r_0^{\mathcal X}, r_1^{\mathcal X})^{-1}(\mathring\sigma_0,\mathring\sigma_1)
\end{equation} of the map
\begin{equation}
(r^{\mathcal X}_0, r^{\mathcal X}_1): \mathscr F(\mathring U^N_X) \raw \mathscr F(\mathring \Sigma_0) \times \mathscr F(\mathring \Sigma_1).
\end{equation}
Thus the second condition in Definition \ref{def:holo def} is equivalent to requiring that the pullback of $(\ref{eq:hom family2})$
\begin{equation}
\begin{tikzcd}
f^*\left(\underline{\hat{\mathscr E}}_{\mathring \Sigma_0}^*\widehat\otimes\underline{\check{\mathscr E}}_{\mathring \Sigma_1}\right) \arrow[d]\\
S \arrow[u, bend right, "f^* Z_\mathcal X"']
\end{tikzcd}
\end{equation} 
is a holomorphic section of a holomorphic vector bundle for all holomorphic maps $f: S \raw \mathscr F(\mathring U^N_X)$.
\end{remark}

The section (\ref{eq:hom family2}) restricted to $\mathscr F(\mathring U^N_X; \mathring\sigma_0, \mathring\sigma_1)$ determines a map
\begin{equation}\label{eq:Z_X}
Z_{\mathcal X}: \mathscr F(\mathring U^N_X; \mathring\sigma_0,\mathring\sigma_1) \raw \hat E_{\mathring \sigma_0}^* \widehat\otimes \check E_{\mathring\sigma_1}
\end{equation}
which is holomorphic in the sense that if $f: S \raw \mathscr F(\mathring U^N_X;\mathring\sigma_0,\mathring\sigma_1)$ is holomorphic then $Z_{\mathcal X} \circ f: S \raw \hat E_{\mathring \sigma_0}^* \widehat\otimes \check E_{\mathring\sigma_1}$ is holomorphic.

The group $\Diff(\mathcal X) := \mathring{\mathcal G}^{\mathscr F}_{[X]}(\mathcal X, \mathcal X)$ of self-diffeomorphisms of a bordism $\mathcal X := (\mathring U^N_X,p,\mathring\theta_0,\mathring\theta_1)$ acts on $\mathscr F(\mathring U^N_X)$ by pulling back germs of sections and the maps $r_i^{\mathcal X}: \mathscr F(\mathring U^N_X) \raw \mathscr F(\mathring \Sigma_i)$ are invariant under this action. Thus $\Diff(\mathcal X)$ restricts to an action on $\mathscr F(\mathring U^N_X; \mathring\sigma_0,\mathring\sigma_1)$. Since the field theory is defined on diffeomorphism classes of bordisms, the map (\ref{eq:Z_X}) is invariant under the action of $\Diff(\mathcal X)$.


\subsubsection{Coherence}
We will make use of the following lemmas which we simply state; they are consequences of $\mathscr F$ being a sheaf of smooth sections of a fiber bundle.  

\begin{lemma}\label{lem:cylindrical germ rep}
Let $\mathring\sigma \in \mathscr F(\mathring \Sigma)$. Then there exists a section $\sigma \in \mathscr F(C_\Sigma)$ whose restriction to $\mathring \Sigma$ is $\mathring \sigma$.
\end{lemma}

\begin{lemma}\label{lem:cylindrical diff rep}
Let $\mathring f \in \Diff(\mathring \Sigma)$. Then there exists a diffeomorphism $f \in \Diff(C_\Sigma)$ whose restriction to $\mathring \Sigma$ is $\mathring f$.
\end{lemma}

Let $(\mathring \Sigma^+, \mathring\sigma)$ be an object in $\Bord_{n,n-1}^s(\mathscr F)$. By Lemma \ref{lem:cylindrical germ rep}, there is a section $\sigma \in \mathscr F(C_\Sigma)$ whose germ at $\mathring \Sigma$ is $\mathring \sigma$.  

Denote the slice at time $t$ by $\Sigma_t:= \Sigma \times \{t\} \subset C_\Sigma$ and let $\mathring \Sigma_t, \mathring \Sigma_t^+$ its manifold germ without and with the positive co-orientation. Let $X_{a,b} := \Sigma \times [a,b] \subset C_\Sigma$ for $a < b \in \R$, $\mathring U^{C_\Sigma}_{X_{a,b}}$ its manifold germ, $\mathring \sigma_{X_{a,b}}$ the restriction of $\sigma$ to $\mathring U^{C_\Sigma}_{X_{a,b}}$, and $p$ be a partition of $\partial X_{a,b}$ that designates $\Sigma_a$ incoming and $\Sigma_b$ outgoing.
Let 
\begin{equation}\label{eq:translation}
\begin{split}
\theta_t: \Sigma \times \R &\xrightarrow{\sim} \Sigma \times \R\\
(x,s) &\mapsto (x,s+t)
\end{split}
\end{equation}
be translation by $t \in \R$ whose germ at $\mathring \Sigma$ is a co-orientation preserving isomorphism of germs 
$\mathring\theta_t \in \mathring{\mathcal G}^c_{[\Sigma]}(\mathring \Sigma^+, \mathring \Sigma_t^+)$.
Then the tuple
\begin{equation}
\mathcal X_{a,b} := (\mathring U^{C_\Sigma}_{X_{a,b}}, p, \mathring\theta_a, \mathring\theta_b)
\end{equation}
is a cylindrical bordism from $\mathring\Sigma^+$ to itself. If we set $\mathring \sigma_t$ to be the pullback by $\mathring \theta_t$ of the restriction of $\sigma$ to $\mathring \Sigma_t$ then 
\begin{equation}\label{eq:cylindrical F bordism}
(\mathcal X_{a,b},\mathring \sigma_{X_{a,b}})
\end{equation}
is an $\mathscr F$-bordism from $(\mathring \Sigma^+, \mathring \sigma_a)$ to $(\mathring\Sigma^+, \mathring \sigma_b)$. 

Let $Z: \Bord_{n,n-1}^s(\mathscr F) \raw\mathcal{NP}$ be a nuclear symmetric monoidal functor. We will denote by
\begin{equation}
\check E_{\mathring \sigma} \mono \hat E_{\mathring\sigma}
\end{equation}
the nuclear assigned by $Z$ to the object $(\mathring \Sigma^+, \mathring\sigma)$

Applying $Z$ to the $\mathscr F$-bordism (\ref{eq:cylindrical F bordism}) gives a nuclear morphism of nuclear pairs
\begin{equation}
\begin{tikzcd}
\check E_{\mathring \sigma_a} \arrow[d, "\check Z_{a,b}"'] \arrow[r, hook, "\iota_a"] 
&\hat E_{\mathring \sigma_a} \arrow[d, "\hat Z_{a,b}"] \arrow[dl]\\
\check E_{\mathring\sigma_b} \arrow[r,hook, "\iota_b"] 
&\hat E_{\mathring\sigma_b}
\end{tikzcd}
\end{equation}
where $\check Z_{a,b}, \hat Z_{a,b}$ are nuclear by Lemma \ref{lem:nuclear implies nuclear}.

\begin{definition}\label{def:cylindrical systems}
The \emph{direct cylindral system of $\sigma$} is the direct system $\mathcal C_\sigma := \{\check E_{\mathring\sigma_s} \hspace{.1cm} | \hspace{.1cm} s < 0\}$ whose morphisms consist of the maps 
\begin{equation}
\check Z_{s,s'}: \check E_{\mathring\sigma_s} \raw \check E_{\mathring\sigma_{s'}}
\end{equation}
for all $s < s' < 0$.
The \emph{inverse cylindrical system of $\sigma$} is the inverse system $\mathcal D_\sigma:= \{\hat E_{\mathring\sigma_t} \hspace{.1cm} | \hspace{.1cm} t > 0\}$ whose morphisms are 
\begin{equation}
\hat Z_{t,t'}: \hat E_{\mathring\sigma_t} \raw \hat E_{\mathring\sigma_{t'}}
\end{equation}
for all $0 < t < t'$.
\end{definition}
We will denote the direct and inverse limits of the cylindrical systems associated to $\sigma$ by
\begin{equation}\label{eq:C and D}
\check E_{\mathcal C_\sigma} := \colim_{\mathcal C_\sigma} \check E_s
\quad \text{and} \quad
\hat E_{\mathcal D_{\sigma}} := \invlim_{\mathcal D_\sigma} \hat E_t 
\end{equation}
The universal properties of direct and inverse limits implies that (\ref{eq:C and D}) are isomorphic to the direct and inverse limits along any countable final and cofinal subsequence in $\mathcal C_\sigma, \mathcal D_\sigma$; it follows that we can apply Proposition \ref{prop:inverse limit nuclear} and Corollary \ref{cor:direct limit nuclear} to conclude they are in $\mathcal{NDF}$ and $\mathcal{NF}$ respectively. The universal properties also imply the existence of a continuous map 
\begin{equation}\label{eq:cts UP map}
\check E_{\mathcal C_\sigma} \raw \hat E_{\mathcal D_{\sigma}}
\end{equation}
which is nuclear by Proposition \ref{prop:NDF2DF}.

The universal properties of the direct and inverse limits further imply the existence of unique maps making the diagrams
\begin{equation}
\begin{tikzcd}
\check E_{\mathring\sigma_s} \arrow[r] \arrow[rd, " \check Z_{s,0}"'] & \check E_{\mathcal C_\sigma} \arrow[d, dashed, "\exists !"]\\
 & \check E_{\mathring\sigma}
\end{tikzcd}
\quad\text{and}\quad
\begin{tikzcd}
\hat E_{\mathcal D_\sigma}  \arrow[r] & \hat E_{\mathring\sigma_t} \\
 \hat E_{\mathring\sigma}\arrow[ru, "\hat Z_{0,t}"'] \arrow[u, dashed, "\exists !"] &
\end{tikzcd}
\end{equation}
commute for all $s<0 < t$; here the horizontal maps $\check E_{\mathring \sigma_s} \raw \check E_{\mathcal C_\sigma}$ and $\hat E_{\mathcal D_{\sigma}} \raw \hat E_{\mathring \sigma_t}$ are the canonical maps to and from the direct and inverse limits.
  
\begin{definition}
We say $Z$ is \emph{coherent} if the above maps are isomorphisms that make the square
\begin{equation}
\begin{tikzcd}
\check E_{\mathcal C_\sigma} \arrow[r] \arrow[d, "\cong" rot] & \hat E_{\mathcal D_\sigma} \\
\check E_{\mathring\sigma} \arrow[r,hook] & \hat E_{\mathring\sigma} \arrow[u, "\cong" rot]
\end{tikzcd}
\end{equation}
commute; here the top map is (\ref{eq:cts UP map}).
\end{definition}

\begin{remark}
The above definition is independent of which section we choose to represent the germ $\mathring\sigma$, for if $\sigma' \in \mathscr F(C_\Sigma)$ with germ $\mathring\sigma$ at $\mathring\Sigma$, there are final and cofinal subsequences on which $\mathcal C_\sigma, \mathcal C_{\sigma'}$ and $\mathcal D_\sigma, \mathcal D_{\sigma'}$ agree which implies canonical isomorphisms
\begin{equation}
\begin{tikzcd}
\check E_{\mathcal C_\sigma} \arrow[r] \arrow[d, "\cong" rot, no head] & \hat E_{\mathcal D_\sigma} \\
\check E_{\mathcal C_{\sigma'}} \arrow[r] & \hat E_{\mathcal D_{\sigma'}} \arrow[u, "\cong" rot, no head]
\end{tikzcd}
\end{equation}
\end{remark}


\subsubsection{The Action of Diffeomorphisms}\label{subsection: action of diffeos}

There is an action
\begin{equation}\label{eq:action}
\begin{split}
\Diff(\mathring U^N_Y) \times \mathscr F(\mathring U^N_Y) &\raw \mathscr F(\mathring U^N_Y)\\
(\mathring f,\mathring \sigma)\quad\quad &\mapsto (\mathring f^{-1})^*\mathring\sigma
\end{split}
\end{equation}
of germs of diffeomorphisms on germs of sections.

We will now see that a consequence of coherence is that the action (\ref{eq:action}) of $\Diff(\mathring \Sigma)$ on $\mathscr F(\mathring \Sigma)$ lifts to an action on the family of nuclear pairs (\ref{eq:nuclear pair family}). We first describe the action of the subgroup $\Diff(\mathring \Sigma^+)$ of co-orientation preserving germs of diffeomorphisms.

Let $\mathring f \in \Diff(\mathring \Sigma^+)$; by Lemma \ref{lem:cylindrical diff rep} there exists $f \in \Diff(C_\Sigma)$ whose germ at $\Sigma$ is $\mathring f$. Let $\mathring \sigma \in \mathscr F(\mathring \Sigma)$ and let $\sigma \in \mathscr F(C_\Sigma)$ whose germ at $\Sigma$ is $\mathring \sigma$. 
The action of $\mathring f$ on $\mathring\sigma$ gives the germ $\mathring\sigma' := (\mathring f^{-1})^*\mathring\sigma$. The section $\sigma':= (f^{-1})^*\sigma \in \mathscr F(C_\Sigma)$ has germ $\mathring\sigma'$ at the 0-slice. As before, we set $\mathring \sigma_t, \mathring\sigma'_t \in \mathscr F(\mathring \Sigma)$ to be the pullback by $\mathring \theta_t$ of the restriction of $\sigma, \sigma'$ to $\mathring \Sigma_t$.

Let $m,M: \R \raw \R$ be the continuous monotone increasing functions defined by
\begin{equation}
\begin{split}
m(t) := \min_{x \in \Sigma_t} pr_2(f(x))\\
M(t) := \max_{x\in\Sigma_t} pr_2(f(x))
\end{split}
\end{equation}
where $pr_2: \Sigma \times \R \raw \R$ is the projection onto the second factor. They satisfy $m(0) = M(0) = 0$ and $m(t) \leq M(t)$ for all $t\in \R$. In particular, we can find monotone sequences $\{s_i\}, \{s_i'\} \subset \R_{<0}$ and $\{t_i\}, \{t_i'\} \subset \R_{>0}$ converging to 0 such that 
\begin{equation}
M(s_i) < s_i'< m(s_{i+1}) \quad \text{and} \quad M(t_{i+1})< t_i' < m(t_{i}).
\end{equation}

Set 
\begin{equation}
\begin{split}
X_{s_i,s_i'} &:= f(\Sigma \times (s_i, \infty)) \cap \Sigma \times (-\infty, s_i')\\
Y_{s_i',s_{i+1}} &:= \Sigma \times (s_i', \infty) \cap f(\Sigma \times (-\infty, s_{i+1}))
\end{split}
\quad\quad
\begin{split}
X_{t_{i+1},t_i'} &:= f(\Sigma \times (t_{i+1}, \infty)) \cap \Sigma \times (-\infty, t_i')\\
Y_{t_i',t_{i}} &:= \Sigma \times (t_i', \infty) \cap f(\Sigma \times (-\infty, t_{i}))
\end{split}
\end{equation}
and let $p_X: \partial X_{a',b} \raw\{0,1\}$, $p_Y: \partial Y_{a,b'} \raw \{0,1\}$ be partitions that designate $\Sigma_{a'}$,$f(\Sigma_a)$ as incoming and $f(\Sigma_b),\Sigma_{b'}$ as outgoing. If we define
\begin{equation}
\begin{split}
\mathcal X_{s_i,s_i'} := (\mathring U^{C_\Sigma}_{X_{s_i,s_i'}}, p_X, \mathring\theta_{s_i}, \mathring f^{-1} \circ \mathring\theta_{s_i'})\\
\mathcal Y_{s_i',s_{i+1}} := (\mathring U^{C_\Sigma}_{Y_{s_i',s_{i+1}}}, p_Y, \mathring f^{-1} \circ\mathring\theta_{s_i'}, \mathring\theta_{s_{i+1}})\\
\end{split}
\quad\quad
\begin{split}
\mathcal X_{t_{i+1},t_i'} := (\mathring U^{C_\Sigma}_{X_{t_{i+1},t_i'}}, p_X, \mathring\theta_{t_{i+1}}, \mathring f^{-1} \circ\mathring\theta_{t_i'})\\
\mathcal Y_{t_i',t_{i}} := (\mathring U^{C_\Sigma}_{Y_{t_i',t_{i}}}, p_Y, \mathring f^{-1} \circ\mathring\theta_{t_i'}, \mathring\theta_{t_{i}})\\
\end{split}
\end{equation}
and set $\mathring \sigma_X, \mathring\sigma_Y$ to be the restrictions of $\sigma$ to $\mathring U^{C_\Sigma}_{X_{a',b}}$ and $\mathring U^{C_\Sigma}_{Y_{a,b'}}$ then
\begin{equation}\label{eq: interweaving bordisms diff}
\begin{split}
(\mathcal X_{s_i,s_i'}, \mathring \sigma_X): (\mathring \Sigma^+, \mathring\sigma_{s_i}) \rightsquigarrow (\mathring \Sigma^+, \mathring \sigma'_{s_i'})\\
(\mathcal Y_{s_i',s_{i+1}}, \mathring \sigma_Y): (\mathring \Sigma^+, \mathring\sigma'_{s_i'}) \rightsquigarrow (\mathring \Sigma^+, \mathring\sigma_{s_{i+1}})
\end{split}
\quad\quad
\begin{split}
(\mathcal X_{t_{i+1},t_i'}, \mathring \sigma_X): (\mathring \Sigma^+, \mathring\sigma_{t_{i+1}}) \rightsquigarrow (\mathring \Sigma^+, \mathring \sigma'_{t_i'})\\
(\mathcal Y_{t_i',t_{i}}, \mathring \sigma_Y): (\mathring \Sigma^+, \mathring\sigma'_{t_i'}) \rightsquigarrow (\mathring \Sigma^+, \mathring\sigma_{t_{i}})
\end{split}
\end{equation}
are $\mathscr F$-bordisms.

Applying a nuclear symmetric monoidal functor 
\begin{equation}
Z^s: \Bord_{n,n-1}^s(\mathscr F) \raw \mathcal{NP}
\end{equation}
to the above gives the direct system
\begin{equation}\label{eq:interwoven direct}
\check E_{\mathring \sigma_{s_1}} \xrightarrow{Z^s(\mathcal X_{s_1,s_1'}, \mathring \sigma_X)} \check E_{\mathring \sigma'_{s_1'}} \xrightarrow{Z^s(\mathcal Y_{s_1',s_{2}}, \mathring \sigma_Y)} \check E_{\mathring\sigma_{s_{2}}} \raw \cdots
\end{equation}
and the inverse system
\begin{equation}\label{eq:interwoven inverse}
\cdots \hat E_{\mathring \sigma_{t_2}} \xrightarrow{Z^s(\mathcal X_{t_2,t_1'}, \mathring \sigma_X)} \hat E_{\mathring \sigma'_{t_1'}} \xrightarrow{Z^s(\mathcal Y_{t_1',t_{1}}, \mathring \sigma_Y)} \hat E_{\mathring\sigma_{t_{1}}}
\end{equation}

The unprimed alternating terms give final and cofinal subsequences $\tilde{\mathcal C}_{\sigma} := \{\check E_{\sigma_{s_i}}\}$ and $\tilde{\mathcal D}_{\sigma} := \{\hat E_{\sigma_{t_i}}\}$ of the direct and inverse systems $\mathcal C_\sigma$ and $\mathcal D_\sigma$. Likewise, the primed alternating terms give final and cofinal subsequences $\tilde{\mathcal C}_{\sigma'}$ and $\tilde{\mathcal D}_{\sigma'}$ of $\mathcal C_{\sigma'}$ and $\mathcal D_{\sigma'}$. Set
\begin{equation}
\begin{split}
\check E_{\tilde{\mathcal C}_{\sigma}} := \colim_{\tilde{\mathcal C}_{\sigma}} \check E_{\sigma_{s_i}}\\
\check E_{\tilde{\mathcal C}_{\sigma'}} := \colim_{\tilde{\mathcal C}_{\sigma'}} \check E_{\sigma'_{s_i'}}
\end{split}
\quad\quad
\begin{split}
\hat E_{\tilde{\mathcal D}_{\sigma}} := \invlim_{\tilde{\mathcal D}_{\sigma}} \hat E_{\sigma_{t_i}}\\
\hat E_{\tilde{\mathcal D}_{\sigma'}} := \invlim_{\tilde{\mathcal D}_{\sigma'}} \hat E_{\sigma'_{t'_i}}
\end{split}
\end{equation}

If $Z^s$ is coherent, we have the following sequence of isomorphisms

\begin{equation}\label{eq:iso chain}
\begin{tikzcd}
\check E_{\mathring\sigma} \arrow[r,"\cong"] \arrow[d, hook] 
& \check E_{\mathcal C_{\sigma}} \arrow[d] \arrow[r,"\cong"] 
& \check E_{\tilde{\mathcal C}_{\sigma}} \arrow[d] \arrow[r,"\cong"] 
& \check E_{\tilde{\mathcal C}_{\sigma'}} \arrow[d] \arrow[r,"\cong"] 
& \check E_{\mathcal C_{\sigma'}} \arrow[d] \arrow[r,"\cong"]
& \check E_{\mathring\sigma'} \arrow[d,hook]\\
\hat E_{\mathring\sigma} \arrow[r,"\cong"]
&\hat E_{\mathcal D_{\sigma}} \arrow[r,"\cong"] 
& \hat E_{\tilde{\mathcal D}_{\sigma}} \arrow[r,"\cong"] 
& \hat E_{\tilde{\mathcal D}_{\sigma'}} \arrow[r,"\cong"] 
& \hat E_{\mathcal D_{\sigma'}} \arrow[r,"\cong"]
& \hat E_{\mathring \sigma'}
\end{tikzcd}
\end{equation}
where the left- and right-most isomorphisms are from coherence and the rest are canonical from the universal properties of the direct and inverse limits. The middle isomorphism is from the interwoven systems (\ref{eq:interwoven direct}) and (\ref{eq:interwoven inverse}). 

\begin{remark}
The composition (\ref{eq:iso chain}) is an isomorphism of nuclear pairs that is independent of the choice of $f,\sigma$ representing $\mathring f$, $\mathring\sigma$.
\end{remark}

We summarize this in the following.
\begin{theorem}\label{thm:diffeomorphism extension}
Let $Z^s: \Bord_{n,n-1}^s(\mathscr F) \raw \mathcal{NP}$ be a nuclear and coherent symmetric monoidal functor of semicategories. Then $Z^s$ extends canonically to a nuclear and coherent symmetric monoidal functor of categories
\begin{equation}
Z: \Bord_{n,n-1}(\mathscr F) \raw \mathcal{NP}
\end{equation}
\end{theorem}

We now describe the action of reversing the co-orientation.
Let $\mathring \delta \in \Diff(\mathring \Sigma)$ be the germ at $\Sigma \subset C_\Sigma$ of the co-orientation reversing diffeomorphism
\begin{equation}
\begin{split}
\delta: \Sigma \times \R &\xrightarrow{\sim} \Sigma \times \R\\
(x,t) & \mapsto (x,-t)
\end{split}
\end{equation}
We use this to define the twisted involution
\begin{equation}
\delta_{\mathscr F}: \Bord_{n,n-1}(\mathscr F) \raw \Bord_{n,n-1}(\mathscr F)^{op}
\end{equation}
acting on objects and morphisms by
\begin{equation}
\begin{split}
(\mathring \Sigma^+, \mathring \sigma) &\mapsto (\mathring \Sigma^+, \mathring\delta^*\mathring\sigma)\\
(\mathcal X, \mathring\sigma_X) &\mapsto (\mathcal X^*, \mathring\sigma_X)
\end{split}
\end{equation}
where $\mathcal X := (\mathring U^N_X,p,\mathring\theta_0, \mathring\theta_1)$ is a bordism and $\mathcal X^* := (\mathring U^N_X, 1-p, \mathring \delta^*\mathring\theta_1, \mathring\delta^*\mathring\theta_0)$ is $\mathcal X$ with incoming and outgoing boundaries reversed.
\begin{remark}
Although we defined the $\delta_{\mathscr F}$ on $\mathscr F$-bordisms, it descends to diffeomorphism classes of bordisms since $\mathring f_i \circ \mathring \delta^* \mathring\theta_i = \mathring\delta^*(\mathring f_i \circ \mathring\theta_i)$ for all diffeomorphisms $\mathring f$ of $\mathscr F$-bordisms.
\end{remark}

\begin{proposition}\label{prop:co-orientation dual}
Let $Z: \Bord_{n,n-1}(\mathscr F) \raw \mathcal{NP}$ be a nuclear and coherent symmetric monoidal functor. Then $Z$ is $(\delta_{\mathscr F},\delta_{\mathcal{NP}})$-equivariant.
\end{proposition}

We refer the reader to the appendix to Section 3 of \cite{KS} for a proof.


\subsubsection{Definition of Field Theory}

\begin{definition}\label{def:Field Theory Def}
An \emph{$\mathscr F$-field theory} is a symmetric monoidal functor 
\begin{equation}\label{eq:semi field theory}
\Bord_{n,n-1}^s(\mathscr F) \raw \mathcal{NP}
\end{equation}
that is nuclear, holomorphic, and coherent.
\end{definition}

\begin{remark}
By Theorem \ref{thm:diffeomorphism extension}, an $\mathscr F$-field theory (\ref{eq:semi field theory}) extends canonically to a symmetric monoidal functor out of the full $\mathscr F$-bordism category
\begin{equation}
\Bord_{n,n-1}(\mathscr F) \raw \mathcal{NP}
\end{equation}
This extension will also be called an $\mathscr F$-field theory.
\end{remark}

We record the following proposition for use in the next section.

\begin{proposition}\label{prop:F to G field theory}
Let $\varphi: \mathscr F \raw \mathscr G$ be a holomorphic morphism between sheaves of complex manifolds admitting a holomorphic right inverse and let $Z_{\mathscr F} : \Bord_{n,n-1}(\mathscr F) \raw \mathcal{NP}$ be an $\mathscr F$-field theory that admits a factorization
\begin{equation}
\begin{tikzcd}
\Bord_{n,n-1}(\mathscr F) \arrow[r, "Z_{\mathscr F}"] \arrow[d, "\varphi_*"] &  \mathcal{NP} \\
\Bord_{n,n-1}(\mathscr G) \arrow[ru, "Z_{\mathscr G}"'] &
\end{tikzcd}
\end{equation}
where $\varphi_*$ is the functor (\ref{eq:induced sheaf functor}) induced from the sheaf morphism. Then $Z_{\mathscr G}$ is a $\mathscr G$-field theory.
\end{proposition}
\begin{proof}
Let $(\mathcal X, \mathring \tau_X)$ be a $\mathscr G$-bordism in $\Bord_{n,n-1}(\mathscr G)$. The sheaf morphism $\varphi$ is surjective (cf. Remark \ref{rem:right inverse surj}) which implies there exists $\mathring \sigma_X \in \mathscr F(\mathring U^N_X)$ satisfying $\varphi(\mathring\sigma_X) = \mathring\tau_X$. Thus $Z_{\mathscr G}(\mathcal X, \mathring \tau_X) = Z_{\mathscr F}(\mathcal X, \mathring\sigma_X)$ is a nuclear morphism of nuclear pairs, which shows that $Z_{\mathscr G}$ is nuclear.

Let $\sigma \in \mathscr G(C_\Sigma)$ with germ $\mathring\sigma \in \mathscr G(\mathring \Sigma)$ and let 
\begin{equation}
\psi(\mathring \Sigma): \mathscr G(\mathring \Sigma) \raw \mathscr F(\mathring \Sigma)
\end{equation}
be a holomorphic right inverse to $\varphi(\mathring\Sigma): \mathscr F(\mathring \Sigma) \raw \mathscr G(\mathring \Sigma)$.
Surjectivity of $\varphi$ implies there exists $\tau \in \mathscr F(C_\Sigma)$ with germ $\mathring\tau \in \mathscr F(\mathring\Sigma)$ satisfying $\varphi(\tau)= \sigma$ and $\varphi(\mathring\tau) = \mathring\sigma$. The direct and inverse cylindrical systems obtained from applying $Z$ to the systems of bordisms given by $\sigma, \tau$ are identical, which proves coherence of $Z_\mathscr{G}$.

Let
\begin{equation}\label{eq:G family}
\begin{tikzcd}[row sep = small, column sep = tiny]
\check{\mathscr E}_{\mathring \Sigma} \arrow[rr,hook] \arrow[rd] & & \hat{\mathscr E}_{\mathring \Sigma} \arrow[ld]\\
& \mathscr G(\mathring\Sigma) &
\end{tikzcd}
\end{equation}
be the family of nuclear pairs assigned to $\mathring\Sigma$ by $Z_{\mathscr G}$
and let $f: S \raw \mathscr G(\mathring\Sigma)$ be a holomorphic map. The pullback of (\ref{eq:G family}) along $\varphi(\mathring\Sigma): \mathscr F(\mathring\Sigma) \raw \mathscr G(\mathring \Sigma)$ is the holomorphic bundle of nuclear pairs assigned to $\mathring \Sigma$ by $Z_{\mathscr F}$, which in turn pulls back along $\psi(\mathring \Sigma) \circ f$ to a holomorphic bundle of nuclear pairs on $S$. We can identify this with the pullback of (\ref{eq:G family}) along $f$. A similar argument shows that the bundle of maps assigned to a bordism $\mathcal X: \mathring \Sigma_0 \rightsquigarrow \mathring \Sigma_1$ by $Z_{\mathscr G}$
\begin{equation}\label{eq:G family2}
\begin{tikzcd}[row sep = small, column sep = tiny]
\hat{\mathscr E}_{\mathring \Sigma_0} \arrow[rr] \arrow[rd] & & \check{\mathscr E}_{\mathring \Sigma_1} \arrow[ld]\\
& \mathscr G(\mathcal X) &
\end{tikzcd}
\end{equation}
pulls back along a holomorphic map $f: S \raw \mathscr G(\mathcal X)$ to a holomorphic map of holomorphic bundles.
\end{proof}


\subsubsection{Reflection Positivity}
We now restrict to sheaves of complex manifolds (\ref{eq:sheaf of complex manifolds}) such that $\mathfrak X: \Man_n \raw \Fib_n$ is equipped with an anti-holomorphic involution, i.e. an involutive natural isomorphism $\alpha: \mathfrak X \implies \mathfrak X$ such that $\alpha(U)$ for $U \in \Man_n$ is a bundle map
\begin{equation}
\begin{tikzcd}[row sep = small, column sep = small]
\mathfrak X(U) \arrow[rd] \arrow[rr, "\alpha(U)"] & & \mathfrak X(U) \arrow[ld]\\
& U &
\end{tikzcd}
\end{equation}
which is fiberwise an antiholomorphic involution.

By Proposition \ref{prop:induced sheaf morphism}, this induces an antiholomorphic involution of sheaves $\alpha_*: \mathscr F \raw \mathscr F$. This in turn induces an involution (\ref{eq:induced sheaf functor}) on $\Bord_{n,n-1}(\mathscr F)$ which we denote
\begin{equation}
\alpha_{\mathscr F}: \Bord_{n,n-1}(\mathscr F) \raw \Bord_{n,n-1}(\mathscr F)
\end{equation}

The involutions $\alpha_{\mathscr F}$ and $\delta_{\mathscr F}$ commute and we denote their composition by $\tau_{\mathscr F} := \alpha_{\mathscr F}\circ\delta_{\mathscr F}$. If $Z$ is $(\alpha_{\mathscr F},\alpha_{\mathcal{NP}})$-equivariant, then Proposition \ref{prop:co-orientation dual} implies that $Z$ is $(\tau_{\mathscr F}, \tau_{\mathcal{NP}})$-equivariant. In particular, $Z$ induces a functor
\begin{equation}
Z^\tau: \Bord_{n,n-1}(\mathscr F)^{\tau_{\mathscr F}} \raw \mathcal{NP}^{\tau_{\mathcal{NP}}}
\end{equation}
between fixed point categories. We say $Z$ is \emph{Hermitian} if $Z^\tau$ sends $\tau_{\mathscr F}$ fixed points to Hermitian nuclear pairs (cf. Definition \ref{def:hermitian np}).

\begin{definition}
An $\mathscr F$-field theory is \emph{reflection positive} if $Z$ is $(\alpha_{\mathscr F},\alpha_{\mathcal{NP}})$-equivariant and Hermitian.
\end{definition}


\section{Volume-dependent Field Theories}\label{section VFTs}

\subsection{Definitions and Properties}\label{section VFT defs and props}
We first recall the definition of an allowable complex metric. Let $V$ be a real vector space of dimension $n$ and let $g \in \Sym^2(V^*) \otimes \C$ be a complex symmetric 2-tensor on $V$ inducing a non-degenerate quadratic form on $V\otimes \C$.  Let $e_1,...,e_n$ be a basis of $V\otimes \C$ that is orthonormal with respect to $g$. Nondegeneracy implies an isomorphism
\begin{equation}
T_g : V\otimes \C \cong V^* \otimes \C
\end{equation} and we let $e^1,...,e^n$ be the basis of $V^* \otimes \C$ defined by $e^i := T_g(e_i)$. Set 
\begin{equation}
\vol_g := e^1 \wedge \cdots \wedge e^n
\end{equation}
which we view as a complex density in
\begin{equation}
|\Lambda^n(V^*)|_\C := (\Lambda^n(V^*) \otimes \C)\otimes \mathfrak o(V)
\end{equation}
where $\mathfrak o(V)$ is the orientation line of $V$. Tensoring by the orientation line implies that $\vol_g$ is independent of the choice of orthonormal basis.
Contraction with $\vol_g$ gives a map
\begin{equation}
\begin{split}
\ast_g: \Lambda^k V^* \otimes \C &\raw \Lambda^{n-k} V^* \otimes \C \otimes \mathfrak o(V)\\
\alpha &\mapsto \iota_{(\Lambda^k T_g^{-1})\alpha} \vol_g
\end{split}
\end{equation}
which restricts to the ordinary Hodge star operator on $\Lambda^\bullet V$ when $g$ is real. Define the quadratic form
\begin{equation}
\begin{split}
Q_g : \Lambda^\bullet V^* \otimes \C &\raw |\Lambda^n V^*|_\C\\
\alpha & \mapsto \alpha \wedge *\alpha
\end{split}
\end{equation}
which is also independent of the choice of orthonormal basis.

\begin{definition}[allowable metric on a vector space]
A nondegenerate complex symmetric 2-tensor $g \in \Sym^2(V^*)\otimes \C$ is \emph{allowable} if the real part of the quadratic form $Q_g$ is positive definite. We denote the set of allowable 2-tensors by $\met_\C(V)$.
\end{definition}

\begin{definition}[allowable density on a vector space]
A complex density in $|\Lambda^n V^*|_\C$ will be called an \emph{allowable density on $V$} if its real part is positive. We will denote by $\dens_\C(V)$ the set of allowable densities.
\end{definition} 

We now recall an equivalent characterization of allowable metrics whose proof can be found in \cite{KS}.

\begin{proposition}[\cite{KS} Theorem 2.2]\label{prop:KS theorem 2.2}
A complex quadratic form $g \in \Sym^2(V^*)\otimes \C$ is allowable if and only if there exists a basis $e_1,...,e_n$ of $V$ in which $g$ can be expressed as
\begin{equation}
g = \sum_{i=1}^n \lambda_i e^i \otimes e^i
\end{equation}
where $(e^i)$ is the coordinate dual basis and $\lambda_i \in \C\backslash \R_{\leq 0}$ are nonzero complex numbers not on the negative real axis satisfying
\begin{equation}\label{eq:metc arg condition}
\sum_{i=1}^n |\arg(\lambda_i)| < \pi.
\end{equation}
\end{proposition}

The condition (\ref{eq:metc arg condition}) is open which implies that allowable metrics $\met_\C(V)$ form an open subset of $\Sym^2(V^*) \otimes \C$ and therefore a complex manifold. Likewise, the allowable densities $\dens_\C(V)$ are an open subset of $|\Lambda^n V^*|_\C$ and also form a complex manifold.
If $g$ is allowable, applying $Q_g$ to $1 \in \Lambda^0 V^*$ shows that $\vol_g \in \dens_\C(V)$.

\begin{proposition}\label{prop:root det vector props}
The map
\begin{equation}\label{eq:root det vector}
\begin{split}
\sqrt{\det} : \met_\C(V) &\raw \dens_\C(V)\\
g &\mapsto \vol_g
\end{split}
\end{equation}
is holomorphic, equivariant with respect to the action of $\GL(V)$ on $met_\C(V)$ and $\dens_\C(V)$, and has a holomorphic right inverse.
\end{proposition}
\begin{proof}
Let $g \in \met_\C(V)$ be fixed. By Proposition \ref{prop:KS theorem 2.2}, there exists a basis $(e_i)$ of $V$ such that $g = \sum_i \lambda_i e^i \otimes e^i$ with $\sum_i |\arg(\lambda_i)| < \pi$. Let
\begin{equation}
z^{1/2} : \C \backslash \R_{\leq 0} \raw \C_{>0}
\end{equation}
be the holomorphic square root defined on the branch away from the negative real axis. Then $(\lambda_i^{1/2} e^i)$ is a basis of $V^* \otimes \C$ which is orthonormal with respect to $g$.
The map (\ref{eq:root det vector}) applied to $g$ then becomes
\begin{equation}
\sum_i \lambda_i e^i \otimes e^i \mapsto \prod_i \lambda_i^{1/2} \cdot |e^1\wedge \cdots \wedge e^n|
\end{equation}
where $|e^1 \wedge\cdots \wedge e^n|$ is a positive density in $\dens_\R(V) := \Lambda^n(V^*) \otimes \mathfrak o(V)$.

The map (\ref{eq:root det vector}) can be expressed in terms of any other basis $(f_i)$ of $V$ as
\begin{equation}\label{eq:root det gen basis}
\sum_{i,j} g_{ij}f^i \otimes f^j \mapsto \det(g_{ij})^{1/2} \cdot |f^1 \wedge \cdots \wedge f^n|
\end{equation}
with
\begin{equation}\label{eq:root det equiv}
\det(g_{ij})^{1/2} = |\det(A)|\cdot\prod_i \lambda_i^{1/2}
\end{equation}
an element of $\C \backslash \R_{\leq 0}$,
where $A \in \GL(V)$ is the unique invertible linear operator satisfying $Ae_i = f_i$. Equation (\ref{eq:root det gen basis}) shows holomorphicity and (\ref{eq:root det equiv}) shows equivariance.

Let $h \in \met_\R(V) \subset \met_\C(V)$ be a metric of Euclidean signature and let $\vol_h \in \dens_\R(V)$ be the associated positive real density on $V$. Any element $\omega \in \dens_\C(V)$ can be expressed as 
\begin{equation}
\omega = c\cdot \vol_h
\end{equation}
for $c \in\C_{>0}$. Then the map
\begin{equation}
\begin{split}
\dens_\C(V) &\raw \met_\C(V)\\
c \cdot \vol_h &\mapsto c^{2/n}\cdot h
\end{split}
\end{equation}
is a holomorphic right inverse to $\sqrt{\det}$.
\end{proof}

\begin{remark}
The \emph{Shilov boundary} of an open subset $U \subset \mathbb A$ of a finite dimensional complex affine space is the smallest compact subset $K$ of the closure $\overline U$ such that every holomorphic function defined on a neighborhood of $\overline U$ attains its maximum modulus on $K$ when restricted to $\overline U$. The set $\met_{\Lor}(V)$ of possibly degenerate metrics of Lorentzian signature is a subset of the Shilov boundary of $\met_\C(V)$. Likewise, the set $\dens_{i\R}$ of purely imaginary densities on $V$ can be identified with the Shilov boundary of $\dens_\C(V)$.
The map (\ref{eq:root det vector}) extends to a $\GL(V)$-equivariant map 
\begin{equation}\label{eq:lor root det vector}
\sqrt{\det}: \met_{\Lor}(V) \raw \dens_{i\R}(V)
\end{equation}
For a more detailed discussion of the Shilov boundary we refer the reader to \cite{KS}.
\end{remark}

Let $M$ be a smooth $n$-manifold, let $\mathcal B(M) \raw M$ be the principal $\GL_n(\R)$ bundle of bases on $M$, and let $V$ be a real vector space of dimension $n$.  The bundle of allowable metrics on $M$ is defined to be the associated bundle
\begin{equation}
\met_\C(M) := \mathcal B(M) \times_{\GL_n} \met_\C(V)
\end{equation}
Similarly, we define the bundle of allowable densities on $M$ to be the associated bundle
\begin{equation}
\dens_\C(M) := \mathcal B(M) \times_{\GL_n} \dens_\C(V).
\end{equation}

Equivariance of the map (\ref{eq:root det vector}) gives a map between associated bundles
\begin{equation}\label{eq:root det bundle}
\sqrt{\det}: \met_\C(M) \raw \dens_\C(M)
\end{equation}
which we also showed is holomorphic on the fibers. 

We denote by
\begin{equation}
\Met_\C: \Man_n^{op} \xrightarrow{\met_\C} \Fib_n^{op} \xrightarrow{C^\infty} \Set
\end{equation}
\begin{equation}
\Dens_\C: \Man_n^{op} \xrightarrow{\dens_\C} \Fib_n^{op} \xrightarrow{C^\infty} \Set
\end{equation}
the sheaves of allowable metrics and densities on the site of smooth $n$-manifolds and embeddings between them which, particular, are sheaves of complex manifolds. 

\begin{lemma}\label{lem: metc to densc}
The morphism of sheaves
\begin{equation}
\sqrt{\det}: \Met_\C \raw \Dens_\C
\end{equation}
whose value on $M\in \Man_n$ is induced by (\ref{eq:root det bundle}), admits a holomorphic right inverse (cf. Definition \ref{def:holo right inverse}).
\end{lemma}
\begin{proof}
Let $M \in \Man_n$ and let $h \in \Met_\R(M)$ be a Riemannian metric. If $\omega \in \Dens_\C(M)$ then it can be expressed as
\begin{equation}
\omega = f\cdot |\vol_h|
\end{equation}
where $f \in C^\infty(M;\C_{>0})$ and $|\vol_h|\in \Dens_\R(M)$ is the positive real density on $M$ associated to $h$. Then Proposition \ref{prop:root det vector props} implies the map
\begin{equation}
\begin{split}
\Dens_\C(M) &\raw \Met_\C(M)\\
f \cdot|\vol_g| &\mapsto f^{2/n}\cdot h
\end{split}
\end{equation}
is a holomorphic right inverse to $(\sqrt{\det})(M)$.
\end{proof}

Likewise, we set
\begin{equation}
\Met_{\Lor}: \Man_n^{op} \xrightarrow{\met_{\Lor}} \Fib_n^{op} \xrightarrow{C^\infty} \Set
\end{equation}
\begin{equation}
\Dens_{i\R}: \Man_n^{op} \xrightarrow{\dens_{i\R}} \Fib_n^{op} \xrightarrow{C^\infty} \Set
\end{equation}
to be the sheaves of possibly degenerate Lorentzian metrics and purely imaginary densities. As in Lemma \ref{lem: metc to densc}, the map (\ref{eq:lor root det vector}) induces a morphism of sheaves
\begin{equation}\label{eq: metlor to densir}
\sqrt{\det}: \Met_{\Lor} \raw \Dens_{i\R}
\end{equation}
which we record for later use.

\begin{definition}
A \emph{Wick-rotated quantum field theory} is a $\Met_\C$-field theory.
\end{definition}

\begin{definition}
A \emph{volume-dependent field theory (VFT)} is a $\Met_\C$-field theory admitting a factorization
\begin{equation}\label{eq:vft def}
\begin{tikzcd}
\Bord_{n,n-1}(\Met_\C) \arrow[d, "\sqrt{\det}_*"'] \arrow[r] & \mathcal{NP}\\
\Bord_{n,n-1}(\Dens_\C) \arrow[ru, "Z"'] &
\end{tikzcd}
\end{equation}
(cf. Definition \ref{def:Field Theory Def}).
\end{definition}

\begin{lemma}
The functor $Z$ in (\ref{eq:vft def}) is a $\Dens_\C$-field theory.
\end{lemma}
\begin{proof}
This follows from an application of Proposition \ref{prop:F to G field theory}.
\end{proof}

Let $\check E_i \mono \hat E_i$ be the nuclear pair assigned to $(\mathring \Sigma_i^+,\mathring\omega_i) \in \Bord_{n,n-1}(\Dens_\C)$ by $Z$ for $i=0,1$. 
Let $\mathcal X := (\mathring U^N_X, p, \mathring \theta_0, \mathring\theta_1)$ be a bordism from $\mathring\Sigma_0^+$ to $\mathring\Sigma_1^+$ with $X$ connected. Recall there is a correspondence of background fields (\ref{eq:correspondence}). As in (\ref{eq:preimage}), given $\mathring\omega_i \in \Dens_\C(\mathring\Sigma_i)$ we set 
\begin{equation}
\Dens_\C(\mathring U^N_X; \mathring\omega_0, \mathring\omega_1) := (r_0^{\mathcal X}, r_1^{\mathcal X})^{-1}(\mathring\omega_0,\mathring\omega_1)
\end{equation}
to be the set of allowable densities on $\mathring U^N_X$ whose restriction to $\mathring U^N_{\partial X_i}$ pulled back along $\mathring\theta_i$ to $\mathring \Sigma_i$ is $\mathring\omega_i$.

The field theory $Z$ determines a holomorphic map
\begin{equation}\label{eq:Z_X dens}
Z_{\mathcal X}: \Dens_\C(\mathring U^N_X;\mathring\omega_0,\mathring\omega_1) \raw \Hom(\hat E_0, \check E_1)
\end{equation}
In fact, this map only depends on the total volume of the density.

\begin{theorem}\label{thm:vol dependence thm}
There is a factorization
\begin{equation}
\begin{tikzcd}
\Dens_\C(\mathring U^N_X;\mathring\omega_0,\mathring\omega_1) \arrow[r, "Z_{\mathcal X}"] \arrow[d, "\int_X"'] &  \Hom(\hat E_0, \check E_1)\\
\C_{>0} \arrow[ru, "V_{X}"'] &
\end{tikzcd}
\end{equation}
through a holomorphic map $V_X$, where $\int_X$ is the map that integrates a density over $X$.
\end{theorem}
\begin{proof}
Let $\mathring\nu_0, \mathring\nu_1 \in \Dens_\C(\mathring U^N_X; \mathring\omega_0, \mathring \omega_1)$ be densities on $\mathring U^N_X$ whose restrictions to $X$ have the same total volume $v:= \int_X \mathring\nu_0 = \int_X \mathring\nu_1$. 
For $t\in [0,1]$, let
\begin{equation}\label{eq:straight line path}
\mathring\nu_t := (1-t)\mathring\nu_0 + t\mathring\nu_1
\end{equation}
be the straight-line path in $\Dens_\C(\mathring U^N_X;\mathring\omega_0,\mathring\omega_1)$. It satisfies $\int_X \mathring \nu_t = v$ for all $t$.

Let $\tilde U \in \mathfrak U^N_X$ be a neighborhood of $X$ in $N$ and let $\tilde \nu_i \in \Dens_\C(\tilde U)$ be densities restricting to $\mathring \nu_i$ on $\mathring U^N_X$. Since $\tilde \nu_0, \tilde \nu_1$ restrict to the same germ at $\mathring U^N_{\partial X_i}$, there exist neighborhoods $V_i \in \mathfrak U^N_{\partial X_i}$ for $i=0,1$ such that $\tilde \nu_0|_{V_i} = \tilde \nu_1|_{V_i}$. The density $\tilde \nu_t:= (1-t)\tilde \nu_0 + t\tilde \nu_1 \in \Dens_\C(\tilde U)$ has germ $\mathring \nu_t$ and the restriction of $\tilde \nu_t$ to $V_i$ is constant for all $t$.
Set 
\begin{equation}
\begin{split}
\omega_i := \tilde \nu_t|_{V_i}\\
U:= X \cup V_0 \cup V_1\\
\nu_t:= \tilde \nu_t|_{U}.
\end{split}
\end{equation} 
Connectedness of $X$ implies $U$ is also connected.

Let $\Omega^n_U(\C_{>0}; \omega_0,\omega_1)$ be the set of $n$-forms valued in $\C_{>0}$ that restrict to $\omega_i$ on $V_i$.
We assume that $U$ is orientable and $\nu_t \in \Omega^n_{U}(\C_{>0}; \omega_0, \omega_1)$. The unorientable case reduces to the orientable case by using a partition of unity and identifying the density bundle with top dimensional forms on local neighborhoods. The $n$-form $\dot\nu_t := \frac{d}{dt}\nu_t = \nu_1 - \nu_0$ satisfies $\int_U\dot\nu_t = 0$. By connectedness of $U$, we conclude that $\dot\nu_t$ is exact, i.e.
\begin{equation}
\dot\nu_t = \nu_1 - \nu_0 = d\eta
\end{equation}
for some $\eta \in \Omega^{n-1}_U(\C)$ constant on $V_0 \cup V_1$; without loss of generality we assume that $\eta$ vanishes on this set. Since $\nu_t$ is a nowhere vanishing $n$-form, there exists a smooth complexified vector field $\xi_t \in C^\infty(U;TU\otimes \C)$ satisfying $\iota_{\xi_t} \nu_t = \eta$ with $\xi_t|_{V_i} = 0$. Applying the exterior derivative gives
\begin{equation}
\mathcal L_{\xi_t}\nu_t = \nu_1-\nu_0
\end{equation}
where $\mathcal L_{\xi_t}$ denotes the Lie derivative with respect to $\xi_t$. 

Let
\begin{equation}
Z_U: \Omega^n_U(\C_{>0};\omega_0,\omega_1) \raw \Hom(\hat E_0, \check E_1)
\end{equation}
be the composition of (\ref{eq:Z_X dens}) with the holomorphic map restricting $n$-forms on $U$ to $\mathring U^N_X$.

Let $\Diff(U; V_0 \cup V_1)$ be the group of diffeomorphisms of $U$ that restrict to the identity on $V_0 \cup V_1$ and let $\Lie(\Diff(U;V_0 \cup V_1))$ denote the Lie algebra of vector fields on $U$ that vanish on $V_0 \cup V_1$.
For fixed $t\in [0,1]$, let
\begin{equation}
\begin{split}
\alpha_{\nu_t}: \Diff(U; V_0 \cup V_1) &\raw \Omega^n_U(\C_{>0}; \omega_0, \omega_1)\\
f &\mapsto (f^{-1})^*\nu_t
\end{split}
\end{equation}
be the orbit map.
The derivative at the identity gives a map
\begin{equation}
\begin{split}
d\alpha_{\nu_t}: \Lie(\Diff(U;V_0\cup V_1)) &\raw T_{\nu_t}\Omega^n_U(\C_{>0};\omega_0,\omega_1)\\
\xi &\mapsto \mathcal L_{\xi}\nu_t
\end{split}
\end{equation}
where
\begin{equation}
T_{\nu_t}\Omega^n_U(\C_{>0};\omega_0,\omega_1) = \{\alpha \in \Omega^n_U(\C) \st \alpha|_{V_i} = 0\}
\end{equation}
is the tangent space at $\nu_t$ to the space of allowable densities on $U$ that restrict to $\omega_i$ on $V_i$.

By diffeomorphism invariance, the composition $Z_U \circ \alpha_{\nu_t}$ is constant. Thus the composition
\begin{equation}
\Lie(\Diff(U)) \xrightarrow{d\alpha_{\nu_t}} T_{\nu_t}\Omega^n_U(\C_{>0};\omega_0,\omega_1) \xrightarrow{d Z_{U}} \Hom(\hat E_0, \check E_1)
\end{equation}
is 0. 
Holomorphicity of $Z_U$ implies that the complexification
\begin{equation}
\Lie(\Diff(U))\otimes \C \xrightarrow{d\alpha_{\nu_t}} T_{\nu_t}\Omega^n_U(\C_{>0};\omega_0,\omega_1) \xrightarrow{d Z_{U}} \Hom(\hat E_0, \check E_1)
\end{equation}
is also 0.
The form $\dot\nu_t = \mathcal L_{\xi_t}\nu_t$ is in the image of the complexified map $d\alpha_{\nu_t}$ which implies $dZ_U(\dot\nu_t) = 0$.  Thus $Z_U(\nu_t)$ is constant for all $t \in [0,1]$ and therefore $Z_{\mathcal X}(\mathring\nu_t)$ is constant for all $t \in [0,1]$.

To see that $V_X$ is holomorphic, consider the holomorphic map
\begin{equation}
\begin{split}
m_\nu: \C_{>0} &\raw \Dens_\C(\mathring U^N_X;\mathring\omega_0,\mathring\omega_1)\\
\lambda &\mapsto \lambda \mathring\nu
\end{split}
\end{equation}
for some density $\mathring\nu$ satisfying $\int_X \mathring \nu = 1$. Then the composition
\begin{equation}
\C_{>0} \xrightarrow{m_\nu} \Dens_\C(\mathring U^N_X;\mathring\omega_0,\mathring\omega_1) \xrightarrow{\int_X} \C_{>0}
\end{equation}
is the identity. By holomorphicity of $Z$, the map
\begin{equation}
V_X: \C_{>0} \xrightarrow{m_\nu} \Dens_\C(\mathring U^N_X; \mathring\omega_0, \mathring\omega_1) \xrightarrow{Z_{\mathcal X}} \Hom(\hat E_0, \check E_1)
\end{equation}
is holomorphic.
\end{proof}

\begin{definition}\label{def:signed volume curve}
Let $\Sigma$ be a connected closed $n-1$-manifold and let $\omega \in \Dens_\C(C_\Sigma)$ be a density. The \emph{signed volume curve of $\omega$} is the curve $\gamma_\omega: \R \raw \C$ given by
\begin{equation}
\gamma_\omega(t) := 
\begin{cases}
-\int_{\Sigma \times [t,0]} \omega & t < 0\\
0 & t = 0\\
\int_{\Sigma \times [0,t]} \omega & t > 0
\end{cases}
\end{equation}
It is smooth and satisfies $\gamma'(t) \in \C_{>0}$ for all $t\in \R$. When $\Sigma := \bigsqcup_i \Sigma^i$ has multiple connected components, we define the signed volume curve $\gamma_\omega: \R \raw \C^{\#\pi_0(\Sigma)}$ to be the curve whose $i$th component is the signed volume curve $\gamma_{\omega_i}: \R \raw \C$ of the restriction $\omega|_{\Sigma^i}$ to the $i$th connected component.
\end{definition}

\begin{proposition}\label{prop:canonical NP isos}
Let $Z: \Bord_{n,n-1}(\Dens_\C) \raw \mathcal{NP}$ be a volume-dependent field theory, let $\mathring\omega, \mathring\omega' \in \Dens_\C(\mathring\Sigma)$, and let $\check E_{\mathring\omega} \mono \hat E_{\mathring\omega}$ and $\check E_{\mathring\omega'} \mono \hat E_{\mathring\omega'}$ be the nuclear pairs assigned to them by $Z$. Then there exists a canonical isomorphism
\begin{equation}\label{eq:canonical density germ isos}
\begin{tikzcd}
\check E_{\mathring\omega} \arrow[d,"\cong" rot, "\check \alpha"] \arrow[r,hook] & \hat E_{\mathring\omega} \arrow[d, "\cong" rot, "\hat \alpha"] \\
\check E_{\mathring\omega'} \arrow[r,hook] & \hat E_{\mathring\omega'}
\end{tikzcd}
\end{equation}
of nuclear pairs.
\end{proposition}

\begin{proof}
If $\Sigma$ is empty, the isomorphism is the identity map on the trivial nuclear pair $\C \mono \C$. If not, let $\omega, \omega' \in \Dens_\C(C_\Sigma)$ be allowable densities whose germs at $\mathring \Sigma$ are $\mathring\omega, \mathring\omega'$ and let $\gamma_\omega, \gamma_{\omega'}$ be their respective signed volume curves. The real parts $\Re(\gamma_\omega)$, $\Re(\gamma_{\omega'})$ are continuous monotone increasing functions vanishing at 0.
Thus there exist monotone real sequences $\{s_i\}_{i=0}^\infty,\{s_i'\}_{i=0}^\infty \subset \R_{<0}$ increasing to 0 and $\{t_i\}_{i=0}^\infty,\{t_i'\}_{i=0}^\infty \subset \R_{>0}$ decreasing to 0 such that 
\begin{equation}\label{eq:nested volumes}
\Re(\gamma_\omega(s_{i})) < \Re(\gamma_{\omega'}(s_{i}')) < \Re(\gamma_\omega(s_{i+1}))
\quad\text{,} \quad
\Re(\gamma_\omega(t_{i+1})) < \Re(\gamma_{\omega'}(t_i')) < \Re(\gamma_\omega(t_{i}))
\end{equation}
for all $i \geq 0$. 
Let
\begin{equation}
\begin{split}
X_i  := \Sigma \times [s_{i}, s_{i+1}]\\
X_i' := \Sigma \times [\tilde s_i,\tilde s_{i+1}]
\end{split}
\quad\quad
\begin{split}
Y_i  := \Sigma \times [t_{i+1}, t_{i}]\\
Y_i' := \Sigma \times [t_{i+1}', t_i']
\end{split}
\end{equation}
which are compact manifolds smoothly embedded in $C_\Sigma$ whose boundaries we co-orient positively. In general for $W := \Sigma \times [a,b] \subset C_\Sigma$, we set $\mathring \omega_W, \mathring\omega_W'$ to be the restrictions of $\omega, \omega'$ to $\mathring U^{C_\Sigma}_W$ and $p$ to be the partition on $\partial W$ that designates $\Sigma \times \{a\}$ incoming and $\Sigma \times \{b\}$ outgoing.
Let $\mathring\theta_t$ be the germ of the translation map (\ref{eq:translation}) and set $\mathring\omega_t := \mathring\theta_t^*\omega$ and $\mathring\omega_t' :=\mathring\theta_t^*\omega'$.
Define
\begin{equation}
\begin{split}
\mathcal X_{i} := (\mathring U^{C_\Sigma}_{X_{i}}, p,  \mathring\theta_{s_{i}}, \mathring\theta_{s_{i+1}})\\
\mathcal X_i' := (\mathring U^{C_\Sigma}_{X_i'}, p,  \mathring\theta_{s_i'}, \mathring\theta_{s_{i+1}'})
\end{split}
\quad\quad\quad
\begin{split}
\mathcal Y_i := (\mathring U^{C_\Sigma}_{Y_{i}}, p, \mathring\theta_{t_{i+1}}, \mathring\theta_{t_{i}})\\
\mathcal Y_{i}' := (\mathring U^{C_\Sigma}_{Y_{i}'}, p, \mathring\theta_{t_{i+1}'}, \mathring\theta_{t_{i}}')
\end{split}
\end{equation}
from which we form the $\Dens_\C$-bordisms
\begin{equation}
\begin{split}
(\mathcal X_{i}, \mathring\omega_{X_{i}}): (\mathring\Sigma^+,\mathring\omega_{s_{i}}) &\rightsquigarrow (\mathring\Sigma^+,\mathring\omega_{s_{i+1}})\\
(\mathcal X_{i}', \mathring\omega_{X_{i}'}'): (\mathring\Sigma^+,\mathring\omega'_{s_{i}'}) &\rightsquigarrow (\mathring\Sigma^+,\mathring\omega'_{s_{i+1}'})
\end{split}
\quad\quad\quad
\begin{split}
(\mathcal Y_{i}, \mathring\omega_{Y_{i}}): (\mathring\Sigma^+,\mathring\omega_{t_{i+1}}) &\rightsquigarrow (\mathring\Sigma^+,\mathring\omega_{t_{i}})\\
(\mathcal Y_{i}', \mathring\omega_{Y_{i}'}'): (\mathring\Sigma^+,\mathring\omega'_{t_{i+1}'}) &\rightsquigarrow (\mathring\Sigma^+,\mathring\omega'_{t_{i}'})
\end{split}
\end{equation}
Choose cylindrical $\Dens_\C$-bordisms 
\begin{equation}
\begin{split}
\mathcal A_{i}: (\mathring \Sigma^+,\mathring\omega_{s_{i}}) &\rightsquigarrow (\mathring \Sigma^+,\mathring\omega'_{s_i'})\\
 \mathcal A_{i}': (\mathring \Sigma^+,\mathring\omega'_{s_{i}'}) &\rightsquigarrow (\mathring \Sigma^+,\mathring\omega_{s_{i+1}})
\end{split}
\quad\quad
\begin{split}
\mathcal B_{i}: (\mathring \Sigma^+,\mathring\omega'_{t_{i}'}) &\rightsquigarrow (\mathring \Sigma^+,\mathring\omega_{t_{i}})\\
\mathcal B_{i}': (\mathring \Sigma^+,\mathring\omega_{t_{i+1}}) &\rightsquigarrow (\mathring \Sigma^+,\mathring\omega'_{t_{i}'})
\end{split}
\end{equation}
with volumes 
\begin{equation}
\begin{split}
\gamma_{\omega'}(\tilde s_{i}') - \gamma_{\omega}(s_{i})\\
\gamma_\omega(s_{i+1}) - \gamma_{\omega'}(s_{i}')
\end{split}
\quad\quad\quad
\begin{split}
\gamma_\omega(t_{i}) - \gamma_{\omega'}(s_{i}')\\
\gamma_{\omega'}(t_{i}') - \gamma_{\omega}(t_{i+1})
\end{split}
\end{equation}
respectively.

By Theorem \ref{thm:vol dependence thm}, 
\begin{equation}
\begin{split}
Z(\mathcal X_{i}, \mathring\omega_{X_{i}})= Z(\mathcal A_{i}') \circ Z(\mathcal A_{i})\\
Z(\mathcal X_{i}', \mathring\omega_{X_{i}'}')= Z(\mathcal A_{i+1}) \circ Z(\mathcal A_{i}')
\end{split}
\quad\quad\quad
\begin{split}
Z(\mathcal Y_{i}, \mathring\omega_{Y_{i}})=Z(\mathcal B_{i}) \circ Z(\mathcal B_{i}')\\
Z(\mathcal Y_{i}', \mathring\omega_{Y_{i}'}')=Z(\mathcal B_{i}') \circ Z(\mathcal B_{i+1})\\
\end{split}
\end{equation}
There are direct and inverse systems
\begin{equation}\label{eq:AB system}
\check E_{\mathring\omega_{s_0}} \xrightarrow{\check Z(\mathcal A_0)} \check E_{\mathring\omega'_{s_0'}} \xrightarrow{\check Z(\mathcal A_0')} \check E_{\mathring\omega_{s_1}} \raw \cdots \raw \hat E_{\mathring\omega_{t_1}} \xrightarrow{\hat Z(\mathcal B_0')} \hat E_{\mathring\omega'_{t_0'}} \xrightarrow{\hat Z(\mathcal B_0)} \hat E_{\mathring\omega_{t_0}}
\end{equation}
Let $\mathcal C_\omega, \mathcal D_\omega$ be the direct and inverse cylindrical systems of $\omega$ and similarly let $\mathcal C_{\omega'}, \mathcal D_{\omega'}$ be the cylindrical systems of $\omega'$ (cf. Definition \ref{def:cylindrical systems}).
Let $\tilde{\mathcal C}_\omega, \tilde{\mathcal D}_\omega$ be the direct and inverse systems consisting of the unprimed terms in (\ref{eq:AB system}) and $\tilde{\mathcal C}_{\omega'}, \tilde{\mathcal D}_{\omega'}$ the direct and inverse systems consisting of the primed terms. Then $\tilde{\mathcal C}_\omega \subset \mathcal C_\omega$, $\tilde{\mathcal C}_{\omega'} \subset \mathcal C_{\omega'}$ are cofinal subsystems and $\tilde{\mathcal D}_{\omega} \subset \mathcal D_{\omega}, \tilde{\mathcal D}_{\omega'} \subset \mathcal D_{\omega'}$ are final subsystems.

There is a sequence of isomorphisms
\begin{equation}\label{eq:iso chain 2}
\begin{tikzcd}[column sep = small]
\check E_{\mathring\omega} \arrow[d,hook] \arrow[r, "\cong"]
   & \colim_{\mathcal C_\omega} \check E_{\mathring\omega_s} \arrow[r, "\cong"] \arrow[d,hook]
   & \colim_{\tilde{\mathcal C}_{\omega}} \check E_{\mathring \omega_{s_{2i}}} \arrow[r,"\cong"] \arrow[d,hook]
   & \colim_{\tilde{\mathcal C}_{\omega'}} \check E_{\mathring \omega_{s_{2i+1}}'} \arrow[r,"\cong"] \arrow[d,hook]
   & \colim_{\mathcal C_{\omega'}} \arrow[r,"\cong"] \check E_{\mathring \omega_s'} \arrow[d,hook]
   & \check E_{\mathring\omega'} \arrow[d,hook]\\
\hat E_{\mathring\omega} \arrow[r, "\cong"]
   & \invlim_{\mathcal D_\omega} \hat E_{\mathring \omega_t} \arrow[r, "\cong"]
   & \invlim_{\tilde{\mathcal D}_{\omega}} \hat E_{\mathring \omega_{t_{2i}}} \arrow[r,"\cong"]
   & \invlim_{\tilde{\mathcal D}_{\omega'}} \hat E_{\mathring \omega_{t_{2i+1}}'} \arrow[r,"\cong"]
   & \invlim_{\mathcal D_{\omega'}} \hat E_{\mathring \omega_t'} \arrow[r,"\cong"]
   & \hat E_{\mathring\omega'} 
\end{tikzcd}
\end{equation}
where the first and last are implied by coherence, the second and fourth are canonical isomorphisms implied by the universal properties of limits and colimits applied to final and cofinal subsequences, and the middle isomorphism results from applying the universal properties to (\ref{eq:AB system}). The composition of these isomorphisms gives what was desired.
\end{proof}

\begin{remark}
We can interpret this proposition as giving a canonical trivialization of the bundle (\ref{eq:nuclear pair family}) for a VFT.
\end{remark}

\begin{proposition}\label{prop:canonical iso is functorial 0}
Let $\mathring f \in \Diff(\mathring \Sigma^+)$ be a germ of a diffeomorphism satisfying $\mathring \omega = \mathring f^*\mathring\omega'$. Then $Z(\mathring f)$ is the canonical isomorphism (\ref{eq:canonical density germ isos}).
\end{proposition}
\begin{proof}
Let $f \in \Diff(C_\Sigma)$ be a diffeomorphism of the cylinder whose germ at $\mathring \Sigma$ is $\mathring f$ and let $\omega, \omega' \in \Dens_\C(C_\Sigma)$ be allowable densities restricting to $\mathring\omega,\mathring\omega'$ on $\mathring \Sigma$. Let $\gamma_\omega,\gamma_{\omega'}$ be the signed volume curves of $\omega,\omega'$.
As in Subsection \ref{subsection: action of diffeos}, let 
\begin{equation}
\begin{split}
m(t) := \min_{x \in \Sigma_t} pr_2(f(x))\\
M(t) := \max_{x\in\Sigma_t} pr_2(f(x)).
\end{split}
\end{equation}
There exist monotone increasing sequences $\{s_i\}, \{s_i'\} \subset \R_{<0}$ and monotone decreasing sequences $\{t_i\}, \{t_i'\} \subset \R_{>0}$ converging to 0 satisfying
\begin{equation}
M(s_i) < s_i'< m(s_{i+1}) \quad \text{and} \quad M(t_{i+1})< t_i' < m(t_{i})
\end{equation}
and $\Dens_\C$-bordisms
\begin{equation}
\begin{split}
(\mathcal X_{s_i,s_i'}, \mathring \omega_X): (\mathring \Sigma^+, \mathring\omega_{s_i}) \rightsquigarrow (\mathring \Sigma^+, \mathring \omega'_{s_i'})\\
(\mathcal Y_{s_i',s_{i+1}}, \mathring \omega_Y): (\mathring \Sigma^+, \mathring\omega'_{s_i'}) \rightsquigarrow (\mathring \Sigma^+, \mathring\omega_{s_{i+1}})
\end{split}
\quad\quad
\begin{split}
(\mathcal X_{t_{i+1},t_i'}, \mathring \omega_X): (\mathring \Sigma^+, \mathring\omega_{t_{i+1}}) \rightsquigarrow (\mathring \Sigma^+, \mathring \omega'_{t_i'})\\
(\mathcal Y_{t_i',t_{i}}, \mathring \omega_Y): (\mathring \Sigma^+, \mathring\omega'_{t_i'}) \rightsquigarrow (\mathring \Sigma^+, \mathring\omega_{t_{i}})
\end{split}
\end{equation}
which have volumes
\begin{equation}
\begin{split}
\gamma_{\omega'}(s_i') - \gamma_{\omega}(s_i)\\
\gamma_\omega(s_{i+1}) - \gamma_{\omega'}(s_i')
\end{split}
\quad\quad
\begin{split}
\gamma_{\omega'}(t_i') - \gamma_{\omega}(t_{i+1})\\
\gamma_\omega(t_{i}) - \gamma_{\omega'}(t_i').
\end{split}
\end{equation}
respectively. Applying the VFT $Z$ gives the interwoven system (\ref{eq:AB system}) which shows that the sequences of isomorphisms (\ref{eq:iso chain}) and (\ref{eq:iso chain 2}) are the same.
\end{proof}

\begin{proposition}\label{prop:canonical iso is functorial}
Let $\mathcal X:= (\mathring U^N_X, p, \mathring\theta_0,\mathring\theta_1)$ be a bordism from $\mathring\Sigma_0^+$ to $\mathring\Sigma_1^+$ and let
\begin{equation}
\begin{split}
(\mathcal X, \mathring\omega_X): (\mathring\Sigma_0,\mathring\omega_0) \rightsquigarrow (\mathring\Sigma_1,\mathring\omega_1)\\
(\mathcal X, \mathring\omega_X'): (\mathring\Sigma_0, \mathring\omega_0') \rightsquigarrow (\mathring\Sigma_1, \mathring\omega_1')
\end{split}
\end{equation}
be $\Dens_\C$-bordisms satisfying $\int_X \mathring\omega_X = \int_X \mathring\omega_X'$. Then the square
\begin{equation}
\begin{tikzcd}
\hat E_{\mathring\omega_0} \arrow[d, "\cong"rot, "\hat\alpha_0"] \arrow[r, "Z(\mathcal X{,} \mathring\omega_X)"] & \check E_{\mathring\omega_1} \arrow[d, "\cong"rot, "\check\alpha_1"] \\
\hat E_{\mathring\omega_0'} \arrow[r, "Z(\mathcal X{,}\mathring\omega_X')"'] & \check E_{\mathring\omega_1'}
\end{tikzcd}
\end{equation}
commutes,
where $\hat\alpha_0, \check\alpha_1$ are the canonical isomorphisms 
implied by Proposition \ref{prop:canonical NP isos}.
\end{proposition}

\begin{proof}
If $\Sigma_0, \Sigma_1$ are both empty then the result is implied by Theorem \ref{thm:vol dependence thm}. Suppose $\Sigma_0$ is nonempty and let $\omega_N, \omega_N' \in \Dens_\C(N)$ be densities on $N$ whose restrictions to $\mathring U^N_X$ are $\mathring\omega_X, \mathring\omega_X'$. Let $\theta_i: C_{\Sigma_i} \mono N$ be an embedding onto a cylindrical neighborhood of $\partial X_i \subset N$ whose germ is $\mathring\theta_i$. Without loss of generality we assume that $N = X \cup \Ima(\theta_0) \cup \Ima(\theta_1)$. Set $\omega_i := \theta_i^*\omega_N$ and $\omega_i' := \theta_i^*\omega_N'$ which are densities on the cylinder $C_{\Sigma_i}$ and let $\mathring \omega_{i,t} := \mathring\theta_t^* \omega_i$ , $\mathring\omega_{i,t}' := \mathring\theta_t^*\omega_i'$. 

As in the previous proof, there exist sequences $\{t_i\} \subset \R_{>0}$ and $\{s_i\} \subset \R_{<0}$ converging to 0 and $\Dens_\C$-bordisms
\begin{equation}
\begin{split}
\mathcal B_{2i}: (\mathring \Sigma^+,\mathring\omega'_{0,t_{2i+1}}) &\rightsquigarrow (\mathring \Sigma^+,\mathring\omega_{0,t_{2i}})\\
\mathcal B_{2i+1}: (\mathring \Sigma^+,\mathring\omega_{0,t_{2i+2}}) &\rightsquigarrow (\mathring \Sigma^+,\mathring\omega'_{0,t_{2i+1}})
\end{split}
\quad\quad
\begin{split}
\mathcal A_{2i}: (\mathring \Sigma^+,\mathring\omega_{1,s_{2i}}) &\rightsquigarrow (\mathring \Sigma^+,\mathring\omega'_{1,s_{2i+1}})\\
 \mathcal A_{2i+1}: (\mathring \Sigma^+,\mathring\omega'_{1,s_{2i+1}}) &\rightsquigarrow (\mathring \Sigma^+,\mathring\omega_{1,s_{2i+2}})
\end{split}
\end{equation}
forming a system of maps
\begin{equation}
\begin{tikzcd}
\cdots \hat E_{\mathring \omega_{0,t_2}} 
   \arrow[r,"\hat Z(\mathcal B_1)"] 
& \hat E_{\mathring\omega_{0,t_1}'} 
   \arrow[r, "\hat Z(\mathcal B_0)"]
& \hat E_{\mathring\omega_{0,t_0}}
   \arrow[rr, "Z(\mathcal X)"]
& & \check E_{\mathring\omega_{1,s_0}}
   \arrow[r, "\check Z(\mathcal A_0)"]
& \check E_{\mathring \omega_{1,s_1}'}
   \arrow[r, "\check Z(\mathcal A_1)"]
& \check E_{\mathring \omega_{1,s_2}} \cdots
\end{tikzcd} 
\end{equation}
such that the even terms are final and cofinal subsequences of $\mathcal D_{\omega_0}, \mathcal C_{\omega_1}$, and the odd terms are final and cofinal subsequences of $\mathcal D_{\omega_0'}$ and $\mathcal C_{\omega_1'}$. Here $\mathcal X: (\mathring\Sigma_0^+, \mathring\omega_{0,t_0}) \rightsquigarrow (\mathring\Sigma_1^+, \mathring\omega_{1,s_0})$ is any $\Dens_\C$ bordism of type $X$ with volume $\int_{X_{t_0,s_0}} \mathring\omega_X$, where 
\begin{equation}
X_{t_0,s_0} = X \backslash \left( \theta_0([0, t_0]) \cup \theta_1([s_0, 0])\right)
\end{equation}
Taking limits and colimits of the odd and even subsequences and applying coherence and the universal properties gives the result.
\end{proof}

\begin{definition}\label{def:semigroup bord}
Let $(S,+)$ be a commutative semigroup. Define $\Bord^s_{n,n-1}(S)$ to be the semicategory with objects consisting of closed $n-1$-manifolds and morphisms consisting of pairs $(X,s)$ of a bordism $X$ with $m:= \#\pi_0(X)$ connected components each labeled by an element of $S$ which we represent by the tuple $s:= (s_1,...,s_m) \in S^m$. Composition of morphisms is defined by
\begin{equation}
(X,s) \circ (Y,t) := (X\circ Y, s\oplus t)
\end{equation}
where $X \circ Y$ is composition of ordinary bordisms and $s\oplus t \in S^{\#\pi_0(X\circ Y)}$ is the tuple obtained from composing labels whenever they are part of the same connected component in $X \circ Y$. The symmetric monoidal product is disjoint union. We define $\Bord_{n,n-1}(S)$ to be the category obtained by adding identity morphisms to $\Bord^s_{n,n-1}(S)$.
\end{definition}

\begin{remark}
$\Bord_{n,n-1}(S) = \Bord_{n,n-1}^s(S)$ if and only if $S$ is a monoid, i.e. has a unit.
\end{remark}

We will now set $S:= (\C_{>0},+)$ to be the semigroup of complex numbers with positive real part.

\begin{definition}\label{def:holo functor def}
Let $F : \Bord_{n,n-1}(\C_{>0}) \raw \mathcal{NP}$ be a nuclear symmetric monoidal functor and let $\check E_\Sigma \mono \hat E_\Sigma$ be the nuclear pair assigned to a closed $n-1$-manifold $\Sigma$. We say $F$ is \emph{holomorphic} if for every bordism $X: \Sigma_0 \rightsquigarrow \Sigma_1$, the map 
\begin{equation}\label{eq:Z_X holo}
\begin{split}
F_X: (\C_{>0})^{\#\pi_0(X)} &\raw \Hom(\hat E_{\Sigma_0}, \check E_{\Sigma_1})\\
s & \mapsto Z(X,s)
\end{split}
\end{equation}
is holomorphic.
\end{definition}

\begin{construction}\label{V construction}
For each closed $n-1$-manifold $\Sigma$, pick a positive $n-1$-dimensional density $\nu_\Sigma$ of total volume 1 on $\Sigma$ and let $\omega_\Sigma := \nu_\Sigma \wedge dt \in \Dens_\R(C_\Sigma)$. We note that its germ $\mathring\omega_\Sigma$ at $\mathring \Sigma$ is time symmetric, i.e. $\tau_{\Dens_\C}(\mathring \Sigma^+, \mathring\omega_\Sigma) = (\mathring \Sigma^+, \mathring\omega_\Sigma)$ is a $\tau_{\Dens_\C}$-fixed point in $\Bord_{n,n-1}^s(\Dens_\C)$.

Let $\check E_\Sigma \mono \hat E_\Sigma$ be the nuclear pair assigned to $\mathring\omega_\Sigma$ by $Z$.
Define a functor
\begin{equation}
V: \Bord_{n,n-1}(\C_{>0}) \raw \mathcal{NP}^{h,nuc}
\end{equation}
that sends an object $\Sigma$ to the nuclear pair $\check E_\Sigma \mono \hat E_\Sigma$ and a bordism $(X,s): \Sigma_0 \rightsquigarrow \Sigma_1$ to the nuclear morphism
\begin{equation}\label{eq:Z bar morphism}
\begin{tikzcd}
\check E_{\Sigma_0} \arrow[d, "\check{V}(X{,}s)"'] \arrow[r,hook] & \hat E_{\Sigma_0} \arrow[d, "\hat V(x{,}s)"] \arrow[ld] \\
\check E_{\Sigma_1} \arrow[r,hook] & \hat E_{\Sigma_1}
\end{tikzcd}
\end{equation}
defined by
\begin{equation}
\check V(X,s) := \check Z(\mathcal X, \mathring\omega_X) \quad \text{and} \quad \hat V(X,s) := \hat Z(\mathcal X, \mathring\omega_X)
\end{equation}
where 
\begin{equation}
(\mathcal X, \mathring\omega_X): (\mathring\Sigma_0, \mathring\omega_{\Sigma_0}) \rightsquigarrow (\mathring\Sigma_1,\mathring\omega_{\Sigma_1})
\end{equation}
is any $\Dens_\C$-bordism of type $X$ whose total volume on each connected component of $X$ is given by $s \in (\C_{>0})^{\#\pi_0(X)}$. 
\end{construction}

\begin{remark}
$V$ is well defined because by Theorem \ref{thm:vol dependence thm} any choice of $\Dens_\C$-bordism $(\mathcal X, \mathring\omega_X)$ of total volume $s$ gives the same map (\ref{eq:Z bar morphism}).
\end{remark}

\begin{remark}\label{rem:V is holo}
By Theorem \ref{thm:vol dependence thm}, $V$ is holomorphic (cf. Definition \ref{def:holo functor def}).
\end{remark}

\begin{lemma}\label{lem:int holo}
Let $\mathcal X := (\mathring U^N_X, p, \mathring \theta_0, \mathring \theta_1)$ be a bordism from $\mathring \Sigma_0$ to $\mathring\Sigma_1$. Then the integration map
\begin{equation}\label{eq:intX map}
\int_X : \Dens_\C(\mathring U^N_X) \raw \C_{>0}^{\#\pi_0(X)}
\end{equation}
is holomorphic.
\end{lemma}
\begin{proof}
The map (\ref{eq:intX map}) is a composition
\begin{equation}
\Dens_\C(\mathring U^N_X) \raw \Dens_\C(X) \xrightarrow{\int_X'} \C_{>0}^{\#\pi_0(X)}
\end{equation}
where the first map is the restriction map and the second is the map that integrates a density on each connected component of $X$. That the restriction map is holomorphic follows from the definitions (cf. Definition \ref{def: holo def 3}). 

There is an extension of $\int_X'$ to a complex linear map
\begin{equation}
\int_X': C^\infty(\Lambda^n(T^*X) \otimes \mathfrak o(X) \otimes\C) \raw \C^{\#\pi_0(X)}
\end{equation}
out of the vector space of complex top dimensional forms twisted by the orientation line bundle on $X$. As $\Dens_\C(X)$ is an open subset of this complex vector space, this shows holomorphicity. 
\end{proof}

Define a functor
\begin{equation}\label{eq:int functor}
\int: \Bord_{n,n-1}(\Dens_\C) \raw \Bord_{n,n-1}(\C_{>0})
\end{equation}
that sends an object $(\mathring \Sigma, \mathring\omega)$ to $\Sigma$, isomorphisms to the identity morphism, and a $\Dens_\C$-bordism $(\mathcal X, \mathring\omega_X): (\mathring\Sigma_0, \mathring\omega_0) \rightsquigarrow (\mathring\Sigma_1,\mathring\omega_1)$ to $(X,s): \Sigma_0 \rightsquigarrow \Sigma_1$ where $s := (\int_{X_j} \mathring\omega_X) \in (\C_{>0})^{\#\pi_0(X)}$ is the tuple of total volumes on each connected component $X_j$ of $X$.

\begin{proposition}\label{prop:proposition: Z' field theory}
The composition
\begin{equation}
Z': \Bord_{n,n-1}(\Dens_\C) \xrightarrow{\int} \Bord_{n,n-1}(\C_{>0}) \xrightarrow{V} \mathcal{NP}^{h,nuc}
\end{equation}
is a $\Dens_\C$-field theory.
\end{proposition}

\begin{proof}
We check the following criteria:
\begin{enumerate}
\item \emph{$Z'$ is nuclear}: by construction.

\item \emph{$Z'$ is holomorphic}: The bundle (\ref{eq:pullback np}) is trivial with fiber $V(\Sigma) = \check E_\Sigma \mono \hat E_\Sigma$ and in particular it is holomorphic.

Let $\mathcal X:= (\mathring U^N_X, p, \mathring\theta_0, \mathring\theta_1)$ be a bordism from $\mathring \Sigma_0$ to $\mathring\Sigma_1$. By Lemma \ref{lem:int holo}, the composition
\begin{equation}
\Dens_\C(\mathring U^N_X) \xrightarrow{\int_X} \C_{>0}^{\#\pi_0(X)} \xrightarrow{V_X} \Hom(\hat E_{\Sigma_0}, \check E_{\Sigma_1})
\end{equation}
is holomorphic. This is equivalent to a holomorphic section of the trivial bundle
\begin{equation}
\begin{tikzcd}
\underline{\Hom(\hat E_{\Sigma_0}, \check E_{\Sigma_1})} \arrow[d]\\
\Dens_\C(\mathring U^N_X) \arrow[u, bend right, "Z'_{\mathcal X}"']
\end{tikzcd}
\end{equation}
which is (\ref{eq:hom family2}) for the field theory $Z'$ and shows holomorphicity.

\item \emph{$Z'$ is coherent}: 
Let $(\mathring\Sigma, \mathring\omega)$ be an object in $\Bord_{n,n-1}(\Dens_\C)$ and let $\omega \in \Dens_\C(C_\Sigma)$ be a density with germ $\mathring \omega$ at $\mathring\Sigma$. Let $\gamma_\omega: \R \raw \C^{\#\pi_0(\Sigma)}$ be the signed volume curve of $\omega$ (cf. Definition \ref{def:signed volume curve}).

The direct and inverse cylindrical systems $\mathcal C_\omega, \mathcal D_\omega$ (cf. Definition \ref{def:cylindrical systems}) are given by the systems of maps
\begin{equation}
\begin{split}
\check E_\Sigma \xrightarrow{\check V(\gamma_\omega(s')-\gamma_\omega(s))} \check E_\Sigma\\
\hat E_\Sigma \xrightarrow{\hat V(\gamma_\omega(t') - \gamma_\omega(t))} \hat E_\Sigma
\end{split}
\end{equation}
and the systems $\mathcal C_{\omega_\Sigma}, \mathcal D_{\omega_\Sigma}$ are given by
\begin{equation}
\begin{split}
\check E_\Sigma \xrightarrow{\check V(s'-s)} \check E_\Sigma\\
\hat E_\Sigma \xrightarrow{\hat V(t'-t)} \hat E_\Sigma
\end{split}
\end{equation}
for all $s < s' < 0 < t < t'$.

By universal property, there are unique maps making the diagrams
\begin{equation}
\begin{tikzcd}
\check E_\Sigma \arrow[r] \arrow[rd, "\check V(-\gamma(s))"'] & \check E_{\mathcal C_\omega} \arrow[d, dashed, "\check \varphi"]\\
& \check E_\Sigma
\end{tikzcd}
\quad \quad
\begin{tikzcd}
\hat E_{\mathcal D_\omega} \arrow[r] & \hat E_\Sigma \\
\hat E_\Sigma \arrow[u, dashed, "\hat \varphi"] \arrow[ur, "\hat V(\gamma(s))"']&
\end{tikzcd}
\end{equation}
commute. Coherence follows from noticing that $\check \varphi, \hat \varphi$ have a factorization into isomorphisms
\begin{equation}
\begin{tikzcd}
\check E_{\mathcal C_\omega} \arrow[dd,bend right, "\check \varphi"'] \arrow[r] \arrow[d, no head, "\cong"rot] & \hat E_{\mathcal D_\omega}\arrow[d, no head, "\cong"rot]\\
\check E_{\mathcal C_{\omega_\Sigma}} \arrow[r] \arrow[d, no head, "\cong"rot] & \hat E_{\mathcal D_{\omega_\Sigma}} \arrow[d, no head, "\cong"rot]\\
\check E_\Sigma  \arrow[r,hook] & \hat E_\Sigma \arrow[uu, bend right, "\hat\varphi"']
\end{tikzcd}
\end{equation}
where the top isomorphism square is from Proposition \ref{prop:canonical NP isos} and the bottom isomorphism square is from coherence of $Z$.
\end{enumerate}
\end{proof}

\begin{proposition}\label{prop:Z nat iso Z'}
$Z$ is naturally isomorphic to $Z'$.
\end{proposition}

\begin{proof}
For any object $(\mathring\Sigma, \mathring \omega)$ in $\Bord_{n,n-1}^s(\Dens_\C)$, let $\check E_{\mathring\omega} \mono \hat E_{\mathring\omega}$ denote the nuclear pair assigned to it by $Z$. Define $\alpha(\mathring\Sigma, \mathring\omega)$ to be the canonical isomorphism 
\begin{equation}
\begin{tikzcd}
\check E_{\mathring\omega} \arrow[r,hook] \arrow[d, "\cong"rot] &\hat E_{\mathring\omega} \arrow[d, "\cong" rot]\\
\check E_\Sigma \arrow[r,hook] &\hat E_\Sigma
\end{tikzcd}
\end{equation}
given by Proposition \ref{prop:canonical NP isos}. Then Propositions \ref{prop:canonical iso is functorial 0} and \ref{prop:canonical iso is functorial} show that this defines a natural isomorphism $\alpha: Z \implies Z'$.
\end{proof}


\subsection{The Lorentzian Limit}\label{section lorentzian limit}

Let $V: \Bord_{n,n-1}(\C_{>0}) \raw \mathcal{NP}$ be the functor from Construction \ref{V construction}.
Let $X: \Sigma_0 \rightsquigarrow \Sigma_1$ be a bordism of volume $\tau \in (\C_{>0})^{\#\pi_0(X)}$ between closed $n-1$-manifolds $\Sigma_0$ and $\Sigma_1$.
 The theory $V$ assigns to $(X,\tau)$ a morphism of nuclear pairs
\begin{equation}
\begin{tikzcd}
\check E_{\Sigma_0} \arrow[d, "\check V_X(\tau)"'] \arrow[r, hook, "\iota_{\Sigma_0}"] 
&\hat E_{\Sigma_0} \arrow[d, "\hat V_X(\tau)"] \arrow[dl]\\
\check E_{\Sigma_1} \arrow[r,hook, "\iota_{\Sigma_1}"] 
&\hat E_{\Sigma_1}
\end{tikzcd}
\end{equation}

\begin{notation}
Given bordisms $X: \Sigma_0 \rightsquigarrow \Sigma_1$ and $Y: \Sigma_1 \rightsquigarrow \Sigma_2$, let
\begin{equation}
\oplus_{Y\circ X}: (\C_{>0})^{\#\pi_0(Y)} \times (\C_{>0})^{\#\pi_0(X)} \raw (\C_{>0})^{\#\pi_0(Y\circ X)}
\end{equation}
be the function defined by the composition
\begin{equation}
(Y\circ X, s \oplus_{Y\circ X} t) := (Y,s) \circ (X,t)
\end{equation}
of bordisms in $\Bord_{n,n-1}(\C_{\geq 0})$ (cf. Definition \ref{def:semigroup bord}). When the context is clear we will omit the subscript and simply write $s \oplus t$.
\end{notation}

\begin{theorem}\label{thm:UP for L}
Let $X: \Sigma_0 \rightsquigarrow \Sigma_1$ be a bordism and let $\zeta \in \R^{\#\pi_0(X)}$. Then there exist unique maps 
\begin{equation}
\check L(X,i\zeta): \check E_{\Sigma_0} \raw \check E_{\Sigma_1}\quad\quad \text{and} \quad\quad
\hat L(X,i\zeta): \hat E_{\Sigma_0} \raw \hat E_{\Sigma_1}
\end{equation}
such that for all $\tau^{(i)} = (\tau^{(i)}_j) \in (\C_{>0})^{\#\pi_0(\Sigma_i)}$, the diagrams
\begin{equation}\label{eq:L check L hat}
\begin{tikzcd}[column sep = large]
\check E_{\Sigma_0} \arrow[r, "\check V_{\Sigma_0 \times I}(\tau^{(0)})"] \arrow[rd, "\check V_X(i\zeta \oplus \tau^{(0)})"'] & \check E_{\Sigma_0} \arrow[d, "\check{L}(X {,} i\zeta)"]\\
 & \check E_{\Sigma_1}
\end{tikzcd}
\quad\text{and}\quad
\begin{tikzcd}[column sep = large]
\hat E_{\Sigma_0} \arrow[d, "\hat L(X{,}i\zeta)"'] \arrow[rd, "\hat V_X(i\zeta \oplus \tau^{(1)})"]\\
 \hat E_{\Sigma_1} \arrow[r, "\hat V_{\Sigma_1 \times I}(\tau^{(1)})"'] & \hat E_{\Sigma_1}
\end{tikzcd}
\end{equation}
commute when 
\begin{equation}
\Sigma_0 \neq \emptyset \quad\quad\text{and}\quad\quad \Sigma_1 \neq \emptyset
\end{equation}
respectively.
\end{theorem}
\begin{proof}
We assume that both $\Sigma_0$ and $\Sigma_1$ are nonempty.
For $i=0,1$ let 
\begin{equation}\label{eq:gamma i curve}
\gamma^{(i)}: \R \raw \C^{\#\pi_0(\Sigma_i)}
\end{equation}
be curves satisfying $\frac{d}{dt}\gamma^{(i)}(t) \in (\C_{>0})^{\#\pi_0(\Sigma_i)}$ and $\gamma^{(i)}(0) = 0$.

Let $\mathcal C_{\gamma^{(0)}} := \{\check E_s := \check E_{\Sigma_0} \st s < 0\}$ be the direct system consisting of maps
\begin{equation}
\check E_s \xrightarrow{\check V_{\Sigma_0 \times I}(\gamma^{(0)}(s') - \gamma^{(0)}(s))} \check E_{s'}
\end{equation}
for all $s < s' < 0$. Similarly, let $\mathcal D_{\gamma^{(1)}} := \{\hat E_t := \hat E_{\Sigma_1} \st t > 0\}$ be the inverse system consisting of maps
\begin{equation}
\hat E_t \xrightarrow{\hat V_{\Sigma_1 \times I}(\gamma^{(1)}(t') - \gamma^{(1)}(t))} \hat E_{t'}
\end{equation}
for all $0 < t < t'$.

As in the proof of coherence in Proposition \ref{prop:proposition: Z' field theory}, there are canonical isomorphisms
\begin{equation}\label{eq:colim gamma0}
\check E_{\Sigma_0} \cong \colim_{\mathcal C_{\gamma^{(0)}}} \check E_s
\end{equation}
\begin{equation}\label{eq:lim gamma1}
\hat E_{\Sigma_1} \cong \invlim_{\mathcal D_{\gamma^{(1)}}} \hat E_t.
\end{equation} 

The maps
\begin{equation}
\check E_s \xrightarrow{\check V_{X}(i\zeta \oplus (-\gamma^{(0)}(s)))} \check E_{\Sigma_1}
\end{equation}
form a cocone to $\mathcal C_{\gamma^{(0)}}$. By the universal property of the colimit and the identification (\ref{eq:colim gamma0}), there exists a unique map 
\begin{equation}
\check L(X, i\zeta):\check E_{\Sigma_0} \raw \check E_{\Sigma_1}
\end{equation}
making the diagrams
\begin{equation}\label{eq:L check}
\begin{tikzcd}[column sep = huge]
\check E_s \arrow[r, "\check V_{\Sigma_0 \times I}(-\gamma^{(0)}(s))"] \arrow[rd, "\check V_X(i\zeta \oplus (-\gamma^{(0)}(s)))"'] & \check E_{\Sigma_0} \arrow[d, dashed, "\check{L}(X {,} i\zeta)"]\\
 & \check E_{\Sigma_1}
\end{tikzcd}
\end{equation}
commute for all $s < 0$.

Similarly, the maps 
\begin{equation}
\hat E_{\Sigma_0} \xrightarrow{\hat V_{X}(i\zeta \oplus \gamma^{(1)}(t))} \hat E_{t}
\end{equation}
form a cone to $\mathcal D_{\gamma^{(1)}}.$ By the universal property of the limit and the identification (\ref{eq:lim gamma1}) there exists a unique map 
\begin{equation}
\hat L(X, i\zeta) : \hat E_{\Sigma_0} \raw \hat E_{\Sigma_1}
\end{equation}
making the diagram
\begin{equation}\label{eq:L hat}
\begin{tikzcd}[column sep = huge]
\hat E_{\Sigma_0} \arrow[d, dashed, "\hat L(X{,}i\zeta)"'] \arrow[rd, "\hat V_X(i\zeta \oplus \gamma^{(1)}(t))"]\\
 \hat E_{\Sigma_1} \arrow[r, "\hat V_{\Sigma_1\times I}(\gamma^{(1)}(t))"'] & \hat E_{t}
\end{tikzcd}
\end{equation}
commute for all $t>0$.

For all $\tau^{(0)} \in (\C_{>0})^{\#\pi_0(\Sigma_0)}$, continuity of $\gamma^{(0)}$ implies that there exists $s < 0$ such that $\tau^{(0)} + \gamma^{(0)}(s) \in (\C_{>0})^{\#\pi_0(\Sigma_0)}$. 
Precomposing diagram (\ref{eq:L check}) by the map
\begin{equation}
\check V_{\Sigma_0\times I}(\tau^{(0)} +\gamma^{(0)}(s)): \check E_{\Sigma_0} \raw \check E_s
\end{equation} 
shows that the left hand diagram in (\ref{eq:L check L hat}) commutes. Likewise, for each $\tau^{(1)} \in (\C_{>0})^{\#\pi_0(\Sigma_1)}$ continuity of $\gamma^{(1)}$ implies that there exists $t>0$ such that $\tau^{(1)} - \gamma^{(1)}(t)\in (\C_{>0})^{\#\pi_0(\Sigma_0)}$. Postcomposing diagram (\ref{eq:L hat}) by 
\begin{equation}
\hat V_{\Sigma_1\times I}(\tau^{(1)} - \gamma^{(1)}(t)): \hat E_{t} \raw \hat E_{\Sigma_1}
\end{equation}
shows that the right hand diagram in (\ref{eq:L check L hat}) commutes.
\end{proof}
\begin{remark}
The theorem in particular implies that the constructions of $\check L(X,i\zeta),\hat L(X,i\zeta)$ are independent of the choice of right-moving curve $\gamma$.
\end{remark}

\begin{corollary}
Let $\Sigma$ be a closed $n-1$-manifold. Then
$\check L(\Sigma \times I,0) = id_{\check E_\Sigma}$ and $\hat L(\Sigma \times I,0) = id_{\hat E_{\Sigma}}$.
\end{corollary}
\begin{proof}
When $X = \Sigma \times I$ and $\zeta = 0 \in \R^{\#\pi_0(\Sigma)}$, the identity maps make the diagrams (\ref{eq:L check L hat}) commute for all $\tau^{(i)} \in (\C_{>0})^{\#\pi_0(\Sigma)}$.
\end{proof}

\begin{lemma}\label{lem:lorentzian np morphism}
Let $X: \Sigma_0 \rightsquigarrow \Sigma_1$ be a bordism and let $\zeta \in \R$. Then the square
\begin{equation}
\begin{tikzcd}
\check E_{\Sigma_0} \arrow[r,hook] \arrow[d, "\check L(X{,}i\zeta)"'] 
& \hat E_{\Sigma_0} \arrow[d, "\hat L(X{,}i\zeta)"]\\
\check E_{\Sigma_1} \arrow[r,hook] 
&\hat E_{\Sigma_1} 
\end{tikzcd}
\end{equation}
commutes.
\end{lemma}

\begin{proof}
The solid arrows in the diagram
\begin{equation}
\begin{tikzcd}[column sep = huge]
\check E_s \arrow[r, "\check V_{\Sigma_0\times I}(\tau^{(0)})"] \arrow[rd, "\check V_X(i\zeta \oplus \tau^{(0)})"']& \check E_{\Sigma_0} \arrow[d, dashed, "\check L(X{,}i\zeta)"] \arrow[r, hook, "\iota_{\Sigma_0}"]
& \hat E_{\Sigma_0} \arrow[d, dashed, "\hat L(X{,}i\zeta)"'] \arrow[rd, "\hat V_X(i \zeta \oplus \tau^{(1)})"]
& \\
& \check E_{\Sigma_1} \arrow[r,hook,"\iota_{\Sigma_1}"]
& \hat E_{\Sigma_1} \arrow[r, "\hat V_{\Sigma_1 \times I}(\tau^{(1)})"'] 
& \hat E_{\Sigma_1,t}
\end{tikzcd}
\end{equation}
commute for all $\tau^{(i)} \in (\C_{>0})^{\#\pi_0(\Sigma_i)}$.
\end{proof}

Let $(i\R,+)$ be the group of purely imaginary numbers.
Let $\Bord_{n,n-1}^{in}(i\R)$, $\Bord_{n,n-1}^{out}(i\R)$, $\Bord_{n,n-1}^{in\wedge out}(i\R)$ be the subcategories of $\Bord_{n,n-1}(i\R)$ (cf. Definition \ref{def:semigroup bord}) whose morphisms consist of diffeomorphism classes of bordisms with non-empty incoming boundary, non-empty outgoing boundary, or both, respectively. 

\begin{proposition}
The maps $\check L(X,i\zeta), \hat L(X,i\zeta)$ can be assembled into symmetric monoidal functors
\begin{equation}
\begin{split}
\check L: \Bord_{n,n-1}^{in}(i\R) \raw \mathcal{NDF}\\
\hat L: \Bord_{n,n-1}^{out}(i\R) \raw \mathcal{NF}\\
L: \Bord_{n,n-1}^{in\wedge out}(i\R) \raw \mathcal{NP}
\end{split}
\end{equation}
\end{proposition}
\begin{proof}
For $\Sigma$ a closed $n-1$-manifold, define 
\begin{equation}
\begin{split}
\check L(\Sigma) := \check E_\Sigma\\
\hat L (\Sigma) := \hat E_\Sigma\\
L(\Sigma) := \check E_\Sigma \mono \hat E_\Sigma
\end{split}
\end{equation}
which all inherit the property from $Z$ that they send disjoint unions to topological tensor products.
Let $X: \Sigma_0 \rightsquigarrow \Sigma_1$ and $Y: \Sigma_1 \rightsquigarrow \Sigma_2$ be bordisms and let $\zeta \in \R^{\#\pi_0(X)}$, $\xi\in \R^{\#\pi_0(Y)}$. For all $\tau^{(i)} \in (\C_{>0})^{\#\pi_0(\Sigma_i)}$, the diagrams

\begin{equation}
\begin{tikzcd}[column sep = huge]
   \check E_{\Sigma_0}\arrow[rr, "\check V_{\Sigma_0\times I}(\tau^{(0)} + \tau')"] \arrow[rd, "\check V_{X}(\tau^{(0)} \oplus i\zeta)"] \arrow[rrdd, bend right, "\check V_{Y\circ X}(\tau^{(0)} \oplus \tau^{(1)} \oplus i\zeta \oplus i\xi)"']
   & & \check E_{\Sigma_0} \arrow[d, "\check L(X{,}i\zeta)"]\\
   & \check E_{\Sigma_1} \arrow[r, "\check V_{\Sigma_1\times I}(\tau^{(1)})"] \arrow[rd, "\check V_Y(\tau^{(1)} \oplus i\xi)"']
   & \check E_{\Sigma_1} \arrow[d, "\check L(Y{,}i\xi)"]\\
   & & \check E_{\Sigma_2}
\end{tikzcd}
\text{and}
\begin{tikzcd}[column sep = huge]
\hat E_{\Sigma_0} \arrow[rrdd, "\hat V_{Y\circ X}(\tau^{(1)} \oplus \tau^{(2)} \oplus i\zeta \oplus i\xi)", bend left] \arrow[d, "\hat L(X{,}i\zeta)"'] \arrow[rd, "\hat V_{X}(\tau^{(1)} \oplus i\zeta)"] &   &  \\
\hat E_{\Sigma_1} \arrow[d, "\hat L(Y{,}i\xi)"'] \arrow[r, "\hat V_{\Sigma_1 \times I}(\tau^{(1)})"'] & \hat E_{\Sigma_1} \arrow[rd, "\hat V_Y(\tau^{(2)} \oplus i\xi)"']  &  \\
\hat E_{\Sigma_2} \arrow[rr, "\hat V_{\Sigma_2 \times I}(\tau'' + \tau^{(2)})"'] &   & \hat E_{\Sigma_2}
\end{tikzcd}
\end{equation}
commute, where $\tau'$ is any element of $(\C_{>0})^{\#\pi_0(\Sigma_0)}$ satisfying
\begin{equation}
0_X \oplus \tau' = \tau^{(1)} \oplus 0_X
\end{equation}
and $\tau''$ is any element of $(\C_{>0})^{\#\pi_0(\Sigma_2)}$ satisfying
\begin{equation}
\tau'' \oplus 0_Y = 0_Y \oplus \tau^{(1)}.
\end{equation}
By Theorem \ref{thm:UP for L}, $\check L(Y\circ X, i\zeta\oplus i\xi)$ and $\hat L(Y\circ X, i\zeta \oplus i\xi)$ are the unique maps making the outer triangles commute and thus we have
\begin{equation}
\begin{split}
\check L(Y,i\xi) \circ \check L(X,i\zeta) = \check L(Y\circ X, i\zeta \oplus i\xi)\\
\hat L(Y,i\xi) \circ \hat L(X,i\zeta) = \hat L(Y\circ X, i\zeta \oplus i\xi)
\end{split}
\end{equation}
which proves functoriality of $\check L$ and $\hat L$. Functoriality of $L$ follows from applying Lemma \ref{lem:lorentzian np morphism}.
\end{proof}

\begin{definition}
We refer to $L$ as the \emph{Lorentzian limit of $V$}.
\end{definition}

Let
\begin{equation}\label{eq:short distance limit}
\begin{split}
\check L_0: \Bord_{n,n-1}^{in} \raw \mathcal{NDF}\\
\hat L_0: \Bord_{n,n-1}^{out} \raw \mathcal{NF}\\
 L_0: \Bord_{n,n-1}^{in\wedge out} \raw \mathcal{NP}
\end{split}
\end{equation}
be the symmetric monoidal functors obtained from restricting $\check L, \hat L, L$ to the subcategories 
\begin{equation}
\begin{split}
\Bord_{n,n-1}^{in} := \Bord_{n,n-1}^{in}(0) \subset \Bord_{n,n-1}^{in}(i\R)\\
\Bord_{n,n-1}^{out} := \Bord_{n,n-1}^{out}(0) \subset \Bord_{n,n-1}^{out}(i\R)\\
\Bord_{n,n-1}^{in\wedge out} := \Bord_{n,n-1}^{in\wedge out}(0) \subset \Bord_{n,n-1}^{in\wedge out}(i\R)
\end{split}
\end{equation}
The functors (\ref{eq:short distance limit}) are topological field theories partially defined on subcategories of the ordinary bordism category. 
\begin{definition}\label{def:short distance limit}
We refer to $L_0$ as the \emph{short distance topological limit}.
\end{definition}

\begin{remark}
There is a family of VFTs controlled by a parameter $\kappa \in \R_{>0}$ called the \emph{coupling constant}
\begin{equation}
V_\kappa: \Bord_{n,n-1}(\C_{>0}) \raw \mathcal{NP}
\end{equation}
defined by 
\begin{align}
V_\kappa(\Sigma):= V(\Sigma)\\
V_\kappa(X,\tau) := V(X,\kappa \tau)
\end{align}
and the short-distance topological limit is typically studied by taking a limit as $\kappa$ approaches 0. This corresponds to taking the curves (\ref{eq:gamma i curve}) to be
\begin{equation}
\gamma^{(i)}(\kappa) = \kappa\tau
\end{equation}
where $\kappa$ is now a real parameter.
\end{remark}


\subsection{Reflection Positive VFTs}\label{reflection positive vfts}

Let $Z: \Bord_{n,n-1}(\Dens_\C) \raw \mathcal{NP}$ be a reflection positive VFT and let $V: \Bord_{n,n-1}(\C_{>0}) \raw \mathcal{NP}$ be the nuclear symmetric monoidal functor that results when Construction \ref{V construction} is applied to $Z$. Recall that for $\Sigma$ a closed $n-1$-manifold, $V(\Sigma) := Z(\mathring \Sigma, \mathring\omega_\Sigma)$ where $(\mathring\Sigma,\mathring\omega_\Sigma)$ is a $\tau_{\Dens_\C}$ fixed point. Reflection positivity of $Z$ implies that $V(\Sigma)$ is Hermitian and $V$ is in fact a nuclear symmetric monoidal functor
\begin{equation}\label{eq:Vh}
V: \Bord_{n,n-1}(\C_{>0}) \raw \mathcal{NP}^{h,nuc}
\end{equation}
valued in Hermitian nuclear pairs and nuclear morphisms.
Postcomposing (\ref{eq:Vh}) with (\ref{eq:NP^h --> Hilb}) gives a symmetric monoidal functor
\begin{equation}
V^{Hilb}: \Bord_{n,n-1}(\C_{>0}) \raw \Hilb.
\end{equation}
We set 
\begin{equation}
\mathcal H_\Sigma := E^{Hilb}_\Sigma
\end{equation}
to be the Hilbert space assigned to $\Sigma$ by $V^{Hilb}$.

\begin{lemma}
$V^{Hilb}$ is holomorphic, i.e. for all bordisms $X: \Sigma_0 \rightsquigarrow \Sigma_1$ the map
\begin{equation}
\begin{split}
(\C_{>0})^{\#\pi_0(X)} &\raw \Hom(\mathcal H_{\Sigma_0}, \mathcal H_{\Sigma_1})\\
s &\mapsto V^{Hilb}(X,s)
\end{split}
\end{equation}
is holomorphic.
\end{lemma}
\begin{proof}
This follows from holomorphicity of $V$ (cf. Remark \ref{rem:V is holo}).
\end{proof}

\begin{proposition}\label{prop:adjoint}
Let $X: \Sigma_0 \rightsquigarrow \Sigma_1$ be a bordism of volume $s \in (\C_{>0})^{\#\pi_0(X)}$. Then
\begin{equation}
V^{Hilb}(X^*,\bar s) = V^{Hilb}(X,s)^\dagger 
\end{equation}
where $X^*: \Sigma_1 \rightsquigarrow \Sigma_0$ is the bordism $X$ with incoming and outgoing boundaries reversed, $\bar s \in (\C_{>0})^{\#\pi_0(X)}$ is the complex conjugate of $s$, and $V^{Hilb}(X,s)^\dagger$ is the Hermitian adjoint of $V^{Hilb}(X,s): \mathcal H_{\Sigma_0} \raw \mathcal H_{\Sigma_1}$.
\end{proposition}

\begin{proof}
Let $(\mathcal X,\mathring\omega_X): (\mathring \Sigma_0^+, \mathring\omega_{\Sigma_0}) \rightsquigarrow (\mathring \Sigma_1^+, \mathring\omega_{\Sigma_1})$ be a $\Dens_\C$-bordism of volume $s$. Reflection positivity of $Z$ implies that 
\begin{equation}
\begin{split}
V(X^*,\bar s) &= Z(\tau_{\Dens_\C}(\mathcal X, \mathring \omega_X)) \\
&= \tau_{\mathcal{NP}} Z(\mathcal X, \mathring\omega_X)\\
&= \tau_{\mathcal{NP}} V(X,s)
\end{split}
\end{equation}
Applying the $(\tau_{\mathcal{NP}}, \tau_{\Hilb})$-equivariant functor (\ref{eq:NP^h --> Hilb}) then gives
\begin{equation}
V^{Hilb}(X^*,\bar s) = V^{Hilb}(X,s)^\dagger
\end{equation}
\end{proof}

For $\Sigma$ a closed $n-1$-manifold, set 
\begin{equation}
V_\Sigma^{Hilb}(s):= V^{Hilb}(\Sigma \times I, s)
\end{equation}
which is a trace-class operator in $\End(\mathcal H_\Sigma)$ for all $s \in (\C_{>0})^{\#\pi_0(X)}$.

\begin{proposition}\label{prop:hamiltonian}
Let $\Sigma$ be a connected closed $n-1$-manifold. There exists a self-adjoint operator $H_\Sigma \in \End(\mathcal H_\Sigma)$ with discrete spectrum unbounded from above and bounded from below with finite multiplicity such that
\begin{equation}\label{eq:heat formula}
V^{Hilb}_\Sigma(s) = \exp(-sH_\Sigma)
\end{equation}
\end{proposition}

\begin{proof}
For all $s, s' \in \C_{>0}$ we have
\begin{equation}\label{eq:semigroup}
V^{Hilb}_\Sigma(s)V^{Hilb}_\Sigma(s') = V^{Hilb}_\Sigma(s + s') = V^{Hilb}_\Sigma(s') V^{Hilb}_\Sigma(s)
\end{equation}
and when $s' = \bar s$ this becomes
\begin{equation}
V^{Hilb}_\Sigma(s)V^{Hilb}_\Sigma(s)^\dagger = V^{Hilb}_\Sigma(s)^\dagger V^{Hilb}_\Sigma(s)
\end{equation}
by Proposition \ref{prop:adjoint}.
Thus $\{V^{Hilb}_\Sigma(s) \hspace{.1cm} | \hspace{.1cm} s \in \C_{>0} \}$ is a holomorphic family of mutually commuting compact (trace-class) normal operators on $\mathcal H_\Sigma$. By the spectral theorem, there exists an isometry
\begin{equation}
\mathcal H_\Sigma \cong \ell^2(\N)
\end{equation}
under which $V^{Hilb}_\Sigma$ becomes a family of diagonal operators
\begin{equation}
V^{Hilb}_\Sigma(s) = \sum_{i \geq 0} \mu_i(s) \Pi_i
\end{equation}
where $\mu_i: \C_{>0} \raw \C$ is a holomorphic function whose value at $s$ is the $i$th eigenvalue of $V^{Hilb}_\Sigma(s)$, and $\Pi_i^2 = \Pi_i \in \End(\mathcal H_\Sigma)$ is the orthogonal projection onto the $i$th eigenspace, which has finite dimension. The functions $\mu_i$ cannot be identically 0 for otherwise this would imply that $V^{Hilb}_\Sigma(s)$ has a kernel for all $s$ which would imply that the induced inclusion of nuclear spaces $\check E_\Sigma \mono \hat E_\Sigma$ is not dense and this would contradict our assumption that $Z$ is a field theory. 

Relation (\ref{eq:semigroup}) implies 
\begin{equation}\label{eq: mu semigroup}
\mu_i(s + s') = \mu_i(s)\mu_i(s').
\end{equation}
If $\mu_i(s_0) = 0$ for some $s \in \C_{>0}$, then (\ref{eq: mu semigroup}) implies $\mu_i(s) = 0$ for all $s \in \C_{>0}$ satisfying $\Re(s) > \Re(s_0)$. Holomorphicity then implies $\mu_i \equiv 0$, which is not allowed. Thus $\mu_i$ does not vanish anywhere. The relation (\ref{eq: mu semigroup}) also implies that the function
\begin{equation}
f_{i,s}(\tau) := \frac{\mu_i(\tau + s) - \mu_i(\tau)}{s\mu_i(\tau)}
\end{equation}
is constant on $\C_{>0}$. In particular the limit 
\begin{equation}
\lim_{s \raw 0} f_{i,s}(\tau) = \frac{\mu_i'(\tau)}{\mu_i(\tau)}
\end{equation} is a holomorphic function equal to a constant $\lambda_i$ which implies $\mu_i'(\tau) = \lambda_i\mu_i(\tau)$ for all $\tau \in \C_{>0}$. Solving the differential equation gives $\mu_i(\tau) = c_i e^{-\lambda_i \tau}$ and (\ref{eq: mu semigroup}) implies $c_i \in \{0,1\}$. Since $\mu_i$ is not identically zero, we have $c_i = 1$.

By Proposition \ref{prop:adjoint}, $\mu_i(s) \in \R$ when $s \in \R_{>0}$ which implies $\lambda_i \in \R$. The operator $V_\Sigma$ is trace-class which implies that $\{\lambda_i\}$ is unbounded above and bounded below with no accumulation points. Set
\begin{equation}
H_\Sigma := \sum_{i\geq 0} \lambda_i\Pi_i
\end{equation}
which is unbounded above and bounded below.
\end{proof}

When $\Sigma := \bigsqcup_{i=1}^k \Sigma^i$ has $k$ connected components, we set
\begin{equation}
\mathcal H_\Sigma := \bigotimes_{i=1}^k \mathcal H_{\Sigma^i}
\end{equation}
\begin{equation}
\hat H_{\Sigma^i} :=  1 \otimes \cdots \otimes H_{\Sigma^i} \otimes \cdots \otimes 1
\end{equation}
\begin{equation}
H_{\Sigma} :=  \sum_{i=1}^k \hat H_{\Sigma^i}
\end{equation}
and
\begin{equation}
\begin{split}
\exp(-sH_\Sigma) := \exp(-\sum_{i=1}^k s_i \hat H_{\Sigma_i}) \\
= \bigotimes_{i=1}^k \exp(-s_i H_{\Sigma^i})
\end{split}
\end{equation}
for $s= (s_1,...,s_k) \in (\C_{>0})^k$. We note that $H_\Sigma$ 
is a self-adjoint operator with discrete spectrum unbounded from above and bounded from below on $\mathcal H_\Sigma$. That $V^{Hilb}$ is symmetric monoidal implies the following.

\begin{corollary}
Proposition \ref{prop:hamiltonian} holds for all closed $n-1$-manifolds $\Sigma$.
\end{corollary}

Let 
\begin{equation}\label{eq:component bordism}
X: \Sigma_0 \rightsquigarrow \Sigma_1
\end{equation}
be a bordism with $\Sigma_0 = \bigsqcup_i \Sigma_0^i$ and $\Sigma_1 = \bigsqcup_j \Sigma_1^j$.

\begin{proposition}\label{prop:Hamiltonian relation}
There is an equality of maps
\begin{equation}
V^{Hilb}(X,s) \circ \hat H_{\Sigma_0^i} = \hat H_{\Sigma_1^j} \circ V^{Hilb}(X,s)
\end{equation}
for all connected components $\Sigma_0^i, \Sigma_1^j$.
\end{proposition}
\begin{proof}
Theorem \ref{thm:vol dependence thm} implies that for all $\tau \in \C_{>0}$
\begin{equation}
V^{Hilb}(X,s) \circ \exp(-\tau\hat H_{\Sigma_0^i}) = \exp(-\tau \hat H_{\Sigma_1^j}) \circ V^{Hilb}(X,s)
\end{equation}
where each side is a holomorphic family of bounded operators from $\mathcal H_{\Sigma_0}$ to $\mathcal H_{\Sigma_1}$ in the parameter $\tau \in \C_{>0}$. Taking the derivative with respect to $\tau$ on both sides followed by the limit $\tau \raw 0$ gives the result.
\end{proof}

When $\Sigma$ is a closed connected $n-1$-manifold, there is an orthogonal decomposition
\begin{equation}\label{eq:isometry}
\mathcal H_\Sigma \cong \widehat\bigoplus_{\lambda \in \Spec(H)} \mathcal H_\Sigma^\lambda
\end{equation}
where $\mathcal H_\Sigma^\lambda$ is the finite dimensional Hilbert space corresponding to the $\lambda$-eigenspace of $H_\Sigma$ and $\widehat\oplus$ denotes the Hilbert space completion of the direct sum. 

If $\Sigma := \bigsqcup_i \Sigma^i$ has multiple connected components, there is again an orthogonal decomposition of the form (\ref{eq:isometry}). There is an isomorphism 
\begin{equation}
\mathcal H_\Sigma^\lambda \cong \bigotimes_{i} \mathcal H_{\Sigma^i}^\lambda
\end{equation}
of the $\lambda$-eigenspace of $H_\Sigma$, where $\lambda$ is an eigenvalue in
\begin{equation}
\Spec(H_{\Sigma}) = \bigcap_i \Spec(H_{\Sigma^i}).
\end{equation}
For $X$ a bordism as in (\ref{eq:component bordism}), we set
\begin{equation}
\Lambda_{\Sigma_0,\Sigma_1} := \Spec(H_{\Sigma_0}) \cap \Spec(H_{\Sigma_1})
\end{equation} 
to be the set of eigenvalues shared by the incoming and outgoing Hamiltonians.
For $k=0,1$, let
\begin{equation}
\pi_{k,\lambda}: \mathcal H_{\Sigma_k} \raw \mathcal H_{\Sigma_k}^\lambda
\end{equation}
denote the orthogonal projection and
\begin{equation}
\iota_{k,\lambda} : \mathcal H_{\Sigma_k}^\lambda \mono \mathcal H_{\Sigma_k}
\end{equation}
the inclusion. 
\begin{corollary}\label{cor:VHilb decomposition}
There is a decomposition
\begin{equation}\label{eq:VHilb_lambda}
V^{Hilb}(X,s) = \sum_{\lambda \in \Lambda_{\Sigma_0,\Sigma_1}} \iota_{1,\lambda} \circ V^{Hilb}_\lambda(X,s) \circ \pi_{0,\lambda}
\end{equation}
with $V^{Hilb}_\lambda(X,s) \in \Hom(\mathcal H_{\Sigma_0}^\lambda, \mathcal H_{\Sigma_1}^\lambda)$, such that
\begin{equation}\label{eq:trace class V}
\sum_{\lambda \in \Lambda_{\Sigma_0,\Sigma_1}} \| V_\lambda^{Hilb}(X,s)\| < \infty
\end{equation}
where $\|\cdot \|$ denotes the operator norm.
\end{corollary}
\begin{proof}
The decomposition follows immediately from Proposition \ref{prop:Hamiltonian relation}. $V^{Hilb}(X,s)$ being trace-class implies (\ref{eq:trace class V}).
\end{proof}

Let $\gamma: \R \raw \C$ be any curve satisfying $\gamma'(t) \in \C_{>0}$ and $\gamma(0)=0$. For $\Sigma$ a closed $n-1$-manifold, let 
\begin{equation}
\mathcal C_\gamma^{Hilb} := \{\mathcal H_s := \mathcal H_\Sigma \st s < 0\}
\end{equation}
be the direct system consisting of maps 
\begin{equation}
\mathcal H_s \xrightarrow{V^{Hilb}_\Sigma(\gamma(s') - \gamma(s))} \mathcal H_{s'}
\end{equation}
for all $s<s'<0$
and
\begin{equation}
\mathcal D_\gamma^{Hilb} := \{\mathcal H_t := \mathcal H_\Sigma \st t >0\}
\end{equation} the inverse system consisting of maps
\begin{equation}
\mathcal H_t \xrightarrow{V^{Hilb}_\Sigma(\gamma(t') - \gamma(t))} \mathcal H_{t'}
\end{equation}
for all $0 < t < t'$.

Let $\check E_\Sigma \mono \mathcal H_{\Sigma} \mono \hat E_{\Sigma}$ be the Hermitian nuclear pair assigned to $\Sigma$ by $V$.
By coherence and Proposition \ref{prop:canonical NP isos}, $\check E_\Sigma$ and $\hat E_\Sigma$ are canonically isomorphic to the direct and inverse limits of $\mathcal C^{Hilb}_\gamma, \mathcal D^{Hilb}_\gamma$. Under the isometry (\ref{eq:isometry}), these can be expressed as
\begin{equation}\label{eq:check E explicit}
\begin{split}
\check E_\Sigma \cong \colim_{\mathcal C^{Hilb}_\gamma} \mathcal H_s 
   &= \{(v_\lambda) \in \prod_{\lambda \in \Spec(H_\Sigma)} \mathcal H_\Sigma^\lambda \st \exists s < 0 \text{ , } (e^{-\gamma(s)\lambda}v_\lambda) \in \mathcal H_\Sigma\}\\
   &= \{(v_\lambda) \in \prod_{\lambda \in \Spec(H_\Sigma)} \mathcal H_\Sigma^\lambda \st \exists \tau \in \C_{>0} \text{ , } (e^{\tau\lambda}v_\lambda) \in \mathcal H_\Sigma\}
\end{split}
\end{equation}
\begin{equation}\label{eq:hat E explicit}
\begin{split}
\hat E_\Sigma \cong \invlim_{\mathcal D^{Hilb}_\gamma} \mathcal H_t 
   &= \{(v_\lambda) \in \prod_{\lambda \in \Spec(H_\Sigma)} \mathcal H_\Sigma^\lambda \st \forall t > 0 \text{ , } (e^{-\gamma(t)\lambda}v_\lambda) \in \mathcal H_\Sigma\}\\
   &= \{(v_\lambda) \in \prod_{\lambda \in \Spec(H_\Sigma)} \mathcal H_\Sigma^\lambda \st \forall \tau \in \C_{>0} \text{ , } (e^{-\tau\lambda}v_\lambda) \in \mathcal H_\Sigma\}
\end{split}
\end{equation}
The Fr\'echet topology on $\hat E_\Sigma$ can be prescribed by a family of seminorms
\begin{equation}
\p_{i} : (v_\lambda) \mapsto \|(e^{-\tau_i\lambda}v_\lambda)\|_{\mathcal H_\Sigma}
\end{equation}
where $\{\tau_i\} \subset \C_{>0}$ is any sequence of complex numbers whose real parts are decreasing to 0. The isomorphism 
\begin{equation}\label{eq:bar check E hat E^* iso}
\overline{\check E} \cong \hat E^*
\end{equation}
of (\ref{eq:fixed point}) can be expressed via the nondegenerate pairing
\begin{equation}\label{eq:check E hat E pairing}
\begin{split}
\overline{\check E_\Sigma} \times \hat E_\Sigma &\raw \C\\
(v_\lambda), (w_\lambda) &\mapsto \langle (e^{\bar \tau\lambda}v_\lambda), (e^{-\tau\lambda} w_\lambda) \rangle
\end{split}
\end{equation}
where $\tau \in \C_{>0}$ is such that $(e^{\bar\tau\lambda}v_\lambda), (e^{-\tau\lambda} w_\lambda) \in \mathcal H_\Sigma$ and $\langle, \rangle$ is the sesquilinear inner product on $\mathcal H_\Sigma$.

We now record two lemmas for future use.
\begin{lemma}\label{lem:incoming and outgoing bound}
Let $\Sigma_0, \Sigma_1$ be nonempty closed $n-1$-manifolds and let $f_\lambda : \mathcal H_{\Sigma_0}^\lambda \raw \mathcal H_{\Sigma_1}^\lambda$ be a collection of linear maps between finite dimensional Hilbert spaces for each $\lambda \in \Lambda_{\Sigma_0,\Sigma_1}$. Then the not necessarily bounded operator
\begin{equation}
f^{Hilb} := \sum_{\lambda \in \Lambda_{\Sigma_0,\Sigma_1}} \iota_{1,\lambda} \circ f_\lambda \circ \pi_{0,\lambda}
\end{equation}
from $\mathcal H_{\Sigma_0}$ to $\mathcal H_{\Sigma_1}$ is induced from a morphism of nuclear pairs
\begin{equation}\label{eq:flambda np morphism}
\begin{tikzcd}
\check E_{\Sigma_0} \arrow[r, hook] \arrow[d, "\check f"] &\mathcal H_{\Sigma_0} \arrow[r,hook] &\hat E_{\Sigma_1} \arrow[d, "\hat f"] \\
\check E_{\Sigma_1} \arrow[r,hook] &\mathcal H_{\Sigma_1} \arrow[r,hook] &\hat E_{\Sigma_1}
\end{tikzcd}
\end{equation}
if and only if for all $t>0$ there exists $C>0$ such that
\begin{equation}\label{eq:growth condition}
\|f_\lambda\| < Ce^{t\lambda}
\end{equation}
for every $\lambda \in \Lambda_{\Sigma_0, \Sigma_1}$.
\end{lemma}
\begin{proof}
Assume the identifications (\ref{eq:check E explicit}) and (\ref{eq:hat E explicit}).
In order for $\check f$ to have image in $\check E_{\Sigma_1}$ we need that for all $v = (v_\lambda)$ in the closed subspace $\widehat \oplus_{\lambda \in \Lambda_{\Sigma_0,\Sigma_1}} \mathcal H_{\Sigma_0}^\lambda \subset \mathcal H_{\Sigma_0}$ and $s \in \C_{>0}$, there exists $s' \in \C_{>0}$ such that
\begin{equation}
e^{s'\lambda}\sum_{\lambda \in \Lambda_{\Sigma_0,\Sigma_1}} f_\lambda(e^{-s\lambda} v_\lambda) \in \mathcal H_{\Sigma_1}
\end{equation}
which is true if and only if (\ref{eq:growth condition}) is satisfied. Conversely, if (\ref{eq:growth condition}) holds, the maps
\begin{equation}\label{check f}
\begin{split}
\check f : \check E_{\Sigma_0} &\raw \check E_{\Sigma_1}\\
(x_\lambda) &\mapsto (f_\lambda(x_\lambda))
\end{split}
\end{equation}
and
\begin{equation}\label{eq:hat f}
\begin{split}
\hat f : \hat E_{\Sigma_0} &\raw \hat E_{\Sigma_1}\\
(x_\lambda) &\mapsto (f_\lambda(x_\lambda))
\end{split}
\end{equation}
are well-defined and fit into the commutative diagram (\ref{eq:flambda np morphism}). Let $s>0$ and let
\begin{equation}
\begin{split}
\p_s: \hat E_{\Sigma_1} &\raw \R\\
y &\mapsto \|e^{-sH_{\Sigma_1}}y\|_{\mathcal H_{\Sigma_1}}
\end{split}
\end{equation} 
be a seminorm on $\hat E_{\Sigma_1}$ and let
\begin{equation}
\begin{split}
\q_t: \hat E_{\Sigma_0} &\raw \R\\
x &\mapsto \|e^{-tH_{\Sigma_0}}x\|_{\mathcal H_{\Sigma_0}}
\end{split}
\end{equation}
be a seminorm on $\hat E_{\Sigma_0}$ with $0<t<s$. By assumption, there exists $C>0$ with 
\begin{equation}
\|f_\lambda\| < Ce^{(s-t)\lambda}
\end{equation}
for all $\lambda \in \Lambda_{\Sigma_0,\Sigma_1}$.
If $\q_t(x) < \delta$, then 
\begin{equation}
\begin{split}
\p_s(f_\lambda x_\lambda) &= \|(e^{-s\lambda}f_\lambda x_\lambda)\|_{\mathcal H_{\Sigma_1}} \\
&\leq  C\|e^{-t\lambda} x_\lambda\|\\
& < C\delta
\end{split}
\end{equation}
which shows continuity of (\ref{eq:hat f}).

The adjoint maps 
\begin{equation}
f_\lambda^\dagger: \mathcal H^\lambda_{\Sigma_1} \raw \mathcal H^\lambda_{\Sigma_0}
\end{equation}
satisfy the same bound (\ref{eq:growth condition}) and thus fit together into a continuous map
\begin{equation}
\check f^\dagger : \hat E_{\Sigma_1} \raw \hat E_{\Sigma_0}
\end{equation}
which can be identified with the conjugate transpose $\overline{\check f^*}: \overline{\check E_{\Sigma_1}^*}\raw \overline{\check E_{\Sigma_0}^*}$ via the isomorphism (\ref{eq:bar check E hat E^* iso}). Proposition (\ref{prop: reflexive transpose}) then implies (\ref{check f}) is continuous.
\end{proof}

\begin{lemma}\label{lem:incoming bound}
Let $\Sigma$ be a nonempty closed $n-1$-manifold and let $g_\lambda: \mathcal H^\lambda_{\Sigma} \raw \C$ be a collection of linear maps for each $\lambda \in \Spec(H_{\Sigma})$. If
\begin{equation}
g := \sum_{\lambda \in \Spec(H_\Sigma)} g_\lambda \circ \pi_\lambda
\end{equation}
defines a continuous map $\check E_\Sigma \raw \C$ then for all $t > 0$ there exists $C>0$ such that
\begin{equation}\label{eq:incoming bound}
   \|g_\lambda\| < Ce^{t\lambda}
\end{equation}
for all $\lambda \in \Spec(H_\Sigma)$.
\end{lemma}
\begin{proof}
Let $(x_\lambda) \in \prod_{\lambda \in \Spec(H_\Sigma)} \mathcal H_\Sigma^\lambda$ with $\|x_\lambda\| = 1$.
If $g$ does not satisfy (\ref{eq:incoming bound}), then there exists $t > 0$ such that 
\begin{equation}
\|g_\lambda\| \geq e^{t\lambda}.
\end{equation}
for infinitely many $\lambda \in \Spec(H_\Sigma)$.
By the characterization given in (\ref{eq:check E explicit}), the element $y := (e^{-s\lambda}x_\lambda)$ lies in $\check E_\Sigma$ for all $s > 0$ since $(e^{-\frac{s}{2}\lambda}x_\lambda) \in \mathcal H_\Sigma$. If $s < t$, then $g(y)$ is unbounded which contradicts the assumption that $g: \check E_\Sigma \raw \C$ is well defined and continuous.
\end{proof}

When the VFT is reflection positive, the Lorentzian limit of $V$ is a functor
\begin{equation}\label{eq:rp lor lim}
L: \Bord_{n,n-1}^{in\wedge out}(i\R) \raw \mathcal{NP}^h
\end{equation}
taking values in Hermitian nuclear pairs.
If $X$ is a bordism having both incoming and outgoing boundary, $L(X,i\zeta)$ is the morphism
\begin{equation}\label{eq:lorentzian morphism}
\begin{tikzcd}
\check E_{\Sigma_0} \arrow[r,hook] \arrow[d, "\check L(X{,}i\zeta)"'] 
   & \mathcal H_{\Sigma_0} \arrow[r,hook] 
   & \hat E_{\Sigma_0} \arrow[d, "\hat L(X{,}i\zeta)"]\\
\check E_{\Sigma_1} \arrow[r,hook, "j"]
   & \mathcal H_{\Sigma_1} \arrow[r,hook]
   & \hat E_{\Sigma_1}
\end{tikzcd}
\end{equation}
If $X = \Sigma \times I$ is a cylinder, $L(X,i\zeta)$ is a unitary morphism (cf. Definition \ref{def:bounded nph morphism}) with
\begin{equation}
L^{Hilb}(\Sigma \times I ,i\zeta) = e^{-i\zeta H_\Sigma}.
\end{equation}
However for a general bordism $X: \Sigma_0 \rightsquigarrow \Sigma_1$ equipped with imaginary volume $i\zeta \in (i\R)^{\#\pi_0(X)}$, there is no bounded map of Hilbert spaces that fits into the diagram. Nevertheless, we can define an unbounded operator 
\begin{equation}
L^{Hilb}(X,i\zeta) := j \circ \check L(X,i\zeta)
\end{equation}
whose domain of definition is $\check E_{\Sigma_0}$ which we identify with its dense image in $\mathcal H_{\Sigma_0}$.

\begin{proposition}\label{prop:L hilb bound}
The unbounded operator $L^{Hilb}(X,i\zeta)$ has a decomposition
\begin{equation}\label{eq:Lhilb decomp}
L^{Hilb}(X,i\zeta) = \sum_{\lambda \in \Lambda_{\Sigma_0,\Sigma_1}} \iota_{1,\lambda} \circ L^{Hilb}_\lambda(X,i\zeta) \circ \pi_{0,\lambda}
\end{equation}
with $L^{Hilb}_\lambda(X,i\zeta) \in \Hom(\mathcal H_{\Sigma_0}^\lambda, \mathcal H_{\Sigma_1}^\lambda)$, such that for all $t > 0$, there exists $C>0$ with 
\begin{equation}\label{eq:unbounded growth condition}
\| L_\lambda^{Hilb}(X,i\zeta) \| < Ce^{t\lambda}
\end{equation}
for all $\lambda \in \Lambda_{\Sigma_0,\Sigma_1}$.
\end{proposition}
\begin{proof}
The map $L^{Hilb}(X,i\zeta)$ is the unique map making the diagram
\begin{equation}
\begin{tikzcd}[column sep = huge]
\mathcal H_{\Sigma_0} \arrow[r, "\exp(-\tau H_{\Sigma_0})"] \arrow[rd, "V^{Hilb}(X{,} \tau \oplus i\zeta)"']
   &\check E_{\Sigma_0} \arrow[d, dashed, "L^{Hilb}(X{,}i\zeta)"]\\
 & \mathcal H_{\Sigma_1}
\end{tikzcd}
\end{equation}
commute for all $\tau \in (\C_{>0})^{\#\pi_0(\Sigma_0)}$. If $L^{Hilb}(X,i\zeta)$ did not satisfy (\ref{eq:Lhilb decomp}), there would be some $v_\lambda \in \mathcal H_{\Sigma_0}^\lambda$ such that $\pi_{1,\mu}(L^{Hilb}(X,i\zeta)v_\lambda) \neq 0$ for $\mu \neq \lambda$ in $\Spec(H_{\Sigma_1})$. Since $\mathcal H_{\Sigma_0}^\lambda$ is an eigenspace of $\exp(-\tau H_{\Sigma_0})$, this would imply that $\pi_{1,\mu}(V^{Hilb}(X,\tau\oplus i\zeta)v_\lambda) \neq 0$ which contradicts Corollary \ref{cor:VHilb decomposition}.
The growth condition (\ref{eq:unbounded growth condition}) is implied by Lemma \ref{lem:incoming and outgoing bound}.
\end{proof}

Thus (\ref{eq:rp lor lim}) sends cylindrical bordisms to unitary morphisms and general bordisms to unbounded morphisms satisfying the sub-exponential growth condition (\ref{eq:unbounded growth condition}). 

By Proposition \ref{prop:sheaf induced functor}, the sheaf morphism (\ref{eq: metlor to densir}) induces a symmetric monoidal functor 
\begin{equation}
\sqrt{\det}_*: \Bord_{n,n-1}^{in\wedge out}(\Met_{\Lor}) \raw \Bord_{n,n-1}^{in\wedge out}(\Dens_{i\R}).
\end{equation} 
The integration functor (\ref{eq:int functor}) extends to a symmetric monoidal functor
\begin{equation}
\int: \Bord_{n,n-1}^{in\wedge out}(\Dens_{i\R}) \raw \Bord_{n,n-1}^{in\wedge out}(i\R)
\end{equation}
which sends a bordism equipped with a purely imaginary density with nonempty incoming and outgoing boundary to the same bordism labeled by its purely imaginary total volume. We set
\begin{equation}
\mathscr L: \Bord_{n,n-1}^{in\wedge out}(\Met_{\Lor}) \xrightarrow{\sqrt{\det}_*} \Bord_{n,n-1}^{in\wedge out}(\Dens_{i\R}) \xrightarrow{\int} \Bord_{n,n-1}^{in\wedge out}(i\R) \xrightarrow{L} \mathcal{NP}^h
\end{equation}
to be the composition. We summarize the above in the following:

\begin{theorem}\label{thm:LHilb2}
A reflection positive volume-dependent field theory induces a symmetric monoidal functor
\begin{equation}\label{eq:lor lim}
\mathscr L : \Bord_{n,n-1}^{in\wedge out}(\Met_{\Lor}) \raw \mathcal{NP}^h
\end{equation}
out of the category of bordisms equipped with smooth, possibly degenerate Lorentzian metrics with nonempty incoming and outgoing boundary. 
It sends cylindrical bordisms to unitary morphisms and general bordisms to unbounded morphisms satisfying the sub-exponential growth condition (\ref{eq:unbounded growth condition}). 
\end{theorem}

\begin{definition}
We call (\ref{eq:lor lim}) the \emph{Lorentzian limit of $Z$}.
\end{definition}

\begin{remark}
Theorem 5.2 of \cite{KS} constructs the Lorentzian limit of a reflection positive $\Met_\C$-field theory on the subcategory of real-analytic and globally hyperbolic bordisms valued in Hilbert spaces and unitary operators. Theorem \ref{thm:LHilb2} extends the domain of this functor to smooth, possibly degenerate Lorentzian bordisms with nonempty incoming/outgoing boundary when the reflection positive field theory is volume-dependent, provided we also enlarge the codomain category to allow unbounded morphisms of Hermitian nuclear pairs.
\end{remark}

When the field theory is reflection positive, the short distance limit (cf. Definition \ref{def:short distance limit}) is a functor
\begin{equation}
L_0: \Bord_{n,n-1}^{in\wedge out} \raw \mathcal{NP}^h.
\end{equation}
We record a lemma which can be checked by unraveling definitions.

\begin{lemma}\label{lem:VFT from L0}
Let $X: \Sigma_0 \rightsquigarrow \Sigma_1$ be a bordism with volume $s \in (\C_{>0})^{\#\pi_0(X)}$ and let $s_i \in (\C_{>0})^{\#\pi_0(\Sigma_i)}$ for $i=0,1$ such that $s_0 \oplus 0_X \oplus s_1 = s$. Then
\begin{equation}
V^{Hilb}(X,s) =
\begin{cases}
\exp(-s_1H_{\Sigma_1}) \circ \hat L_0(X) & \Sigma_0 = \emptyset, \Sigma_1 \neq \emptyset\\
\check L_0(X) \circ \exp(-s_0H_{\Sigma_0}) & \Sigma_0 \neq \emptyset, \Sigma_1 = \emptyset\\
\exp(-s_1H_{\Sigma_1}) \circ L_0^{Hilb}(X) \circ \exp(-s_0H_{\Sigma_0}) & \Sigma_0,\Sigma_1 \neq \emptyset
\end{cases}
\end{equation}
where in our notation we regard $\check E_{\Sigma_i} \subset \mathcal H_{\Sigma_i} \subset \hat E_{\Sigma_i}$ as subspace inclusions and $L_0^{Hilb}(X):\mathcal H_{\Sigma_0} \raw \mathcal H_{\Sigma_1}$ is the unbounded operator induced by $L_0(X)$.
\end{lemma}


\subsection{The Long-Distance Topological Limit}\label{section long distance limit}
In this section we define the long-distance topological limit of a reflection positive VFT.
Let $Z$ be a reflection positive volume-dependent theory with induced functor
\begin{equation}
V^{Hilb}: \Bord_{n,n-1}(\C_{>0}) \raw \Hilb.
\end{equation}
Recall that by Proposition \ref{prop:hamiltonian} $V^{Hilb}$ evaluates on a cylindrical bordism $\Sigma \times I$ of volume $s \in \C_{>0}$ to the operator
\begin{equation}
V^{Hilb}_\Sigma(s) := \exp(-sH_\Sigma)
\end{equation} 
where $H_\Sigma$ are self-adjoint unbounded operators on $\mathcal H_\Sigma$ with discrete spectrum bounded below with finite multiplicity. In this section we assume that all $H_\Sigma$ have non-negative spectrum with lowest eigenvalue 0.

Using isomorphisms (\ref{eq:check E explicit}) and (\ref{eq:hat E explicit}), we identify $\ker(H_\Sigma)$ as a finite dimensional subspace in the sequence of inclusions
\begin{equation}
\ker(H_\Sigma) \subset \check E_{\Sigma} \subset \mathcal H_{\Sigma} \subset \hat E_{\Sigma}.
\end{equation}
Let $\pi_\Sigma \in \End(\mathcal H_\Sigma)$ be the orthogonal projection onto $\ker(H_\Sigma)$ and let $\iota_\Sigma : \ker(\Sigma) \mono \mathcal H_\Sigma$ be the inclusion. If $\Sigma = \emptyset$ is the empty manifold we set $\pi_\Sigma = \iota_\Sigma = id$ to be the identity  map on $\C$.

As $\Re(s) \raw +\infty$, the operator $V_\Sigma^{Hilb}(s)$ converges to $\pi_\Sigma$ uniformly. In light of Lemma \ref{lem:VFT from L0}, we make the following definition.

\begin{definition}
The \emph{long-distance topological limit} of $Z$ is the functor
\begin{equation}
L_\infty: \Bord_{n,n-1} \raw \Vect
\end{equation}
which sends $\Sigma$ to $\ker(H_\Sigma)$ and a bordism $X: \Sigma_0 \rightsquigarrow \Sigma_1$ with at least one nonempty boundary to
\begin{equation}
L_\infty(X) :=
\pi_{\Sigma_1} \circ \tilde L_0(X) \circ \iota_{\Sigma_0} 
\end{equation}
where
\begin{equation}
\tilde L_0(X) :=
\begin{cases}
\hat L_0(X) & \Sigma_1 \neq \emptyset\\
\check L_0(X) & \Sigma_0 \neq \emptyset
\end{cases}
\end{equation}
\end{definition}

\begin{remark}
If both $\Sigma_0 = \Sigma_1 = \emptyset$, we can express $X$ as the composition of bordisms with at least one nonempty boundary. If neither $\Sigma_0$ nor $\Sigma_1$ are empty, $\pi_{\Sigma_1} \circ \check L_0(X) \circ \iota_{\Sigma_0} = \pi_{\Sigma_1} \circ \hat L_0(X) \circ \iota_{\Sigma_0}$. $L_\infty$ is symmetric monoidal because $\ker(H_{\Sigma \sqcup \Sigma'}) = \ker(H_{\Sigma}) \otimes \ker(H_{\Sigma'})$ and thus defines a topological field theory.
\end{remark}


\subsection{Classification of 2d Reflection Positive VFTs}\label{section 2d classification}

We first recall the definition of a Frobenius algebra.
\begin{definition}
Let $k$ be a field.
A \emph{Frobenius algebra  $(A,\theta)$ over $k$} is a finite dimensional unital associative algebra $A$ over $k$ equipped with a \emph{trace}
\begin{equation}
\theta: A \raw k
\end{equation}
such that the \emph{Frobenius pairing}
\begin{equation}
\begin{split}
A \times A &\raw k\\
(x,y) & \mapsto \theta(xy)
\end{split}
\end{equation}
is non-degenerate. 
\end{definition}

\begin{remark}
We will also sometimes denote a Frobenius algebra by a triple $(A,m,\theta)$ where $m: A\times A \raw k$ is the multiplication on $A$ which we suppressed in the definition above.
\end{remark}

The Frobenius pairing induces an isomorphism
\begin{equation}\label{eq:theta iso}
A^* \cong A
\end{equation}
which, when composed with the transpose $\theta^*: k \raw A^*$, gives the \emph{unit map}
\begin{equation}
\begin{split}
u: k &\raw A\\
1 &\mapsto 1_A
\end{split}
\end{equation}
where $1_A \in A$ is the unit in $A$.
Similarly, the transpose $m^*$ of the multiplication map
\begin{equation}
m: A \otimes A \raw A
\end{equation}
becomes the coproduct 
\begin{equation}
w: A \raw A \otimes A
\end{equation}
upon applying the isomorphism (\ref{eq:theta iso}).

\begin{definition}
A \emph{Hermitian Frobenius algebra} is a triple $(A,\theta,c)$ such that $(A,\theta)$ is a Frobenius algebra over $\C$ and 
\begin{equation}
c: A \xrightarrow{\cong} \overline A
\end{equation}
is an antilinear involution of algebras such that the sesquilinear pairing
\begin{equation}\label{eq:frob hilbert <>}
\begin{split}
\langle , \rangle : \overline A \times A &\raw \C\\
x,y &\mapsto \theta(c(x)y)
\end{split}
\end{equation}
is positive definite.
\end{definition}
\begin{remark}
The pairing (\ref{eq:frob hilbert <>}) is a Hilbert space inner product on $A$.
\end{remark}

Let $Z: \Bord_{2,1}(\Dens_\C) \raw \mathcal{NP}$ be a reflection positive VFT in dimension 2. Construction \ref{V construction} gives a functor
\begin{equation}
V: \Bord_{2,1}(\C_{>0}) \raw \mathcal{NP}^{h}
\end{equation}
which, by Proposition \ref{prop:Z nat iso Z'}, recovers $Z$ up to natural isomorphism. Postcomposing $V$ by (\ref{eq:NP^h --> Hilb}) gives the functor
\begin{equation}
V^{Hilb}: \Bord_{2,1}(\C_{>0}) \raw \Hilb.
\end{equation}
and coherence implies that $V$ can be reconstructed from $V^{Hilb}$. Let 
\begin{equation}
\check{\mathcal A} \mono \mathcal A \mono \hat{\mathcal A}
\end{equation} be the Hermitian nuclear pair assigned by $V$ to the circle; the Hilbert space $\mathcal A$ is assigned to the circle by $V^{Hilb}$. By Proposition \ref{prop:hamiltonian} there exists a self-adjoint unbounded operator $H$ on $\mathcal A$ with discrete spectrum bounded from below such that 
\begin{equation}
V^{Hilb}(S^1 \times I, s) = \exp(-sH).
\end{equation} 
Let 
\begin{equation}
\mathcal A \cong \bigoplus_{\lambda \in \Spec(H)} \mathcal A_\lambda
\end{equation}
be the orthogonal spectral decomposition of $H$. As in (\ref{eq:check E explicit}), (\ref{eq:hat E explicit}) we have the identifications
\begin{equation}
\begin{split}
\check{\mathcal A} \cong \{(x_\lambda) \in \prod_{\lambda \in \Spec(H)} \mathcal A_\lambda \st \exists \tau \in \C_{>0}, (e^{\tau\lambda}x_\lambda) \in \mathcal A \}\\
\hat{\mathcal A} \cong \{(x_\lambda) \in \prod_{\lambda \in \Spec(H)} \mathcal A_\lambda \st \forall \tau \in \C_{>0}, (e^{-\tau\lambda}x_\lambda)\in \mathcal A \}.
\end{split}
\end{equation}

The short distance limit for bordisms with incoming boundary
\begin{equation}
\check L_0: \Bord^{in}_{2,1} \raw \mathcal{NDF}
\end{equation}
sends the disk 
\begin{equation}
D : S^1\rightsquigarrow \emptyset
\end{equation} 
to a continuous map
\begin{equation}
\begin{split}
\theta: \check{\mathcal A} &\raw \C\\
(x_\lambda) &\mapsto (\theta_\lambda(x_\lambda))
\end{split}
\end{equation}
which satisfies the bound (\ref{eq:incoming bound}) by Lemma \ref{lem:incoming bound}. 

The short distance limit for bordisms with nonempty incoming and outgoing boundary
\begin{equation}
L_0: \Bord^{in\wedge out}_{2,1} \raw \mathcal{NP}
\end{equation}
sends the pair of pants 
\begin{equation}
P: S^1 \sqcup S^1 \rightsquigarrow S^1
\end{equation} 
to the morphism of Hermitian nuclear pairs
\begin{equation}
\begin{tikzcd}
\check{\mathcal A} \otimes \check{\mathcal A} \arrow[r,hook] \arrow[d, "\check L_0(P)"] & \mathcal A \otimes \mathcal A \arrow[r,hook] \arrow[d, "m", dashed] & \hat{\mathcal A} \otimes \hat{\mathcal A} \arrow[d, "\hat L_0(P)"]\\
\check{\mathcal A} \arrow[r,hook] &\mathcal A \arrow[r,hook] & \hat{\mathcal A}
\end{tikzcd}
\end{equation}
where 
\begin{equation}
\begin{split}
m: \mathcal A \otimes \mathcal A &\raw \mathcal A\\
(x_\lambda) \otimes (y_\lambda) &\mapsto (m_\lambda(x_\lambda,y_\lambda))
\end{split}
\end{equation}
is a possibly unbounded operator with
\begin{equation}
m_\lambda: \mathcal A_\lambda \otimes \mathcal A_\lambda \raw \mathcal A_\lambda
\end{equation}
satisfying the bound (\ref{eq:growth condition}) by Lemma \ref{lem:incoming and outgoing bound}.

For each $\lambda \in \Spec(H)$, there is a symmetric monoidal functor
\begin{equation}
L_{0,\lambda} : \Bord_{2,1} \raw \Vect
\end{equation}
which sends $S^1$ to $\mathcal A_\lambda$ and a bordism $X: \Sigma_0 \rightsquigarrow \Sigma_1$ with nonempty boundary to
\begin{equation}
L_{0,\lambda}(X) := 
\begin{cases}
\pi_{\mathcal A_\lambda} \circ \hat L_0(X) & \Sigma_0 = \emptyset, \Sigma_1 \neq \emptyset\\
\check L_0(X) \circ \iota_{\mathcal A_\lambda} & \Sigma_0 \neq \emptyset, \Sigma_1 = \emptyset\\
\pi_{\mathcal A_\lambda} \circ L^{Hilb}_0(X) \circ \iota_{\mathcal A_\lambda} & \Sigma_0, \Sigma_1 \neq \emptyset
\end{cases}
\end{equation}
These functors satisfy
\begin{equation}
V^{Hilb}(X,s) = \sum_{\lambda \in \Spec(H)} e^{-s\lambda} \iota_{\mathcal A_\lambda} \circ L_{0,\lambda}(X) \circ \pi_{\mathcal A_\lambda}.
\end{equation}
By the classification of 2d TFTs, the multiplication 
\begin{equation}
m_\lambda = L_{0,\lambda}(P)
\end{equation}
and the trace
\begin{equation}
\theta_\lambda = L_{0,\lambda}(D)
\end{equation}
make $(\mathcal A_\lambda, m_\lambda, \theta_\lambda)$ a commutative Frobenius algebra.

The Frobenius and Hilbert space inner products induce isomorphisms 
\begin{equation}
\mathcal A_\lambda \cong \mathcal A_\lambda^*
\quad\quad \text{and} \quad\quad
\mathcal A_\lambda^* \cong \overline{\mathcal A_\lambda}
\end{equation}
which, when composed, gives a real structure
\begin{equation}
c_\lambda: \mathcal A_\lambda \raw \overline{\mathcal A_\lambda}
\end{equation}
satisfying 
\begin{equation}
\langle x,y\rangle = \theta_\lambda(c_\lambda(x)y).
\end{equation}
Thus $(\mathcal A_\lambda, \theta_\lambda, c_\lambda)$ is a Hermitian Frobenius algebra. Lemma \ref{lem:incoming bound} and Lemma \ref{lem:incoming and outgoing bound} imply the bounds (\ref{eq:m and theta bounds}).
When $H$ has lowest eigenvalue 0, $(\mathcal A_0, \theta_0)$ is the Frobenius algebra associated to the long-distance topological limit. Thus we have shown one direction of:

\begin{theorem}\label{thm:2d classification}
A reflection positive 2d VFT
\begin{equation}\label{eq:2d reflection positive vft}
Z : \Bord_{2,1}(\Dens_\C) \raw \mathcal{NP}
\end{equation}
is equivalent, up to natural isomorphism, to a Hilbert space $(\mathcal A, \langle, \rangle)$ equipped with an unbounded self-adjoint operator $H$ with discrete spectrum bounded below whose spectral decomposition
\begin{equation}
\mathcal A \cong \widehat \bigoplus_{\lambda \in \Spec(H)} \mathcal A_\lambda
\end{equation}
is a Hilbert sum of commutative Hermitian Frobenius algebras $(\mathcal A_\lambda, m_\lambda, \theta_\lambda, c_\lambda)$ such that the operator norms $\|\theta_\lambda\|$, $\|m_\lambda\|$ have sub-exponential growth in $\lambda$, i.e. for all $t >0$ there exists $C>0$ with
\begin{equation}\label{eq:m and theta bounds}
\|\theta_\lambda\| < Ce^{\lambda t} \quad\quad\text{and} \quad\quad \|m_\lambda\| < Ce^{\lambda t}
\end{equation}
for all $\lambda \in \Spec(H)$.
When $H$ has lowest eigenvalue 0 then the long distance topological limit is the 2d topological field theory determined by the Frobenius algebra $(\mathcal A_0,\theta_0)$. 
\end{theorem}

\begin{proof}
To show the other direction, assume given Hermitian Frobenius algebras $\mathcal A_\lambda$ satisfying the bounds (\ref{eq:m and theta bounds}).

By compatibility of the Hilbert and Frobenius pairings, the coproducts $w_\lambda$ and the units $u_\lambda$ have norms
\begin{equation}
\|w_\lambda\| = \|m_\lambda\|\quad\quad\text{and}\quad\quad
\|u_\lambda\| = \|\theta_\lambda\|
\end{equation} 
and therefore also satisfy the growth conditions (\ref{eq:m and theta bounds}). 

Set 
\begin{equation}
   \mathcal A := \widehat\bigoplus_{\lambda \in \Spec(H)} \mathcal A_\lambda
\end{equation}
to be the Hilbert space completion of the direct sum and let
\begin{equation}
\begin{split}
\iota_\lambda^{\otimes n}: \mathcal A_\lambda^{\otimes n} \mono \mathcal A^{\otimes n}\\
\pi_\lambda^{\otimes n} : \mathcal A^{\otimes n} \raw \mathcal A_\lambda^{\otimes n}
\end{split}
\end{equation}
be the inclusion and orthogonal projection respectively of $n$-fold tensor products of Hilbert spaces. When the context is clear, we will suppress the tensor notation and simply refer to these maps by $\iota_\lambda$ and $\pi_\lambda$. 

By the classification of 2d TFTs, the Frobenius algebra $\mathcal A_\lambda$ determines a symmetric monoidal functor
\begin{equation}
F_\lambda: \Bord_{2,1} \raw \Vect
\end{equation}
uniquely up to isomorphism satisfying $F_\lambda (S^1) = \mathcal A_\lambda$ and
whose values on the elementary bordisms $P,D,P^*,D^*$ are $m_\lambda,\theta_\lambda,w_\lambda,u_\lambda$ (c.f. \cite{moore2006dbranes} Appendix A). Let $X: \Sigma_0 \rightsquigarrow \Sigma_1$ be a bordism between closed 1-manifolds. Define
\begin{equation}\label{eq: V sum}
V(X,s) := \sum_{\lambda \in \Spec(H)} e^{-s\lambda} \pi_\lambda \circ F_\lambda(X) \circ \iota_\lambda
\end{equation}
Choose a Morse decomposition
\begin{equation}
   X = X_n \circ \cdots \circ X_1
\end{equation}
into elementary bordisms $X_i$. The bounds (\ref{eq:m and theta bounds}) imply that for all $s_i \in \C_{>0}$ the maps $V(X_i,s_i)$ are bounded maps between Hilbert spaces. Choosing a partition $s = \sum_i s_i$ and rewriting
\begin{equation}
   V(X,s) = \sum_{\lambda \in \Spec(H)} V(X_n, s_n) \circ \cdots \circ V(X_1,s_1)
\end{equation}
shows that (\ref{eq: V sum}) is a bounded map of Hilbert spaces
\begin{equation}
   V(X,s): \mathcal A^{\otimes \#\pi_0 \Sigma_0} \raw \mathcal A^{\otimes \#\pi_0(\Sigma_1)}.
\end{equation}
Functoriality of $F_\lambda$ implies that these maps fit together into a functor
\begin{equation}
V: \Bord_{2,1}(\C_{>0}) \raw \Hilb.
\end{equation}
which in turn determines a reflection positive VFT up to natural isomorphism.
\end{proof}

\begin{example}[2d Yang-Mills]
We now sketch how Yang-Mills theory in 2 dimensions can be described as a VFT using the above theorem. For background on this theory, we refer the reader to \cite{cmp/1104248198}. We also refer to \cite{segal_TFTs} which includes a broader discussion of 2-dimensional and real volume-dependent theories.

Let $G$ be a compact Lie group, let $\g := \Lie(G)$ be its Lie algebra, and let $\langle,\rangle$ denote its Killing form. Choose a basis $(X_i)$ of $\g$ and let $(X^i)$ be its dual with respect to $\langle,\rangle$. Set
\begin{equation}
\begin{split}
\mathcal A &:= L^2(G)^G\\
H &:= \sum_{i} X_i X^i
\end{split}
\end{equation}
to be the class functions on $G$ with Hamiltonian given by the quadratic Casimir viewed as a left-invariant differential operator on $L^2(G)$. By the Peter-Weyl Theorem there is an orthogonal decomposition
\begin{equation}
\mathcal A \cong \widehat\bigoplus_{\chi \in \check G}\mathcal A_\chi
\end{equation}
where $\mathcal A_\chi$ is the 1-dimensional eigenspace of $H$ spanned by a character $\chi \in \check G$ of unit norm; we denote its eigenvalue by $c_2(\chi)$. Let $\pi_\chi: G \raw \Aut(V_\chi)$ be the irreducible representation whose character is $\chi$.

The eigenspace $\mathcal A_\chi$ is a one-dimensional Frobenius algebra with multiplication, comultiplication, trace, and unit given by
\begin{equation} 
m_\chi: \chi\otimes \chi \mapsto \frac{1}{\dim(V_\chi)}\chi
\end{equation}

\begin{equation}
w_\chi : \chi \mapsto \frac{1}{\dim(V_\chi)}\chi\otimes \chi
\end{equation}

\begin{equation}
\theta_\chi: \chi \mapsto \dim(V_\chi)
\end{equation}

\begin{equation}
u_{\mathcal A_\chi} : 1 \mapsto  \dim(V_\chi)\cdot \chi
\end{equation}
Let $\h \subset \g$ be the Cartan subalgebra and let $R \subset \h$ be the set of roots of $\g$. Choose a base $\Delta \subset R$, and let $R^+$ be the set of positive roots relative to $\Delta$. Set
\begin{equation}
\delta := \frac{1}{2}\sum_{\alpha \in R^+} \alpha
\end{equation}
and let $\mu_\chi \in \h$ be the highest weight of $(V_\chi,\pi_\chi)$.

Let $\mathcal D \subset \h$ denote the subset of dominant integral elements. We recall that there is a one-to-one correspondence 
\begin{equation}
\begin{split}
   \mathcal D &\leftrightarrow \{\text{Irreps of } G\} \\
   \mu &\leftrightarrow V_\mu
\end{split}
\end{equation}
where $V_\mu$ is the unique irreducible representation with highest weight $\mu$. 

Consider the functions
\begin{equation}
\begin{split}
c : \mathcal D &\raw \R\\
\mu &\mapsto \langle \mu + \delta , \mu + \delta \rangle - \langle \delta, \delta\rangle
\end{split}
\end{equation}
and 
\begin{equation}
\begin{split}
d: \mathcal D &\raw \R \\
\mu &\mapsto \frac{\prod_{\alpha \in R^+}\langle \alpha, \mu+\delta\rangle}{\prod_{\alpha \in R^+} \langle \alpha,\delta \rangle}
\end{split}
\end{equation}
The function $c(\mu)$ is positive and quadratic and $d(\mu)$ is polynomial. This implies that for all $t > 0$ there is a constant $C>0$ with 
\begin{equation}\label{eq:dc bound}
d(\mu) < Ce^{tc(\mu)}
\end{equation}

Setting $\mu:=\mu_\chi$ gives
\begin{equation}
c_2(\chi) = c(\mu_\chi)
\end{equation}
and by the Weyl dimension formula we also have
\begin{equation}
\dim(V_\chi) = d(\mu_\chi).
\end{equation}
For derivations of these formulas we refer the reader to \cite{humphreys2012introduction}. Thus (\ref{eq:dc bound}) becomes
\begin{equation}
\dim(V_\chi) < Ce^{tc_2(\chi)}
\end{equation}
and the bounds (\ref{eq:m and theta bounds}) are satisfied.
\end{example}


\appendix

\section{Nuclear Spaces}\label{appendix Nuclear Spaces}
\subsection{Topological Vector Spaces}
\begin{definition}
A \emph{filter} on a set $X$ is a family of subsets $\mathscr F \subset 2^X$ satisfying:
\begin{enumerate}
   \item $\emptyset \notin \mathscr F$
   \item $U, V \in \mathscr F \implies U\cap V \in \mathscr F$
   \item $A \supset U$ and $U \in \mathscr F \implies A \in \mathscr F$
\end{enumerate}
\end{definition}

A topology on $X$ determines a filter $\mathscr F_x$ for every $x \in X$ consisting of the neighborhoods of $x$. We recall that a neighborhood $U$ of $x$ contains an open sub-neighborhood $x$ and we can regard $U$ as a neighborhood of $y$ for every $y \in V$. Thus for each $x\in X$, $\mathscr F_x$ satisfies 
\begin{enumerate}
   \item $x \in \bigcap_{U \in \mathscr F_x} U$
   \item For every $U \in \mathscr F_x$, there exists $V \in \mathscr F_x$ such that for all $y \in V$, $U \in \mathscr F_y$. 
\end{enumerate}

Conversely, a collection of filters $\mathscr F_x$ for each $x \in X$ satisfying the above two properties defines a topology on $X$.

\begin{definition}
Let $X$ be a topological space. A filter $\mathscr F$ \emph{converges} to $x \in X$ if every neighborhood of $x$ belongs to $\mathscr F$.
\end{definition}

Let $E$ be a topological vector space.

\begin{definition}
A filter $\mathscr F \subset 2^E$ is \emph{Cauchy} if to every neighborhood $U$ of $0 \in E$ there is a subset $M \in \mathscr F$ such that $M-M \subset U$.
\end{definition}

\begin{definition}
$E$ is \emph{complete} if every Cauchy filter on $E$ converges to a point $x\in E$.
\end{definition}

Before defining completions, we record some definitions while $E$ is a general topological vector space.

\begin{definition}
A subset $B \subset E$ is \emph{bounded} if for every neighborhood $U$ of $0 \in E$ there exists $t > 0$ such that $B \subset tU$.
\end{definition}

\begin{definition}
A subset $D \subset E$ is a \emph{disk} if it is convex and balanced.
\end{definition}

\begin{definition}
A subset $T \subset E$ is a \emph{barrel} if it is a closed and absorbing disk.
\end{definition}

\begin{definition}
$E$ is \emph{barrelled} if every barrel is a neighborhood of $0 \in E$.
\end{definition}

We now assume $E$ is Hausdorff.
Let $\mathscr C$ be the set of all Cauchy filters on $E$. Let $\mathscr F, \mathscr G \in \mathscr C$ be two Cauchy filters. We will write $\mathscr F \sim \mathscr G$ if for all neighborhoods $U$ of $0\in E$, there is an element $A$ of $\mathscr F$ and an element $B$ of $\mathscr G$ such that $A -B \subset U$.

\begin{lemma}
The relation $\sim$ is an equivalence relation.
\end{lemma}

\begin{proof}
\cite{treves2006topological} Theorem 5.2 (1)
\end{proof}

\begin{definition}
The \emph{completion} of $E$ is the quotient $\tilde E: = \mathscr C/\sim$
\end{definition}

\begin{proposition}\label{prop:completion UP}
$\tilde E$ is a complete Hausdorff topological vector space and there exists a bicontinuous, dense injection $E\mono \tilde E$ satisfying the following universal property: for every continuous linear map $E \raw F$ with $F$ Hausdorff, there exists a unique continuous extension
\begin{equation}
\begin{tikzcd}
E \arrow[r] \arrow[d, hook] & F \\
\tilde E \arrow[ru,dashed, "\exists!"] &
\end{tikzcd}
\end{equation}
\end{proposition}
\begin{proof}
\cite{treves2006topological} Theorem 5.2
\end{proof}

\subsection{Locally Convex Hausdorff Spaces}

Let $E$ be a locally convex Hausdorff vector space.

\begin{proposition}
The completion $\tilde E$ is a locally convex Hausdorff space.
\end{proposition}

Let $E':= \Hom(E, \C)$ denote the dual vector space of continuous linear functionals with no choice of topology.

\begin{definition}
The \emph{polar} of a subset $A \subset E$ is the subset $A^0 \subset E'$ defined by
\begin{equation}
   A^0 := \{x' \in E' \st \sup_{x \in A} \langle x', x\rangle \leq 1 \}
\end{equation}
\end{definition}

\begin{definition}
A subset $S \subset E'$ is \emph{equicontinuous} if for all $\varepsilon > 0$, there is a neighborhood $U$ of $0 \in E$ such that for all $f \in E'$ and $x \in U$, $|f(x)| < \varepsilon$.
\end{definition}

\begin{proposition}\label{prop:equicontinuous polar}
A subset $S \subset E'$ is equicontinuous if and only if it is contained in the polar of a neighborhood of $0 \in E$.
\end{proposition}

\begin{proof}
\cite{treves2006topological} Proposition 32.7
\end{proof}

Denote by $E_\sigma'$ and $E^*$ the dual of $E$ endowed with the weak and strong topologies respectively, i.e. the topologies of uniform convergence on finite and bounded sets respectively. 

\begin{definition}
Given a continuous map $u:E \raw F$ between locally convex Hausdorff vector spaces, the \emph{transpose of $u$} is the map
\begin{equation}
u^*: F' \raw E'
\end{equation}
which precomposes a linear functional on $F$ by $u$.
\end{definition}

\begin{proposition}\label{prop: Imu dense u^* iff injective}
Let $u : E \raw F$ be a continuous map between locally convex Hausdorff spaces and let $u^* :F' \raw E'$ be its transpose. Then $\Ima(u)$ is dense in $F$ if and only if $u^*$ is injective.
\end{proposition}
\begin{proof}
This is a consequence of the Hahn Banach theorem -- see \cite{treves2006topological} Corollary 5 of Theorem 18.2.
\end{proof}

\begin{proposition}\label{prop: transpose continuous}
The transpose is a continuous linear map
\begin{equation}
u^*: F^* \raw E^*
\end{equation}
between strong duals.
\end{proposition}
\begin{proof}
\cite{treves2006topological} Corollary 19.5
\end{proof}

\begin{proposition}
There is an isomorphism of vector space $E \simeq (E_\sigma')'$
\end{proposition}
\begin{proof}
\cite{treves2006topological} Proposition 35.1
\end{proof}

\begin{proposition}\label{prop:tensor bilinear}
Let $E,F$ be locally convex Hausdorff spaces. Then there is an isomorphism of vector spaces 
\begin{equation}
E \otimes F \simeq B(E_\sigma', F_\sigma')
\end{equation}
\end{proposition}

\begin{proof}
\cite{treves2006topological} Proposition 42.4
\end{proof}

\begin{definition}
Let $A \subset E$ and $A^0 \subset E'$ its polar. The \emph{bipolar} $A^{00}$ of $A$ is the polar of $A^0$ regarded as a subset $E \simeq (E_\sigma')'$
\end{definition}

\begin{proposition}\label{prop:bipolar thm}
Let $A \subset E$ be a subset. The bipolar $A^{00}$ is the closed convex balanced hull of $A$.
\end{proposition}

\begin{proof}
\cite{treves2006topological} Proposition 35.3
\end{proof}

\begin{definition}
$E$ is \emph{Montel} if $E$ is barrelled and if every closed bounded subset is compact.
\end{definition}

\begin{proposition}\label{prop:dual montel}
Suppose $E$ is Montel. Then every closed, bounded subset of $E^*$ is compact. Moreover, the strong and weak topologies coincide on bounded subsets of $E^*$.
\end{proposition}

\begin{proof}
\cite{treves2006topological} Proposition 34.6
\end{proof}

\begin{proposition}\label{prop:dual montel weak=strong}
Let $E$ be a Montel space. Then every weakly convergent sequence in $E'$ is strongly convergent.
\end{proposition}
\begin{proof}
\cite{treves2006topological} Corollary 34.6
\end{proof}

\begin{definition}\label{def: reflexive}
A locally convex Hausdorff vector space $E$ is \emph{reflexive} if $(E^*)^*$.
\end{definition}

\begin{proposition}\label{prop: reflexive transpose}
A linear map $u: E \raw F$ between reflexive spaces is continuous if and only if $u^*$ is continuous.
\end{proposition}
\begin{proof}
This follows from Proposition \ref{prop: transpose continuous}.
\end{proof}

\begin{proposition}\label{prop:montel is reflexive}
Let $E$ be a Montel space. Then $E$ is reflexive and $E^*$ is also Montel.
\end{proposition}
\begin{proof}
\cite{treves2006topological} Proposition 36.10
\end{proof}

\subsection{Topological Tensor Products}
Let $E,F$ be locally convex Hausdorff spaces. 
\begin{definition}
The $\pi$-topology on $E\otimes F$ is the finest locally convex topology such that the map
\begin{equation}
\begin{split}
   E\times F &\raw E \otimes F\\
   (x,y) &\mapsto x \otimes y
\end{split}
\end{equation}
is continuous. We will call the tensor product endowed with this topology the \emph{projective tensor product} and denote it by $E \otimes_\pi F$.
\end{definition}

Let $B(E,F)$ denote the vector space of continuous bilinear forms on $E\times F$ and $\mathscr B(E,F)$ the vector space of separately continuous bilinear forms on $E\times F$. 
By Proposition \ref{prop:tensor bilinear} we can regard $E \otimes F$ as a vector subspace of $\mathscr B(E_\sigma', F_\sigma')$. 

\begin{definition}
The $\varepsilon$-topology on $E \otimes F$ is the subspace topology when $\mathscr B(E_\sigma', F_\sigma')$ is equipped with the topology of uniform convergence on products of equicontinuous sets. This topological tensor product is called the \emph{injective tensor product} and will be denoted by $E\otimes_\varepsilon F$.
\end{definition}

\begin{notation}
We will denote the completions of these two topological tensor products by $E \widehat \otimes_\pi F$ and $E \widehat \otimes_\varepsilon F$.
\end{notation}

We now examine some properties of the projective tensor product. The first is a universal property.

\begin{proposition}\label{prop:UP}
Let $E,F,G$ be locally convex Hausdorff spaces and $E \times F \raw G$ a continuous bilinear map. Then there exists a unique continuous extension
\begin{equation}\label{eq:extension}
\begin{tikzcd}
   E \times F \arrow[r]\arrow[d] & G\\
   E \otimes_\pi F \arrow[ur, dashed, "\exists !"] &
\end{tikzcd}
\end{equation} 
\end{proposition}

\begin{proof}
\cite{treves2006topological} Proposition 43.4
\end{proof}

\begin{corollary}
The map in (\ref{eq:extension}) extends uniquely to a map $E \widehat\otimes_\pi F \raw G$.
\end{corollary}
\begin{proof}
Apply Proposition \ref{prop:completion UP}.
\end{proof}

We note that by definition of the injective tensor product,
the map 
\begin{equation}\label{eq:ExF to EeF}
\begin{split}
E\times F &\raw E \otimes_\varepsilon F\\
(x,y) &\mapsto x \otimes y
\end{split}
\end{equation}
is continuous.

\begin{corollary}\label{cor:ExF to EeFcomplete}
The map (\ref{eq:ExF to EeF}) extends to a continuous map
\begin{equation}\label{eq:EpiF to EeF}
E \widehat\otimes_\pi F \raw E \widehat\otimes_\varepsilon F.
\end{equation}
Thus the projective topology is finer than the injective topology.
\end{corollary}

\begin{lemma} Let $E_1 \xrightarrow{u} E_2, F_1\xrightarrow{v} F_2$ be continuous maps between locally convex Hausdorff spaces. Then $u\otimes v: E_1 \otimes_\tau F_1 \raw E_2 \otimes_\tau F_2$ is continuous for $\tau = \pi, \varepsilon$.
\end{lemma}
\begin{proof}
\cite{treves2006topological} Proposition 43.6
\end{proof}

\begin{corollary}\label{cor:completed tensor maps}
The previous maps extend uniquely to maps on completions
\begin{equation}
u \widehat\otimes_\pi v : E_1\widehat\otimes_\pi F_1 \raw E_2 \widehat\otimes_\pi F_2
\end{equation}
\begin{equation}
u \widehat\otimes_\varepsilon v : E_1 \widehat\otimes_\varepsilon F_1 \raw E_2 \widehat\otimes_\varepsilon F_2
\end{equation} 
\end{corollary}
\begin{proof}
Apply Proposition \ref{prop:completion UP}.
\end{proof}

\begin{proposition}\label{prop:injective and dense maps}
Let $u,v$ be as above. If $u,v$ are injective, then $u \widehat \otimes_\varepsilon v$ is injective. Likewise, if $u,v$ have dense images, then $u\widehat\otimes_\pi v$ has dense image.
\end{proposition} 
\begin{proof}
\cite{treves2006topological} Exercise 39.3 and Proposition 43.9
\end{proof}

We will also need the following description of elements of the projective tensor product of Fr\'echet spaces.

\begin{proposition}\label{prop:projective Frechet}
Let $E,F$ be Fr\'echet spaces. Then every $\theta \in E \widehat\otimes_\pi F$ can be expressed as an absolutely convergent series
\begin{equation}
   \theta = \sum_{n} \lambda_n x_n \otimes y_n
\end{equation}
where $\{\lambda_n\} \in \ell^1$ and $\{x_n\}, \{y_n\}$ are sequences converging to 0 in $E, F$.
\end{proposition}
\begin{proof}
\cite{treves2006topological} Theorem 45.1
\end{proof}

\subsection{Nuclear Maps}
Let $E,F$ be Banach spaces. There is a continuous map 
\begin{equation}
E^* \times F \raw \Hom(E,F)
\end{equation}
where the codomain is equipped with the topology induced by the operator norm. By Proposition \ref{prop:UP}, there is a unique continuous map $E^* \widehat\otimes_\pi F \raw \Hom(E,F)$ and we denote its image by $L^1(E,F) \subset \Hom(E,F)$.

\begin{definition}
Let $E,F$ be Banach spaces. A continuous map $u: E \raw F$ is \emph{nuclear} if $u \in L^1(E,F)$.
\end{definition}

We now describe nuclear maps when $E,F$ are general locally convex Hausdorff spaces.
Let $D \subset E$ be a disk and let $E_D := \text{span}(D) \subset E$ be the subspace spanned by $D$. There is a seminorm $\p_D: E_D \raw \R$ given by
\begin{equation}
   \p_D(x) := \inf_{\{\rho \in \R_{>0}\st x \in \rho U\}} \rho.
\end{equation}
Note that if $B$ is a bounded disk, $\p_B$ is a norm on $E_B$.

\begin{definition}
We say a subset $B\subset E$ is a \emph{Banach disk} if $B$ is a bounded disk and $(E_B, \p_B)$ is a Banach space.
\end{definition}

\begin{lemma}
Let $B \subset E$ be a complete, bounded disk. Then $B$ is a Banach disk.
\end{lemma}

\begin{proof}
\cite{treves2006topological} Lemma 36.1
\end{proof}

\begin{corollary}\label{cor:compact disk}
Let $B \subset E$ be a compact disk. Then $B$ is a Banach disk.
\end{corollary}

\begin{proof}
\cite{treves2006topological} Corollary 36.1
\end{proof}

\begin{definition}
A subset $U \subset E$ is a \emph{dual Banach disk} if it is a closed disk that is a neighborhood of $0 \in E$.
\end{definition}

A dual Banach disk $U$ is absorbing which implies $\p_U$ is a seminorm on all of $E$. We remark that $U$ is the closed unit semiball of $\p_U$.

\begin{construction}
Let $\p: E \raw \R$ be a seminorm and $V$ its closed unit semiball. There is a norm induced by $\p$ on the quotient $E/\ker(\p)$. We will use $E_\p$ and $E^V$ interchangeably to denote the Banach space completion $\overline{E/\ker(\p)}$.
\end{construction}
Applying the construction to $\p_U$ yields the Banach space $E^U$.

Let $\mathfrak U$ be the collection of all dual Banach disks on $E$ and $\mathfrak B$ the collection of all Banach disks on $F$. Let $U \in \mathfrak U$, $B \in \mathfrak B$ and suppose $u\in L^1(E^U, F_B)$. There are continuous maps $E \xrightarrow{q_U} E^U$ and $F_B \xrightarrow{i_B} F$ induced by the vector space quotient $E \raw E/\ker(\p_U)$ and inclusion $F_B \mono F$ respectively. Thus there is a map
\begin{equation}
\begin{split}
L^1(E^U, F_B) &\raw \Hom(E,F)\\
u \hspace{.5cm} &\mapsto \hspace{.2cm} i_B \circ u \circ q_U
\end{split}
\end{equation}
and we denote its image by $L^1_{U,B}(E,F)$.
Set
\begin{equation}
   L^1(E,F) := \bigcup_{U \in \mathfrak U, B \in \mathfrak B} L^1_{U,B}(E, F)
\end{equation}

\begin{definition}
A map $u: E\raw F$ is \emph{nuclear} if $u \in L^1(E,F)$.
\end{definition}

\begin{proposition}\label{prop:nuclear composition}
Let $f: G \raw E \xrightarrow{u} F \raw H$ be a composition of continuous maps between locally convex Hausdorff spaces and suppose $u$ is nuclear. Then $f$ is nuclear.
\end{proposition}

\begin{proof}
\cite{treves2006topological} Proposition 47.1
\end{proof}

\begin{proposition}\label{prop:nuclear map characterization}
Let $u: E \raw F$ be continuous. The following are equivalent:
\begin{enumerate}
   \item $u$ is nuclear
   \item $u$ is a composition
      \begin{equation}
      u: E \raw V \xrightarrow{v} W \raw F\nonumber
      \end{equation}
      with $V,W$ Banach spaces and $v$ nuclear
   \item $u$ can be expressed as a map
      \begin{equation}\label{eq:nuclear sum}
      u: x \mapsto \sum_k \lambda_k\langle x_k', x \rangle y_k
      \end{equation}
   where $\{x_k'\} \subset E'$ is an equicontinuous sequence, $\{y_k\} \subset B$ is a sequence contained in a  Banach disk $B \subset F$, $\{\lambda_k\} \subset \C$ is absolutely summable.
\end{enumerate}
\end{proposition}

\begin{proof}
\cite{treves2006topological} Proposition 47.2
\end{proof}

\begin{corollary}
$L^1(E,F) \subset \Hom(E,F)$ is a subspace.
\end{corollary}

We have the following refinement of Proposition \ref{prop:nuclear map characterization} for maps between Fr\'echet spaces. 

\begin{proposition}\label{prop:nuclear map between Frechet}
Let $u: E\raw F$ be a continuous map between Fr\'echet spaces. Then $u$ is nuclear if and only if $u$ has a representation (\ref{eq:nuclear sum}) with $\{x_k'\}, \{y_k\}$ bounded sequences in $E^*, F$.
\end{proposition}

\begin{proof}
\cite{treves2006topological} Proposition 47.2 Corollary 2
\end{proof}

We now record some properties of the transpose of a nuclear map. Recall that by Proposition \ref{prop: transpose continuous} the transpose of a continuous map is continuous.

\begin{proposition}\label{prop:nuclear transpose}
Let $E\xrightarrow{u} F$ be a nuclear map between two locally convex Hausdorff spaces. Then $u^*: F^* \raw E^*$ is nuclear.
(Here, $E^*$ denotes the dual vector space equipped with the strong topology.)
\end{proposition}
\begin{proof}
\cite{treves2006topological} Proposition 47.4
\end{proof}

\subsection{Nuclear Spaces}
Let $E$ be a locally convex Hausdorff space. Suppose $\p,\q$ are continuous seminorms on $E$ with $\q \geq \p$. Then the topology defined by $\q$ on $E$ is finer than the topology defined by $\p$ and $\ker(\q) \subset \ker(\p)$. Thus there is a continuous map $E/\ker(\q) \raw E/\ker(\p)$ which extends to a continuous map of completions $E_\q \raw E_\p$.

\begin{definition}
   A locally convex Hausdorff vector space $E$ is \emph{nuclear} if for every seminorm $\p: E \raw \R$ there exists a seminorm $\q \geq \p$ such that
   \begin{equation}
      E_\q \raw E_\p
   \end{equation}
   is a nuclear map of Banach spaces.
\end{definition}

\begin{proposition}\label{prop:nuclear space char}
Let $E$ be a locally convex Hausdorff space. The following are equivalent:
\begin{enumerate}
\item $E$ is nuclear
\item The map (\ref{eq:EpiF to EeF}) induces an isomorphism 
\begin{equation}\label{eq:nuclear iso}
E \widehat\otimes_\pi F \simeq E \widehat\otimes_\varepsilon F
\end{equation}
for every $F$ locally convex Hausdorff.
\item Every continuous map $E \raw B$ to a Banach space $B$ is nuclear. 
\end{enumerate}
\end{proposition}

\begin{proof}
\cite{treves2006topological} Theorem 50.1
\end{proof}

\begin{notation}
In light of the isomorphism (\ref{eq:nuclear iso}), we will denote $E \widehat\otimes F := E \widehat\otimes_\pi F\simeq E\widehat\otimes_\varepsilon F$ when either $E$ or $F$ is nuclear.
\end{notation}

\begin{proposition}
The following are true:
\begin{enumerate}
   \item $E$ is nuclear if and only if $\tilde E$ is nuclear.
   \item A linear subspace of a nuclear space is nuclear.
   \item Let $F \subset E$ be a closed subspace. Then $E/F$ is nuclear.
   \item Arbitrary colimits of nuclear spaces are nuclear.
   \item Countable limits of nuclear spaces are nuclear.
   \item If $E,F$ are nuclear, then $E \widehat\otimes F$ is nuclear.
\end{enumerate}
\end{proposition}

\begin{proof}
\cite{treves2006topological} Proposition 50.1
\end{proof}

\subsection{The categories \texorpdfstring{$\mathcal{NF}$}{NF} and \texorpdfstring{$\mathcal{NDF}$}{NDF}}

\begin{definition}
A locally convex Hausdorff space is \emph{dual Fr\'echet} if it is the strong dual of a Fr\'echet space.
\end{definition}

We will denote by $\mathcal F$, $\mathcal{NF}$, $\mathcal{NDF}$ the categories of Fr\'echet, nuclear Fr\'echet, and nuclear dual Fr\'echet spaces respectively.

\begin{proposition}\label{prop:N(D)F sym monoidal}
$\mathcal{NF}$ and $\mathcal{NDF}$ are symmetric monoidal categories
\end{proposition}
\begin{proof}
\cite{costello2011renormalization} Appendix 2
\end{proof}

\begin{proposition}\label{prop:N(D)F Montel}
Suppose $E \in \mathcal{NF}$ or $E \in \mathcal{NDF}$. Then $E$ is Montel.
\end{proposition}

\begin{corollary}\label{cor:NF reflexive}
Let $E \in \mathcal{NF}$. Then $E$ is reflexive.
\end{corollary}

\begin{proposition}\label{prop:dual NF}
Let $E \in \mathcal{F}$. Then $E \in \mathcal{NF}$ if and only if $E^* \in \mathcal{NDF}$.
\end{proposition}

\begin{proposition}
Taking strong duals gives an equivalence of symmetric monoidal categories
\begin{equation}
\mathcal{NF}^{op} \simeq \mathcal{NDF}
\end{equation}
\end{proposition}

\begin{proposition}\label{prop:Hom Isos}
Let $E,F \in \mathcal{NF}$. Then we have the following isomorphisms
\begin{equation}\label{eq:iso1}
   E \widehat\otimes F \simeq \Hom(E^*,F)
\end{equation}
\begin{equation}\label{eq:iso2}
   E^*\widehat\otimes F \simeq \Hom(E,F)
\end{equation}
\begin{equation}
   E^*\widehat\otimes F^* \simeq (E\widehat\otimes F)^* \simeq B(E,F)
\end{equation}
\end{proposition}
\begin{proposition}\label{prop:inverse limit nuclear}
The inverse limit of a countable inverse system of Fr\'echet spaces with nuclear maps is nuclear Fr\'echet.
\end{proposition}

\begin{proof}
Let $(E_i)_{i\in I}$ be the inverse system of nuclear maps between Fr\'echet spaces, i.e. $I$ is a countable directed poset where if $i \geq j$ there exists a nuclear map $f_{ij}: E_i \raw E_j$ of Fr\'echet spaces. The inverse limit $\hat E := \invlim_{i \in I} E_i$ is the vector space
\begin{equation}
\hat E = \{(x_i) \in \prod_{i\in  I} E_i \st x_j = f_{ij}(x_i) \text{ for all } i \geq j\}
\end{equation}
equipped with the topology defined by the seminorms given by composing
\begin{equation}
\hat \p: \hat E \xhookrightarrow{\iota} \prod_{i \in I} E_i \xrightarrow{\pi_k} E_k \xrightarrow{\p} \R
\end{equation}
where $\iota$ is the inclusion, $\pi_k$ is the canonical projection onto the $k$th factor and $\p$ is a seminorm defining the topology on $E_k$. Set $\sigma_k := \pi_k \circ \iota$ and $F_k := \overline{\sigma_k(\hat E)} \subset E_k$ be the closure of the image of $\sigma_k$. Since Fr\'echet spaces are closed under taking countable products and closed subspaces, $\hat E$ and $F_k$ are both Fr\'echet. Observe that for any $j \geq k$, $\sigma_k = f_{j,k} \circ \sigma_j: \hat E \raw F_k$ is a composition of a nuclear map with a continuous map and is therefore nuclear by Proposition \ref{prop:nuclear composition}. 

Let $V$ be a Banach space and let $u: \hat E \raw V$ be a continuous map. Then for all $\varepsilon > 0$, there exists $i \in I$ and $\p: E_i \raw \R$ such that $u(U_{\hat p}) \subset B_\varepsilon$, where $U_{\hat \p}$ is the unit semiball of $\hat \p$ and $B_{\varepsilon}$ is the ball of radius $\varepsilon$ in $V$. Thus $u$ admits a factorization
\begin{equation}
u = u_i \circ \sigma_i
\end{equation}
with $u_i: F_i \raw V$ continuous. Proposition \ref{prop:nuclear composition} implies $u$ is nuclear and therefore Proposition \ref{prop:nuclear space char} implies $\hat E$ is nuclear.
\end{proof}

\begin{corollary}\label{cor:direct limit nuclear}
The direct limit of a direct system of dual Fr\'echet spaces with nuclear maps is nuclear dual Fr\'echet.
\end{corollary}
\begin{proof}
Dualize and apply Proposition \ref{prop:inverse limit nuclear}.
\end{proof}

\begin{proposition}\label{prop:NDF2DF}
Let $u: \check E \raw \hat F$ be a continuous map between $\check E \in \mathcal{NDF}$ and $\hat F \in \mathcal{NF}$. Then $u$ is nuclear.
\end{proposition}
\begin{proof}
Under the isomorphism (\ref{eq:iso1}) of Proposition \ref{prop:Hom Isos}, we can view $u \in \check E^* \widehat\otimes \check F$, where $\check E^* \in \mathcal{NF}$ by Corollary \ref{cor:NF reflexive} and Proposition \ref{prop:dual NF}. By Proposition \ref{prop:projective Frechet}, $u$ can be expressed as an absolutely convergent series
\begin{equation}\label{eq:u series}
   u = \sum_n \lambda_n x'_n \otimes y_n
\end{equation}
with $\lambda_n \in \ell^1$ and $\{x_n'\},\{y_n\}$ sequences converging to 0 in $E', F$. In particular, $\{x_n'\}, \{y_n\}$ are bounded. Applying isomorphism (\ref{eq:iso1}) and Proposition \ref{prop:nuclear map between Frechet} shows that $u$ is nuclear.
\end{proof}


\section{Bundles of Topological Vector Spaces}\label{appendix bundles of tvs}

For thorough treatments of infinite dimensional differential geometry we refer the reader to \cite{Schmeding_2022} and \cite{kriegl1997convenient}.

We assume all topological vector spaces are locally convex and Hausdorff.

\begin{definition}
A continuous map $T: E \raw F$ between real topological vector spaces is \emph{$C^1$-differentiable} on an open subset $U \subset E$ if the limit
\begin{equation}
dT|_{x}(v) := \lim_{t \raw 0} \frac{T(x + tv) - T(x)}{t}
\end{equation}
exists for all $x \in U, v \in E$ and defines a continuous map
\begin{equation}
dT: U \times E \raw F
\end{equation}
\end{definition}

\begin{remark}
If $T$ is $C^1$-differentiable on $U \subset E$, then its differential $dT : U \times E \raw F$ is linear in the second factor.
\end{remark}

\begin{definition}
We inductively define smooth maps as follows.
Let $T: E \raw F$ be a continuous map between real topological vector spaces, $U \subset E$ be an open subset, and set $d^1T := dT$. We say $T$ is \emph{$C^{k+1}$-differentiable on $U$} if it is $C^k$-differentiable on $U$ and the map
\begin{equation}
d^{k+1}T: U \times E^{k+1} \raw F
\end{equation}
given by
\begin{equation}
d^{k+1}(x, v_1,..,v_{k+1}) 
= \lim_{t\raw 0}\frac{d^{k}T(x+tv_{k+1}, v_1,...,v_k) - d^kT(x, v_1,...,v_k)}{t}
\end{equation}
exists and is continuous. We say $T$ is \emph{$C^\infty$-differentiable on $U$} or \emph{smooth on $U$} if it is $C^k$-differentiable on $U$ for all $k$. 
\end{definition}

\begin{definition}
Let $T: E \raw F$ be smooth on an open subset $U \subset E$, where $E,F$ are complex topological vector spaces. We say $T$ is \emph{holomorphic on $U$} if $dT: U \times E \raw F$ is complex linear in the second factor.
\end{definition}

\begin{definition}
Let $S$ be a finite dimensional manifold of dimension $d$, $E$ and $F$ real topological vector spaces, and suppose 
\begin{equation}
f: S \times E \raw F
\end{equation}
is a continuous map. We say $f$ is \emph{smooth} if for every smooth embedding $\psi: \R^d \mono S$, the composition
\begin{equation}
\R^d \times E \xhookrightarrow{\psi} S \times E \xrightarrow{f} F
\end{equation}
is a smooth map of real topological vector spaces.
\end{definition}

\begin{definition}
Let $S$ be a finite dimensional complex manifold of dimension $d$. Suppose $E,F$ are complex topological vector spaces and let
\begin{equation}
g: S \times E \raw F
\end{equation}
be a continuous map. We say $g$ is \emph{holomorphic} if for every holomorphic embedding $\psi: \C^d \mono S$, the composition
\begin{equation}
\C^d \times E \xhookrightarrow{\psi} S \times E \xrightarrow{g} F
\end{equation}
is a holomorphic map of complex topological vector spaces.
\end{definition}

\begin{definition}
Let 
\begin{equation}
\pi : \mathscr E := \bigsqcup_{s \in S} E_s \raw S
\end{equation}
be a set of real topological vector spaces parametrized by a finite dimensional manifold $S$. Suppose there exists an open cover $S = \bigcup_{\alpha \in A} U_\alpha$, topological vector spaces $E_\alpha$ for $\alpha \in A$, and maps
\begin{equation}
\begin{tikzcd}
U_\alpha \arrow[rd, "pr_1"'] \times E_\alpha \arrow[rr, "\varphi_\alpha"] && \bigsqcup_{s \in U_\alpha} E_s \arrow[ld, "\pi"]\\
& U_\alpha &
\end{tikzcd}
\end{equation}
such that $\varphi_\alpha|_{\{s\}\times E}$ is an isomorphism of topological vector spaces. Set $U_{\alpha\beta} := U_\alpha \cap U_\beta$. We say $\pi: \mathscr E \raw S$ is a \emph{smooth vector bundle} if the compositions 
\begin{equation}\label{eq:g_ab}
g_{\alpha\beta}: U_{\alpha\beta} \times E_\alpha \xrightarrow{\varphi_\beta^{-1}\circ \varphi_\alpha} U_{\alpha\beta} \times E_\beta \xrightarrow{pr_2} E_\beta
\end{equation}
are smooth maps.
\end{definition}

\begin{definition}\label{def:holomorphic bundle}
Let
$\pi: \mathscr E \raw S$
be a smooth vector bundle over a finite dimensional complex manifold $S$, whose fibers are complex vector spaces. We say $\pi$ is \emph{holomorphic} if the maps (\ref{eq:g_ab}) are holomorphic.
\end{definition}

\bibliographystyle{alpha}
\bibliography{bibliography}

\end{document}